%% file: main.tex
\pdfoutput=1

\documentclass[12pt, reqno]{amsart}

\usepackage{main}

\hypersetup{
  pdftitle={Design Effect Ratios for Bayesian Survey Models: A Diagnostic Framework for Identifying Survey-Sensitive Parameters},
  pdfauthor={JoonHo Lee, Matthew R. Williams, Terrance D. Savitsky}
}
\usepackage{setspace}
\usepackage{booktabs}
\usepackage{array}
\usepackage{makecell}
\usepackage{tabularx}
\usepackage{multirow}
\usepackage{adjustbox}
\usepackage{pdflscape}
\usepackage{longtable}
\usepackage{needspace}

\graphicspath{{figures/}}
\usepackage{titletoc}

\usepackage[style=apa, backend=biber, natbib=true, backref=true]{biblatex}
\DefineBibliographyStrings{english}{%
  backrefpage = {},
  backrefpages = {}
}
\renewbibmacro*{pageref}{%
  \iflistundef{pageref}
    {}
    {\printtext[brackets]{\printlist[pageref][-\value{listtotal}]{pageref}}}}

\newgeometry{margin=1.25in}

\newcommand{\slotnum}[1]{#1}

\newcommand{\Vtarget}{V_{\mathrm{target}}}

\title{\Large Design Effect Ratios for Bayesian Survey Models: A Diagnostic Framework for Identifying Survey-Sensitive Parameters}

\author{JoonHo Lee}
\author{Matthew R.\ Williams}
\author{Terrance D.\ Savitsky}

\date{\footnotesize September 2026. \\[0.5em]
Lee: College of Education, The University of Alabama, Tuscaloosa, AL, USA. jlee296@ua.edu.
Williams: U.S.\ Bureau of Labor Statistics, Washington, DC, USA.
Savitsky: U.S.\ Bureau of Labor Statistics, Washington, DC, USA}

\numberwithin{equation}{section}

\renewcommand{\theHequation}{\thesection.\arabic{equation}}
\renewcommand{\theHfigure}{\thesection.\arabic{figure}}
\renewcommand{\theHtable}{\thesection.\arabic{table}}

\makeatletter
\renewcommand{\theHALG@line}{\thealgorithm.\arabic{ALG@line}}
\makeatother

\begin{document}


\begin{abstract}
Bayesian hierarchical models fitted to complex survey data often use
weighted pseudo-posteriors, followed by a correction to match a
design-based sandwich covariance.
Correcting every parameter can needlessly widen or spuriously narrow
credible intervals for parameters stabilized by shrinkage or identified
by between-group variation.
We propose the Design Effect Ratio (DER) to guide selective correction.
It is the ratio of each parameter's design-based sandwich variance under
a declared target to its model-based posterior variance.
For hierarchical Gaussian models, we derive exact finite-sample
decompositions under balance and common-design-effect assumptions.
Fixed-effect DERs scale with the design effect, attenuated by the
between-group share of identifying variation.
Random-effect DERs factor into design effect, shrinkage, and a
group-count term; a conservation identity links the two levels.
A general matrix formula covers arbitrary parameter blocks and
non-Gaussian likelihoods.
We evaluate the diagnostic in a simulation with informative two-stage
sampling and a re-analysis of the 2019 National Survey of Early Care
and Education, reporting both design-PSU and model-group variance targets.
These analyses distinguish survey-sensitive parameters from those
protected by shrinkage and show how selective correction avoids the
interval distortions caused by blanket rescaling.
The R package \texttt{svyder} implements the workflow.
\end{abstract}

\maketitle

\noindent\textbf{Keywords:} Bayesian hierarchical model; design effect; pseudo-posterior; sandwich variance; survey weighting

\pagestyle{plain}

\newpage


\input{section1}
\input{section2}
\input{section3}
\input{section4}
\input{section5_application}
\input{section6}


\bigskip

\paragraph{Software and reproducibility.}
All analyses were conducted in R~4.5.1 \citep{RCore2025} with posterior
sampling via Stan \citep{Carpenter2017} through the \texttt{cmdstanr}
interface.
The DER computation requires only base R and the \texttt{Matrix}
package.
The simulation study comprises \slotnum{4{,}800} replications
(\slotnum{24} design cells $\times$ 200 replications) with two-stage
sampling and exact inverse-probability weights.
The Stan models use a non-centered parameterization for hierarchical
logistic regression.
We screen every replication for convergence ($\widehat{R} \le 1.01$
on all parameters, zero divergent transitions) and report the outcomes
with the results.
The \texttt{svyder} R package, version~0.2.0
(\url{https://github.com/joonho112/svyder}; documentation at
\url{https://joonho112.github.io/svyder/}) implements the
compute-classify-correct workflow with the declared variance target as
an explicit argument (aggregation unit, strata, weight convention) and
native support for \texttt{brms}, \texttt{cmdstanr}, and
\texttt{rstanarm} backends.
The package complements the \texttt{csSampling} package
\citep{Hornby2023} that implements the underlying pseudo-posterior and
correction machinery.
Replication code and pre-computed results are available at \url{https://github.com/joonho112/svyder-replication}.

\paragraph{Data availability.}
The empirical results in this paper are based on data from the 2019
National Survey of Early Care and Education (NSECE). All key variables
used in the analysis are available only through the
NSECE Level 1 Restricted-Use Files \citep{HHSNSECE2024Restricted}
(ICPSR Study 38445, Version 3;
\url{https://doi.org/10.3886/ICPSR38445.v3}). Access requires an
application, a Restricted Data Use Agreement, and institutional review
board approval or exemption.

\paragraph{Conflicts of interest.}
The authors declare no conflicts of interest.

\paragraph{Disclaimer.}
The opinions expressed in this article are those of the authors and do
not represent the views of the U.S.\ Bureau of Labor Statistics.


\printbibliography[title={References}]


\newpage
\appendix
\setcounter{section}{0}
\renewcommand{\thesection}{S-\Alph{section}}
\numberwithin{equation}{section}
\numberwithin{figure}{section}
\numberwithin{table}{section}

\renewcommand{\theHequation}{Supp.\thesection.\arabic{equation}}
\renewcommand{\theHfigure}{Supp.\thesection.\arabic{figure}}
\renewcommand{\theHtable}{Supp.\thesection.\arabic{table}}
\renewcommand{\theHlemma}{Supp.\thesection.\arabic{lemma}}
\renewcommand{\theHtheorem}{Supp.\thesection.\arabic{theorem}}
\renewcommand{\theHcorollary}{Supp.\thesection.\arabic{corollary}}
\renewcommand{\theHproposition}{Supp.\thesection.\arabic{proposition}}
\renewcommand{\theHdefinition}{Supp.\thesection.\arabic{definition}}
\renewcommand{\theHremark}{Supp.\thesection.\arabic{remark}}
\renewcommand{\theHexample}{Supp.\thesection.\arabic{example}}
\renewcommand{\theHconjecture}{Supp.\thesection.\arabic{conjecture}}

\setcounter{lemma}{0}       \numberwithin{lemma}{section}
\setcounter{theorem}{0}     \numberwithin{theorem}{section}
\setcounter{corollary}{0}   \numberwithin{corollary}{section}
\setcounter{proposition}{0} \numberwithin{proposition}{section}
\setcounter{definition}{0}  \numberwithin{definition}{section}
\setcounter{remark}{0}      \numberwithin{remark}{section}
\setcounter{example}{0}     \numberwithin{example}{section}
\setcounter{conjecture}{0}  \numberwithin{conjecture}{section}

\begin{center}
{\Large\bfseries Online Supplement}\\[1em]
{\large Design Effect Ratios for Bayesian Survey Models:\\
A Diagnostic Framework for Identifying\\
Survey-Sensitive Parameters}\\[1.5em]
{\normalsize JoonHo Lee \quad Matthew R.\ Williams \quad Terrance D.\ Savitsky}
\end{center}

\vspace{2em}

\startcontents[appendices]
\printcontents[appendices]{}{1}{\textbf{Contents}\vskip1em\hrule\vskip1em}
\vskip1em\hrule\vskip2em

\input{osm_body}

\end{document}

%% file: section1.tex

\section{Introduction}
\label{sec:introduction}

Bayesian hierarchical models fitted with survey-weighted
pseudo-likelihoods are increasingly used for surveys with multistage
designs and unequal selection probabilities.
A model may contain dozens or hundreds of parameters, including
fixed-effect coefficients for policy variables, random intercepts for
geographic groups, and variance components.
Some credible intervals need adjustment for the sampling design;
others are already stabilized by the hierarchical prior.
No routine diagnostic indicates which intervals require correction.
This is a Bayesian version of the question posed by
\citet{Pfeffermann1993} about when survey weights matter for model-based
inference.
Correcting every parameter may seem conservative when there are only a
few, but across these parameter types the correction can be
conservative for some and harmful for others.

A widely used approach constructs a pseudo-posterior by exponentiating
likelihood contributions by survey weights
\citep{SavitskyToth2016, WilliamsSavitsky2020}.
Under the design and regularity conditions of \citet{SavitskyToth2016},
the pseudo-posterior contracts on the superpopulation parameters.
The weights distort the log-likelihood curvature, however, so the
pseudo-posterior covariance is generally miscalibrated for frequentist use.
\citet{WilliamsSavitsky2021} transform the MCMC draws to match a
design-based sandwich covariance \citep{Binder1983}; their empirical
evaluations report near-nominal frequentist coverage across a range of
designs.
Applied to the full parameter vector, the correction treats all
parameters uniformly and costs $O(d^3)$ for $d$ parameters.

\citet{Gelman2007} noted the tension between design-based and
model-based perspectives on weighting, and the treatment and scaling
of weights in hierarchical models remain debated
\citep{PfeffermannEtAl1998b, Asparouhov2006, RabeHeskethSkrondal2006,
Carle2009, Cai2013}.
The classical design effect of \citet{Kish1965} and its
estimand-specific generalizations \citep{Skinner1989} describe the
design and estimand, rather than individual parameters within a
hierarchical specification.
The corrections of \citet{WilliamsSavitsky2021},
\citet{WilliamsMcGuireSavitsky2025}, and
\citet{LeonNoveloSavitsky2023} adjust posterior inference but do not
explain why some parameters are design-sensitive while others are
protected.
None provides a parameter-specific diagnostic with an interpretable
decomposition to guide that choice.

We introduce the Design Effect Ratio ($\DER$), the ratio of each
parameter's design-based sandwich variance to its model-based
posterior variance.
We use it to identify a subset of parameters, $\mathcal{F}$, for joint
correction.
The correction of \citet{WilliamsSavitsky2021} specifies how to adjust
the posterior; the $\DER$ indicates when adjustment is needed, which
parameters require it, and why.

Because the sandwich variance is not unique, we define the $\DER$
relative to a \emph{declared variance target} constructed from the
model pseudo-likelihood.
We specify two choices often treated as implementation details: the
aggregation unit and the survey-weight scaling convention.
The aggregation unit can be the design's primary sampling units (PSUs)
or the model groups when the two differ.
Weights can be normalized globally to have unit mean or rescaled
within each group.
Both choices affect the target variance; weight scaling also changes
the fitted pseudo-posterior.
These choices have a long history in the multilevel weighting
literature \citep{PfeffermannEtAl1998b, RabeHeskethSkrondal2006}.
Our application shows that both change the diagnostic's results and
must be specified to interpret the $\DER$.

We derive finite-sample decompositions that link the classical design
effect to hierarchical shrinkage (\Cref{sec:framework}).
Under the model-group target and the stated model conditions,
coefficients identified from within-group variation inherit the design
effect nearly in full and require credible-interval adjustment.
Coefficients identified from between-group contrasts are shielded and
typically need no adjustment.
Random effects are shielded in proportion to their shrinkage,
connecting the diagnostic to small area estimation
\citep{EfronMorris1975, RaoMolina2015}.
In the balanced Gaussian case with a common design effect, a
conservation identity shows that shrinkage moves design sensitivity
from the group random effects to the grand mean rather than
eliminating it.
The decompositions express design sensitivity through quantities that
practitioners already use.

We also develop an easy-to-implement compute-classify-correct algorithm
for selective correction.
Correcting the selected block costs $O(|\mathcal{F}|^3)$, compared with
$O(d^3)$ for the full vector.
We evaluate the algorithm in \slotnum{4{,}800} replications of
hierarchical logistic regression across \slotnum{24} design cells.
The simulation draws two-stage samples from finite populations.
Inclusion probabilities depend on the random effects and the response,
and the weights are exact inverse inclusion probabilities.

We apply the diagnostic and selective correction to the 2019 National
Survey of Early Care and Education (NSECE), using both the design-PSU
and model-group variance targets.
The within-state poverty coefficient is flagged under every target
and weight convention we examine.
Random effects are flagged only under the design-PSU target.
With global unit-mean weights, \slotnum{28} of 54 parameters are
flagged in total under that target, compared with \slotnum{1} of 54
under the model-group target.
This contrast shows why the aggregation unit---PSU or state
grouping---must be declared.

The $\DER$ measures how much design-based variability survives the
hierarchical structure.
It therefore depends on the model and prior and does not replace the
Kish design effect as a summary of the design.
It also does not address identification of cluster-dependence
parameters from marginal inclusion weights
\citep{RaoVerretHidiroglou2013, YiRaoLi2016, Slud2026}.
The diagnostic presupposes a fitted pseudo-posterior and a valid
variance target; \Cref{sec:background} specifies the design settings
in which this is warranted.
A $\DER$ at or below one does not establish that an unflagged
parameter is correctly calibrated: the ratio compares variances and
does not diagnose bias.

A preliminary version of the $\DER$ appeared in the supplementary
materials of \citet{Lee2026}.
Here we develop it into a self-contained diagnostic framework.
The comparison follows the calibrated Bayes perspective
\citep{Little2004, Little2012}: we state the design-based target
explicitly and report which posterior variances would change under
that target and by how much.

\Cref{sec:background} establishes notation and constructs the sandwich
variance from its ``bread'' and ``meat'' components.
\Cref{sec:framework} develops the $\DER$ and its closed-form Gaussian
decompositions.
\Cref{sec:simulation,sec:application} present the simulation and NSECE
application, and \Cref{sec:discussion} discusses implications and
limitations.
Sections~S-A and~S-B of the Supporting Information give the conditions
and complete proofs, respectively, with each condition annotated
where it is used.

%% file: section2.tex

\section[Setup: Model, Design, and the Declared Variance Target]
{Setup: Model, Design, and\\the Declared Variance Target}
\label{sec:background}

We define the hierarchical model, its weighted pseudo-posterior, and
the variance target used to compute the $\DER$.
The target depends on how we aggregate scores and scale survey weights.
A valid $\DER$ also requires a consistent pseudo-posterior estimator;
we discuss this requirement below.

\subsection{Model and design vocabulary}
\label{sec:vocabulary}

The design and model may specify distinct levels of aggregation.
The \emph{sampling design} selects units within \emph{primary
sampling units} (PSUs), indexed by $c$, possibly nested within design
\emph{strata}, indexed by $h$.
The \emph{model} organizes observations into \emph{groups}, indexed
by $j = 1, \ldots, J$, that carry the random effects.
Model groups may coincide with design PSUs.
In our NSECE application, however, the groups are the 51 states,
whereas the design samples 415 PSUs within 30 strata.
Every model group appears in the sample by construction.
Units $i = 1, \ldots, n_j$ within group $j$ give
$N = \sum_{j} n_j$ observations, and we write a generic two-level
model
\begin{equation}\label{eq:hierarchical}
  y_{ij} \mid \theta_j, \bbeta
  \;\sim\; f\!\pr{y_{ij} \mid \bx_{ij}^{\t}\bbeta + \theta_j},
  \qquad
  \theta_j \mid \boldsymbol{\psi}
  \;\sim\; g\!\pr{\theta_j \mid \boldsymbol{\psi}},
\end{equation}
with $\bbeta \in \bR^{p}$ the fixed effects, $\theta_j$ a scalar
random effect, $\boldsymbol{\psi}$ the hyperparameters (e.g.,
$\sigma^{2}_{\theta}$), and
$\boldsymbol{\phi} = (\bbeta^{\t}, \btheta^{\t})^{\t}$ of dimension
$d = p + J$.  We use this model to derive conditions under which the
$\DER$ recommends adjusting credible intervals for each parameter class.

\subsection{Pseudo-posterior and the weight-scaling convention}
\label{sec:pseudo}

When inclusion probabilities depend on the response even after
conditioning on covariates, unweighted Bayesian inference is
generally inconsistent for superpopulation parameters
\citep{Pfeffermann1993, PfeffermannEtAl1998}.
\citet{SavitskyToth2016} exponentiate each likelihood contribution by
its survey weight:
\begin{equation}\label{eq:pseudo-posterior}
  \pi^{*}\!\pr{\boldsymbol{\phi}, \boldsymbol{\psi} \mid \by}
  \;\propto\;
  \prod_{j=1}^{J}\prod_{i=1}^{n_j}
    \bk{f\!\pr{y_{ij} \mid \bx_{ij}^{\t}\bbeta + \theta_j}}^{\tilde{w}_{ij}}
  \;\times\;
  \prod_{j=1}^{J} g\!\pr{\theta_j \mid \boldsymbol{\psi}}
  \;\times\;
  p(\bbeta)\, p(\boldsymbol{\psi}),
\end{equation}
where the weights enter only the likelihood.
Under their design and regularity conditions the pseudo-posterior
contracts on the superpopulation parameters, but the weights distort
the curvature of the log-pseudo-likelihood, so the MCMC covariance
$\SigMCMC$ is generally miscalibrated.

Scaling $\tilde w_{ij}$ changes the pseudo-posterior and can affect
group-level inference, as recognized in the multilevel weighting literature
\citep{PfeffermannEtAl1998b, Asparouhov2006, RabeHeskethSkrondal2006,
Carle2009}.
Our declared convention is \emph{global unit-mean} scaling,
$\tilde{w}_{ij} = N \, w_{ij} / \sum_{i'j'} w_{i'j'}$, which
preserves between-group weight variation; the main alternative
rescales within each group ($\sum_i \tilde w_{ij} = n_j$), which
removes it.
Fixed-effect inferences are typically stable across these conventions.
Group-level effective sample sizes and quantities that depend on
shrinkage can change.
\Cref{sec:application} reports the diagnostic under both conventions;
the weight-scaling convention is part of the declared target.

\subsection{The declared (sandwich) variance target}
\label{sec:target}

We compare each parameter's posterior variance with its diagonal entry
in a design-based sandwich covariance \citep{Huber1967, Binder1983}.
We evaluate this target at the posterior mean $\hat{\boldsymbol{\phi}}$:
\begin{equation}\label{eq:sandwich}
  \Vtarget
  \;=\;
  H^{-1}\, J_{\mathrm{c}} \, H^{-1}.
\end{equation}
Weight scaling enters both components; the aggregation unit determines
the clusters used to compute the second.

\paragraph{The bread.}
$H$ is the \emph{penalized (posterior) Hessian}: the
log-pseudo-likelihood curvature plus the random-effect prior
precision,
\begin{equation}\label{eq:bread}
  H
  \;=\;
  -\sum_{j,i}
    \tilde{w}_{ij}\,
    \nabla^{2}_{\boldsymbol{\phi}}
    \log f\!\pr{y_{ij} \mid \bx_{ij}^{\t}\bbeta + \theta_j}
    \Big|_{\hat{\boldsymbol{\phi}}}
  \;+\;
  \diag\!\pr{\mathbf{0}_{p},\, \tau, \ldots, \tau},
  \qquad \tau = 1/\sigma^{2}_{\theta},
\end{equation}
with no curvature contributed by the diffuse-$\bbeta$ prior.
In computation, $\tau$ is evaluated at the posterior plug-in
$\hat{\tau} = 1/\hat{\sigma}^{2}_{\theta}$.
The $\DER$ is a \emph{conditional, post-fit diagnostic}: it uses the
fitted $\hat{\sigma}^{2}_{\theta}$ without estimating that variance
component anew.
Uncertainty in $\hat{\sigma}^{2}_{\theta}$ is outside the scope of the
$\DER$ (\Cref{sec:limits}).

\begin{remark}[Why the posterior Hessian is the right bread]
\label{rem:bread}
The pseudo-likelihood Hessian alone cannot serve as the bread: for
any linear predictor containing an intercept and a full set of group
effects, the direction $(1, \mathbf{0}_{p-1}, -\mathbf{1}_J)$ leaves
every $\bx_{ij}^{\t}\bbeta + \theta_j$ unchanged, so the likelihood
Hessian is singular for every data set.
The prior precision on $\btheta$ removes exactly this flat direction.
We therefore use the posterior Hessian, which describes the curvature
of the objective sampled by MCMC, throughout this paper.
\end{remark}

\paragraph{The meat and its aggregation unit.}
$J_{\mathrm{c}}$ is the stratified, centered outer product of
cluster-level score totals with the standard finite-cluster
correction \citep{Binder1983}:
\begin{equation}\label{eq:meat}
  \begin{aligned}
  J_{\mathrm{c}}
  &\;=\;
  \sum_{h} \frac{C_h}{C_h - 1}
  \sum_{c \in h}
  \pr{\mathbf{t}_{hc} - \bar{\mathbf{t}}_{h}}
  \pr{\mathbf{t}_{hc} - \bar{\mathbf{t}}_{h}}^{\t},
  \\
  \mathbf{t}_{hc}
  &= \sum_{(i,j) \in c}
    \tilde{w}_{ij}\,
    \nabla_{\boldsymbol{\phi}}
    \log f\!\pr{y_{ij} \mid \bx_{ij}^{\t}\bbeta + \theta_j}
    \Big|_{\hat{\boldsymbol{\phi}}},
  \end{aligned}
\end{equation}
where $c$ runs over the clusters of the \emph{declared aggregation
unit}, $C_h$ counts clusters in stratum $h$, and centering matters
because the likelihood score totals do not vanish at the posterior
mean.
Aggregating at the design PSUs targets the design's sampling
variability; aggregating at the model groups targets variability in
group-level structure.
We use the model-group target to derive the closed-form $\DER$ results
in \Cref{sec:framework}.
When PSUs and model groups differ, we report both targets
(\Cref{sec:application}).
With a single stratum and uncentered totals, \eqref{eq:meat} reduces to
$\sum_j \mathbf{s}_j \mathbf{s}_j^{\t}$.
Since $\operatorname{rank}(J_{\mathrm{c}})$ is at most the cluster
count, $\Vtarget$ need not be full rank when $d$ exceeds it.
The diagonal entries used by the $\DER$ remain well defined.
Correction of a flagged sub-block requires that block to be positive
definite (\Cref{sec:algorithm}).

\paragraph{A fixed-dimensional target for the regression coefficients.}
For the fixed effects one can work directly with the $p \times p$
marginal target from the block (Schur complement) decomposition of
\eqref{eq:sandwich}: writing
$H = \begin{psmallmatrix} \mathbf{A} & \mathbf{G} \\
\mathbf{G}^{\t} & \mathbf{D} \end{psmallmatrix}$
with $\mathbf{D}$ the penalized random-effect block,
\begin{equation}\label{eq:vbeta}
  V_{\bbeta}
  \;=\;
  S^{-1} M_{\bbeta}\, S^{-1},
  \qquad
  S = \mathbf{A} - \mathbf{G}\,\mathbf{D}^{-1}\mathbf{G}^{\t},
\end{equation}
where
$M_{\bbeta} = \begin{bsmallmatrix} \mathbf{I}_p &
-\mathbf{G}\mathbf{D}^{-1} \end{bsmallmatrix} J_{\mathrm{c}}
\begin{bsmallmatrix} \mathbf{I}_p &
-\mathbf{G}\mathbf{D}^{-1} \end{bsmallmatrix}^{\t}$
is the correspondingly profiled meat block.
$V_{\bbeta}$ has dimension $p \times p$ regardless of $J$.
Under the regularity and target-validity conditions in Section~S-A of
the Supporting Information, the standard Taylor/sandwich argument
applies to $V_{\bbeta}$.
By the block-inverse identity, its entries equal the $\bbeta$-entries
of the full sandwich.
The Schur complement incorporates the $\btheta$-prior into $S$
through conditioning, without approximation.

\subsection{The Williams--Savitsky correction}
\label{sec:cholesky}

\citet{WilliamsSavitsky2021} calibrate the pseudo-posterior by
mapping each draw through
$\boldsymbol{\phi}^{*(s)} = \hat{\boldsymbol{\phi}} +
\mathbf{R}_2 \mathbf{R}_1^{-1}(\boldsymbol{\phi}^{(s)} -
\hat{\boldsymbol{\phi}})$, with $\mathbf{R}_1, \mathbf{R}_2$ Cholesky
factors of $\SigMCMC$ and of the sandwich target.
The corrected draws have the target covariance by construction.
Their empirical evaluations report near-nominal coverage across a
range of designs.
Applied to the full vector, the correction costs $O(d^3)$ and treats
all parameters uniformly.
The $\DER$ introduced in Section~\ref{sec:der-definition} indicates
which parameters need correction and which intervals it may distort.

\subsection{Design effect and shrinkage: the two ingredients}
\label{sec:ingredients}

The Kish design effect \citep{Kish1965, SarndalEtAl1992} summarizes
variance inflation from unequal weighting:
$\DEFF = N \sum w_{ij}^2 / (\sum w_{ij})^2$, computed globally or
within group $j$ ($\DEFF_j$). \citet{Skinner1989} generalizes design
effects to arbitrary estimands.
The hierarchical shrinkage factor governs how strongly a group's
estimate relies on its own data: in the balanced Gaussian case,
\begin{equation}\label{eq:shrinkage}
  B_j
  \;=\;
  \frac{\sigma^{2}_{\theta}}
       {\sigma^{2}_{\theta} + \sigma^{2}_{e} / n_j},
\end{equation}
the reliability ratio of empirical Bayes and small area estimation
\citep{EfronMorris1975, FayHerriot1979, RaoMolina2015}:
$B_j \to 1$ for data-rich groups (little shrinkage), $B_j \to 0$ for
data-poor ones.
Strong shrinkage shields the group effect from design-induced
variance inflation, while re-exposing quantities such as the grand
mean that absorb the between-group information.
The $\DER$ formalizes this interaction.

%% file: section3.tex

\section{The Design Effect Ratio Framework}
\label{sec:framework}

We define the $\DER$ and use its decompositions to explain which
parameters need correction under the declared design-based variance
target and which may be harmed by it.

\subsection{Definition and basic properties}
\label{sec:der-definition}

\begin{definition}[Design Effect Ratio]\label{def:der}
For the $k$th component of
$\boldsymbol{\phi} = (\bbeta^{\t}, \btheta^{\t})^{\t}$, the
\emph{Design Effect Ratio} with respect to the declared variance
target \eqref{eq:sandwich} is
\begin{equation}\label{eq:der}
  \DER_k
  \;=\;
  \frac{[\Vtarget]_{kk}}{[\SigMCMC]_{kk}},
\end{equation}
where $\SigMCMC$ is the posterior covariance under the
pseudo-posterior~\eqref{eq:pseudo-posterior} and $k$ denotes scalar parameter $\boldsymbol{\phi}_{k} \in \boldsymbol{\phi}$.
\end{definition}

$\DER_k \approx 1$ means the posterior variance matches the target,
so no adjustment is needed.
$\DER_k > 1$ means the posterior understates the target variance, so
correction widens the interval.
$\DER_k < 1$ means the posterior is
wider than the target, so correction would narrow it,
reducing the influence of the prior shape.
The ratio also summarizes the intensity of the joint correction in
\Cref{sec:cholesky}: $\DER_k = ([\mathbf{R}_2\mathbf{R}_1^{-1}]_{kk})^2$
when off-diagonal entries are negligible.

\begin{remark}[Population and estimated DER]\label{rem:estimated}
\Cref{def:der} gives the plug-in diagnostic computed from one sample.
The results below describe its population counterpart, in which the
meat estimates the design covariance of the score totals.
The distinction matters most for random effects under the model-group
target: the $\theta_j$ meat entry is then built from one centered
score-total contrast per group (the group total's deviation from the
stratum mean of totals).
This entry is design-unbiased for its variance contribution but uses
effectively one observation.
Estimated per-group $\DER$s therefore carry substantial noise even
when the population formulas are exact.
In the simulation, wide random-effect margins nevertheless make
classification stable (\Cref{sec:simulation}).
This caution applies to per-group diagnostics.
\end{remark}

\subsection{Decomposition results}
\label{sec:decomposition}

We decompose the $\DER$ into the classical design effect and the
shrinkage factor, so that correction can target the design effect.
The decompositions are finite-sample algebraic identities for the
model class below.
They require no asymptotics in $J$ or $n_j$.
The target's frequentist interpretation comes from published
pseudo-posterior theory
(\citealp{SavitskyToth2016, WilliamsSavitsky2021};
\Cref{sec:target}).

\paragraph{Conditions.}
All results in this subsection assume:
\begin{enumerate}[nosep, label=(C\arabic*)]
  \item \emph{Model class.} The conjugate Gaussian hierarchical model
    $y_{ij} \mid \theta_j, \bbeta \sim \Norm(\bx_{ij}^{\t}\bbeta +
    \theta_j, \sigma^2_e)$ with $\theta_j \sim \Norm(0,
    \sigma^2_\theta)$ and the intercept included in $\bx_{ij}$ (its
    coefficient, the grand mean, is denoted $\mu$ where convenient);
    balanced groups of common size $n$; flat prior on $\bbeta$;
    known variances; covariates not collinear.
  \item \emph{Weights.} Strictly positive weights under the declared
    global unit-mean convention (\Cref{sec:pseudo}), with finite
    within-group Kish design effects $\DEFF_j$.
  \item \emph{Target.} The declared target of \Cref{sec:target} with
    the model groups as the aggregation unit.
\end{enumerate}
Balance in (C1) makes the shrinkage factor common across groups
($B_j \equiv B$), but does not make the within-group design effects
common.
Section~S-A of the Supporting Information restates these conditions in
survey terms and gives their published sources; Section~S-B identifies
where each condition is used in the proofs.
We start with the balanced conjugate case because \eqref{eq:der} is
available in closed form.
The general formula \eqref{eq:der} requires none of (C1)--(C3) beyond
a well-defined target and is used in all computations
(\Cref{sec:algorithm}).

\Needspace{12\baselineskip}
\begin{theorem}[Fixed-Effect DER]\label{thm:fixed-effect}
Under (C1)--(C3) and common within-group design effects
($\DEFF_j \equiv \DEFF$), for the $k$th fixed-effect coefficient
$\beta_k$,
\begin{equation}\label{eq:der-fixed}
  \DER_{\beta_k}
  \;=\;
  \DEFF \cdot (1 - R_k),
\end{equation}
where $R_k \in [0, B]$ measures the fraction of identifying variation
for covariate~$k$ that comes from between-group differences, and
$B = \sigma^2_\theta / (\sigma^2_\theta + \sigma^2_e / n)$ is the
common shrinkage factor~\eqref{eq:shrinkage}.  In particular:
\begin{enumerate}[nosep, label=(\roman*)]
  \item When the covariate varies only within groups
    ($R_k = 0$): $\DER_{\beta_k} = \DEFF$.
  \item When the covariate varies only between groups
    ($R_k = B$): $\DER_{\beta_k} = \DEFF\,(1 - B)$.
\end{enumerate}
Without the common-$\DEFF$ hypothesis, the fixed-effect $\DER$ is
given by a general matrix expression (Section~S-B.1 of the Supporting Information) of
which \eqref{eq:der-fixed} is the common-$\DEFF$ specialization.
\end{theorem}

\begin{proof}[Proof sketch]
Partition $H$ conformably with $(\bbeta^{\t}, \btheta^{\t})^{\t}$.
The marginal precision for $\bbeta$ is the Schur complement
$\mathbf{S}_\beta = \mathbf{H}_{\beta\beta} -
\mathbf{H}_{\beta\theta}\,\mathbf{H}_{\theta\theta}^{-1}\,
\mathbf{H}_{\theta\beta}$; eliminating $\btheta$ introduces factors
of $(1-B)$ that attenuate between-group variation, while the same
elimination applied to the meat produces factors of $(1-B)^2$.
The mismatch generates \eqref{eq:der-fixed}.
Full proof with step-level condition annotations: Section~S-B.1 of the Supporting Information.
\end{proof}

Under these conditions, $\DER_{\beta_{k}}$ depends on whether the
covariate varies within or between groups.
Within-group covariates are fully \emph{exposed}: the prior on
$\btheta$ cannot absorb design-induced correlation among units in the
same group.
Between-group covariates are \emph{shielded} by the factor $(1-B)$.
Coefficients of covariates that vary mostly within groups are the
primary correction candidates.

\Needspace{12\baselineskip}
\begin{theorem}[Random-Effect DER]\label{thm:random-effect}
Under (C1)--(C3) and common within-group design effects
($\DEFF_j \equiv \DEFF$), with the common shrinkage factor $B$ implied
by balance:
\begin{enumerate}[nosep, label=(\alph*)]
  \item \textbf{Conditional on the grand mean:}
    \begin{equation}\label{eq:der-re-cond}
      \DER_\theta^{\mathrm{cond}}
      \;=\;
      B \cdot \DEFF.
    \end{equation}
  \item \textbf{Marginal (finite $J$):}
    \begin{equation}\label{eq:der-re-exact}
      \DER_\theta
      \;=\;
      B \cdot \DEFF \cdot \kappa(J),
      \qquad
      \kappa(J) = 1 - \frac{1}{J(1 - B) + B}.
    \end{equation}
\end{enumerate}
\end{theorem}

\begin{proof}[Proof sketch]
Conditioning on $\mu$, the posterior precision of $\theta_j$ is
$a + \tau$ with $a = n/\sigma^2_e$ and $\tau = 1/\sigma^2_\theta$,
while the target variance inflates the data information by the
design effect; the ratio yields \eqref{eq:der-re-cond}.
The finite-$J$ factor arises because the marginal posterior of
$\theta_j$ integrates over $\mu$, coupling groups through the grand
mean; $\kappa(J) \to 1$ as $J \to \infty$.
Full proof: Section~S-B.2 of the Supporting Information.
\end{proof}

$B$ is the reliability ratio of the Fay--Herriot model
\citep{FayHerriot1979, RaoMolina2015}.
When $B$ is small, the estimate borrows heavily from the global mean
and is barely affected by the design ($\DER_\theta \approx 0$).
When $B$ is near one, the group's own data dominate and the conditional design
effect passes through ($\DER_\theta^{\mathrm{cond}} \approx \DEFF$).
$\kappa(J) \approx 1$ for $J \ge 20$ when $B$ is not close to
one.
$\DER_\theta$ has an inverted-U shape as a function of $B$, with a
peak at $B^* = \sqrt{J}/(\sqrt{J}+1)$.
$\DER_\theta$ is strictly less than $1$ under simple random sampling.
Section~S-C of the Supporting Information gives details (see
Corollary~S-C.1).

\begin{corollary}[Conservation identity]\label{cor:conservation}
Under the hypotheses of
\Cref{thm:fixed-effect,thm:random-effect},
\begin{equation}\label{eq:conservation}
  \DER_\mu + \DER_\theta^{\mathrm{cond}}
  \;=\;
  \DEFF.
\end{equation}
\end{corollary}

This follows from \Cref{thm:fixed-effect} with $R_k = B$ and
\Cref{thm:random-effect}(a).
Within this model class, the prior moves design sensitivity from the
group random effects to the grand mean; the total is unchanged.
Whether an analogous identity holds beyond the balanced Gaussian
class is open (\Cref{sec:discussion}).

\begin{remark}[Heterogeneous design effects]\label{rem:heterogeneous}
When the $\DEFF_j$ differ across groups, the random-effect $\DER$
depends on the whole ensemble $\{(B_k, \DEFF_k)\}$ through the
grand mean.
The naive per-group product $B_j \DEFF_j \kappa_j(J)$ is therefore
not exact: with $J = 3$, common
$B = 0.5$, and $\DEFF = (1, 9, 1)$, the exact $\DER$ of the first
group effect is $0.583$ while the naive product gives $0.25$.
Section~S-B.3 of the Supporting Information gives the ensemble-exact
expression, verified algebraically and numerically for unbalanced
group sizes and heterogeneous design effects.
The general formula \eqref{eq:der} requires no closed form.

For non-Gaussian likelihoods and unbalanced multi-covariate models,
the closed forms explain the mechanism; we compute the diagnostic
from the declared target using the general formula \eqref{eq:der}.
\end{remark}

\subsection{Limits of the diagnostic: hyperparameters}
\label{sec:limits}

The $\DER$ excludes hyperparameters such as $\sigma^2_\theta$.
The target uses data-level scores, but the score for
$\sigma^2_\theta$ has no direct data contribution: this hyperparameter
enters only through the unweighted group-level density.
More generally, the more a quantity relies on the prior rather than
the weighted data, the less design adjustment it needs and the less
informative the $\DER$ is about it.
Hyperparameter uncertainty under informative designs requires
separate methods \citep{SavitskyWilliams2022, WilliamsMcGuireSavitsky2025},
outside the scope of the $\DER$ framework.

\Needspace{24\baselineskip}
\subsection{The compute-classify-correct algorithm}
\label{sec:algorithm}

Algorithm~\ref{alg:ccc} uses $\DER_{k} > c_{0}$ to select parameter
$k$ for correction to the target variance.
\begin{algorithm}[H]
\caption{Compute--Classify--Correct (CCC)}
\label{alg:ccc}
\begin{algorithmic}[1]
\Require MCMC draws $\{\boldsymbol{\phi}^{(s)}\}_{s=1}^{T}$;
  declared-target matrices $H$, $J_{\mathrm{c}}$; threshold~$c_0$
  (default $c_0 = 1.2$)
\Ensure Selectively corrected posterior draws
\Statex
\Statex \textbf{Step 1: Compute}
\State $\Vtarget \gets H^{-1}\,J_{\mathrm{c}}\,H^{-1}$
\State $\SigMCMC \gets \widehat{\Var}\!\pr{\boldsymbol{\phi}^{(1)},
  \ldots, \boldsymbol{\phi}^{(T)}}$
\State $\DER_k \gets [\Vtarget]_{kk}\,/\,[\SigMCMC]_{kk}$
  \quad for each parameter~$k$
\Statex
\Statex \textbf{Step 2: Classify}
\State $\mathcal{F} \gets \{k : \DER_k > c_0\}$
  \Comment{Flag design-sensitive parameters}
\Statex
\Statex \textbf{Step 3: Correct}
\If{$\mathcal{F} \neq \varnothing$}
  \State Extract $\Vtarget[\mathcal{F},\mathcal{F}]$ and $\SigMCMC[\mathcal{F},\mathcal{F}]$
  \State $\mathbf{R}_1 \gets \mathrm{Chol}(\SigMCMC[\mathcal{F},\mathcal{F}])$;
    \quad $\mathbf{R}_2 \gets \mathrm{Chol}(\Vtarget[\mathcal{F},\mathcal{F}])$
  \State $\boldsymbol{\phi}_{\mathcal{F}}^{*(s)} \gets
    \hat{\boldsymbol{\phi}}_{\mathcal{F}} +
    \mathbf{R}_2\,\mathbf{R}_1^{-1}
    (\boldsymbol{\phi}_{\mathcal{F}}^{(s)} - \hat{\boldsymbol{\phi}}_{\mathcal{F}})$
    \quad for $s = 1, \ldots, T$
\EndIf
\State \Return corrected draws
  $\{\boldsymbol{\phi}^{*(s)}\}$ with unflagged parameters
  unchanged
  \Comment{Cost: $O(|\mathcal{F}|^3)$ vs.\ $O(d^3)$}
\end{algorithmic}
\end{algorithm}
\Needspace{2\baselineskip}
Step~3 applies a joint Cholesky transformation to the flagged block.
It matches the marginal variances and covariance structure of the
target on $\mathcal{F}$, reduces to a scalar rescaling when
$|\mathcal{F}| = 1$, and costs
$O(|\mathcal{F}|^3)$.
It requires the flagged sub-block
$\Vtarget[\mathcal{F},\mathcal{F}]$ to be positive definite.
Because
$\operatorname{rank}(J_{\mathrm{c}})$ is bounded by the cluster
count, this requirement can fail when $|\mathcal{F}|$ approaches
that count.
In that case, the software substitutes the nearest
positive-definite matrix \citep{Higham2002} and labels the correction
approximate.
The software leaves unflagged parameter draws unchanged.
With the conservative default $c_0 = 1.2$, we correct only parameters
whose target variance exceeds their posterior variance by more than
20\%.
\Cref{sec:simulation} and the application examine sensitivity to $c_0$.
Selective correction targets marginal calibration of the
flagged parameters.
Joint credible regions spanning flagged and unflagged parameters
require the full correction of \citet{WilliamsSavitsky2021}; the
$\DER$ helps interpret that correction.

\subsection{Properties and interpretation}
\label{sec:properties}

\paragraph{Covariate-specificity.}
The $\DER$ of a fixed effect depends on the within/between variation
of its covariate, not the outcome's marginal distribution.
Two covariates in one model can lie at opposite ends of the
$[\DEFF(1-B), \DEFF]$ band.

\paragraph{Shrinkage--design tradeoff.}
\Cref{thm:random-effect} and \eqref{eq:conservation} show that, under
the stated conditions, the parameters that benefit most from
shrinkage are least affected by the design.
This explains why blanket correction harms random effects under
the model-group target.

\begin{remark}[Derived quantities]\label{rem:reparam}
The $\DER$ is not reparameterization-invariant.
For a functional $g(\boldsymbol{\phi})$, the delta method gives
\begin{equation}\label{eq:der-delta}
  \DER(g)
  \;=\;
  \frac{\nabla g^{\t}\, \Vtarget\, \nabla g}
       {\nabla g^{\t}\, \SigMCMC\, \nabla g},
\end{equation}
with $\nabla g$ at $\hat{\boldsymbol{\phi}}$, giving a
fixed-dimensional diagnostic for area means, marginal effects, and
other scalar summaries, computed from the same two matrices.
\end{remark}

\paragraph{Effective sample size.}
$N_{\mathrm{eff}}^{\mathrm{design}}(\beta_k) = N / \DER_{\beta_k}$:
for within-group covariates this recovers the Kish effective sample
size $N/\DEFF$; for between-group covariates it exceeds the Kish
value, reflecting the prior's contribution.

\subsection{Four-tier classification}
\label{sec:classification}

We classify parameters by their exposure to the survey design, using
each parameter's $\DER_{k}$ to assess the need for variance correction.
\Cref{tab:four-tier} presents the four tiers (I-a, I-b, II, and III)
used throughout the paper.
Tier~I-a parameters inherit the full design effect and are the
primary correction candidates; Tier~I-b parameters are attenuated by
$(1-B)$, typically below threshold even when $\DEFF$ is large.
\Needspace{5\baselineskip}
Tier~II diagnostics depend on the declared unit: well below one under
the model-group target, but not necessarily under a design-PSU target
that captures within-group design variability
(\Cref{sec:application}).
Tier~III falls outside the sandwich construction.
The simulation examines the framework under two-stage informative
sampling.
\Cref{fig:der-curves} shows how the closed-form DER decompositions
vary with the shrinkage factor.

\begin{table}[H]
  \centering
  \caption{Four-tier classification of parameters by design
    sensitivity.  The ``DER (balanced)'' column gives the closed-form
    expression under (C1)--(C3) with common within-group design
    effects (\Cref{thm:fixed-effect,thm:random-effect}); Tier~III is
    excluded from the diagnostic by construction (\Cref{sec:limits}).}
  \label{tab:four-tier}
  \small
  \begin{tabularx}{\textwidth}{@{}l >{\raggedright}p{3.2cm} >{\raggedright}p{2.8cm} >{\centering}p{2.2cm} >{\raggedright\arraybackslash}X@{}}
    \toprule
    \textbf{Tier} & \textbf{Parameter Type}
      & \textbf{Info.\ Source}
      & \textbf{DER (Bal.)}
      & \textbf{Action} \\
    \midrule
    I-a & Fixed effects (within-group)
      & Within-group variation
      & $\DEFF$
      & Correct when flagged \\[3pt]
    I-b & Fixed effects (between-group)
      & Between-group variation
      & $\DEFF\,(1 - B)$
      & Monitor; often $< 1$ \\[3pt]
    II & Random effects ($\theta_j$)
      & Both (prior $+$ data)
      & $B \cdot \DEFF \cdot \kappa(J)$\textsuperscript{b}
      & Monitor; correction rarely needed under the model-group
        target \\[3pt]
    III & Hyperparameters ($\boldsymbol{\psi}$)
      & Between-group ($J$ groups)
      & N/A
      & Outside the diagnostic; profile or multi-level methods \\
    \bottomrule
  \end{tabularx}
  \vspace{2pt}
  {\footnotesize \textsuperscript{b}Common-$\DEFF$ balanced form;
    with heterogeneous $\DEFF_j$ the ensemble-exact expression in
    Section~S-B.3 of the Supporting Information applies.  $\kappa(J) \approx 1$ for
    $J \geq 20$ unless $B$ is close to one.\par}
\end{table}

\begin{figure}[H]
  \centering
  \includegraphics[width=\textwidth]{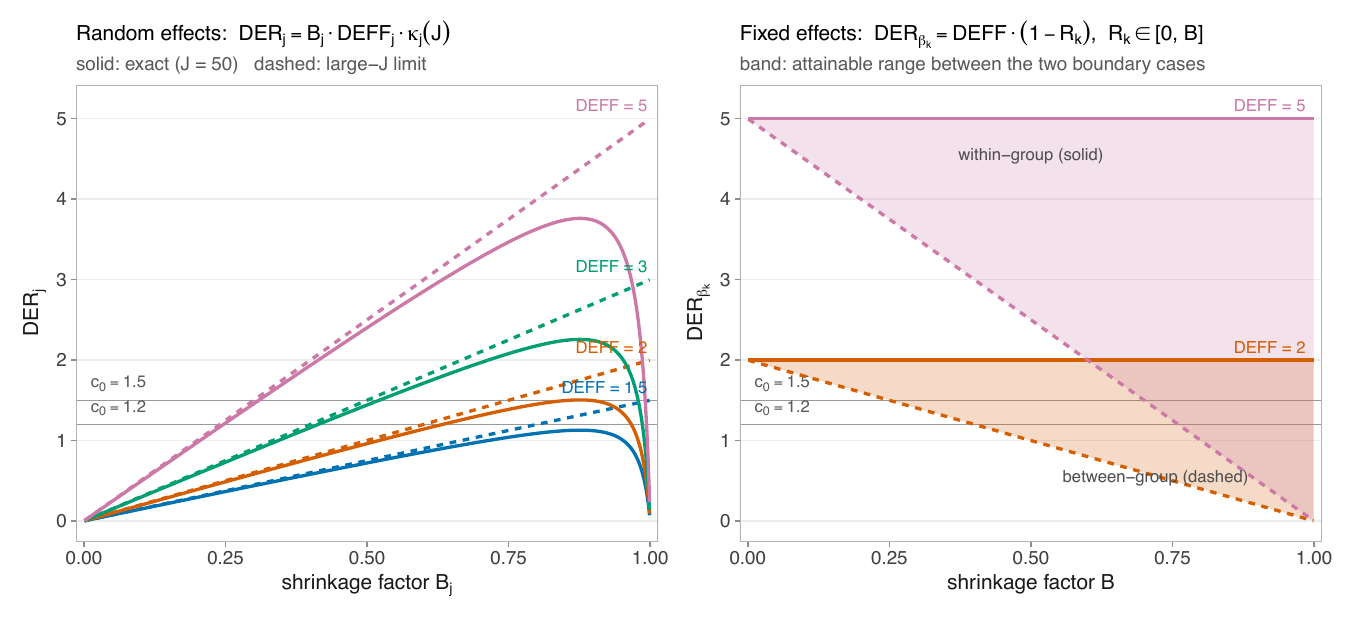}
  \caption{DER decomposition as a function of the shrinkage
    factor~$B$, using the closed forms of
    \Cref{thm:fixed-effect,thm:random-effect} under a common
    within-group design effect.
    Left: random-effect DER (\Cref{thm:random-effect}) has an
    inverted-U with peak at $B^* = \sqrt{J}/(\sqrt{J}+1)$; dashed
    lines show the large-$J$ limit $B \cdot \DEFF$.
    Right: fixed-effect DER (\Cref{thm:fixed-effect}) spans
    $[\DEFF(1-B),\, \DEFF]$, with within-group covariates at the
    upper boundary (solid) and between-group covariates at the lower
    (dashed).
    Horizontal gray lines mark candidate thresholds
    $c_0 \in \{1.2, 1.5\}$.}
  \label{fig:der-curves}
\end{figure}

\subsection{Scope: variance calibration, not identification}
\label{sec:scope}

We first need to establish whether the available design information
identifies the model's cluster-dependence parameters.
First-order (marginal) weights and sampled clusters are generally
insufficient: consistent estimation of variance components in
two-level models requires higher-order inclusion probabilities or
additional design structure
\citep{RaoVerretHidiroglou2013, YiRaoLi2016}.  \citet{Slud2026}
shows that in some small-cluster settings no consistent estimator of
cluster-dependence parameters exists based on first-order
probabilities alone.
The $\DER$ compares the fitted pseudo-posterior with the declared
target; it does not resolve identification.
Two settings support the consistency premise in our analyses.
First, when all model groups are observed rather than
sampled (as with the 51 states in \Cref{sec:application}),
group-level quantities are identified from within-group data, and
the impossibility results for sampled clusters do not bind.
Second, for sampled-group designs,
\citet{SavitskyWilliamsSrivastava2020} analyze consistency of the
weighted estimator under a within-group balance condition and
provide supporting simulation evidence.
Where neither holds (informative within-group selection with small
$n_j$), the group-level diagnostic inherits this identification
limitation.
We include small-$n_j$ simulation cells to examine this setting.

%% file: section4.tex

\section{Simulation Study}
\label{sec:simulation}

We test three predictions from the decompositions.
Within-group fixed effects should inherit the realized design effect,
while between-group coefficients and, under the model-group target,
random effects are shielded.
Selective correction should recover coverage where needed without
affecting the parameters it leaves uncorrected.
When design PSUs differ from model groups, the two targets should
measurably differ in the predicted direction.
The simulation specifies finite populations, inclusion probabilities,
and realized samples, with weights determined by the sampling mechanism.

\subsection{Design}
\label{sec:sim-design}

\paragraph{Populations and sampling mechanism.}
In each replication, we generate a finite population from the
hierarchical logistic model
\begin{equation}\label{eq:sim-model}
  y \mid \theta_j
  \;\sim\; \Bern\!\bigl(\expit(\beta_0 + \beta_1 x^{(\mathrm{W})}
  + \beta_2 x^{(\mathrm{B})}_j + \theta_j)\bigr),
  \qquad
  \theta_j \;\sim\; \Norm(0, \sigma^2_\theta),
\end{equation}
with $(\beta_0, \beta_1, \beta_2) = (-0.5, 0.5, 0.3)$ and
$\sigma^2_\theta$ set by a latent-scale intraclass correlation of
$0.15$. We then draw a two-stage sample with exact inverse-probability
weights under one of two structural arms:
\begin{itemize}[nosep]
  \item \emph{Sampled groups} ($c = j$): stage one samples $J$ of
    $M = 200$ population groups by randomized systematic PPS with
    size measures $\propto \exp(\gamma_{\mathrm{grp}} \theta_g +
    \text{noise})$; stage two Poisson-samples units with
    probabilities $\propto \exp(\gamma_{\mathrm{unit}} y +
    \text{noise})$, so that selection operates through the random
    effects and the response.  The design PSU is the model group.
  \item \emph{PSUs within groups} ($c \neq j$): all $J$ groups are
    observed, as in the NSECE. Within each group, we sample 4 of 12
    population PSUs by PPS, with size measures depending on a PSU-level
    outcome shock that model \eqref{eq:sim-model} deliberately omits.
    We Poisson-sample units within PSUs. The PSU stage is stratified
    by group. The two targets therefore differ, and the true parameter
    values are known.
\end{itemize}
Weights are the exact inverse first-order inclusion probabilities of
the design, scaled by the declared global unit-mean
convention.

\paragraph{Empty-cluster handling.}
We retain empty second-stage Poisson samples as drawn.
Redrawing them or forcing inclusion would change the effective
first-order probabilities, making them inconsistent with the weights.
Selected PSUs with empty second-stage samples remain in the design-PSU
cluster universe and enter the stratified meat with zero score totals.
In the smallest cells, these account for 6--7\% of selected PSUs per
replication.
We regenerate a replication only when an entire model group is empty,
using a deterministically offset seed. This occurs with probability of
order $10^{-3}$ or below per replicate, only in the smallest cells.
All summaries therefore condition on every group being observed.
Regeneration was needed in 2 of the \slotnum{4{,}800} replications;
trigger counts for each replication are provided with the replication
code.

Informativeness has three levels ($\gamma = $ none, moderate, strong).
We measure design effects on the realized samples rather than setting
them in advance. Realized global Kish $\DEFF$ averages
\slotnum{1.65}, \slotnum{2.0}, and \slotnum{2.5--3.9} across the three
levels.

\paragraph{Factorial grid.}
We cross $J \in \{20, 50\}$, expected group sample size
$\bar{n}_j \in \{10, 50\}$ (the small-$\bar n_j$ cells probe the
small-cluster regime of \Cref{sec:scope}), the three informativeness
levels, and the two structural arms:
$2 \times 2 \times 3 \times 2 = 24$ cells, 200 replications each,
\slotnum{4{,}800} models.

\paragraph{Estimation, convergence, and targets.}
We fit each replication using \texttt{cmdstanr} (non-centered
parameterization, 4 chains $\times$ 1{,}500 post-warmup draws after
1{,}000 warmup iterations, \texttt{adapt\_delta} $= 0.90$).
We require $\widehat{R} \le 1.01$ on all
parameters including every $\theta_j$, bulk ESS $\ge 100$, and zero
divergences.
\slotnum{99.3\%} of replications meet these criteria.
We retain failures in the archive but exclude them from the summaries,
with counts reported in every table.
We compute the $\DER$ under the model-group target in both arms and
under the design-PSU target (stratified by group) in the
$c \neq j$ arm.
We compare three interval strategies at nominal 90\%: \emph{naive},
\emph{selective} (\Cref{alg:ccc} at $c_0 = 1.2$, flagged block
corrected jointly), and \emph{blanket} (every parameter rescaled
marginally; the full joint correction is infeasible at the
model-group target, whose meat has rank at most $J < d$).
With 200 replications, coverage estimates carry Monte Carlo standard
errors of about 2 percentage points.

The main simulation summaries appear in
\Cref{tab:sim-coverage,fig:der-separation,fig:coverage}.

\subsection{Results}
\label{sec:sim-results}

\input{tables/tab2_sim_coverage}

\paragraph{Design-effect inheritance and shielding.}
The within-group coefficient $\beta_1$ tracks the realized design
effect: its mean $\DER$ rises from \slotnum{1.4} under weak designs
to \slotnum{2.0} under strong informativeness.
The between-group coefficients remain shielded (cell-mean
$\DER \le$ \slotnum{0.16}). Under the model-group target, random
effects also have low DERs, consistent with Tier~II (median $\DER$
\slotnum{0.03--0.08}).
These results follow the pattern of
\Cref{thm:fixed-effect,thm:random-effect} under selection operating
through $\theta$ and $y$.
Strong stage-one informativeness also biases the point estimates of
between-group and group-level quantities. In the strongest
sampled-groups cells, pooled naive coverage falls to \slotnum{0.53}
and \slotnum{0.63}, with per-cell coverage as low as \slotnum{0.47},
even though their $\DER$s are far below one.
Variance matching does not correct this location bias, and the
$\DER$ does not diagnose it (\Cref{sec:introduction}).

\begin{figure}[H]
  \centering
  \includegraphics[width=\textwidth]{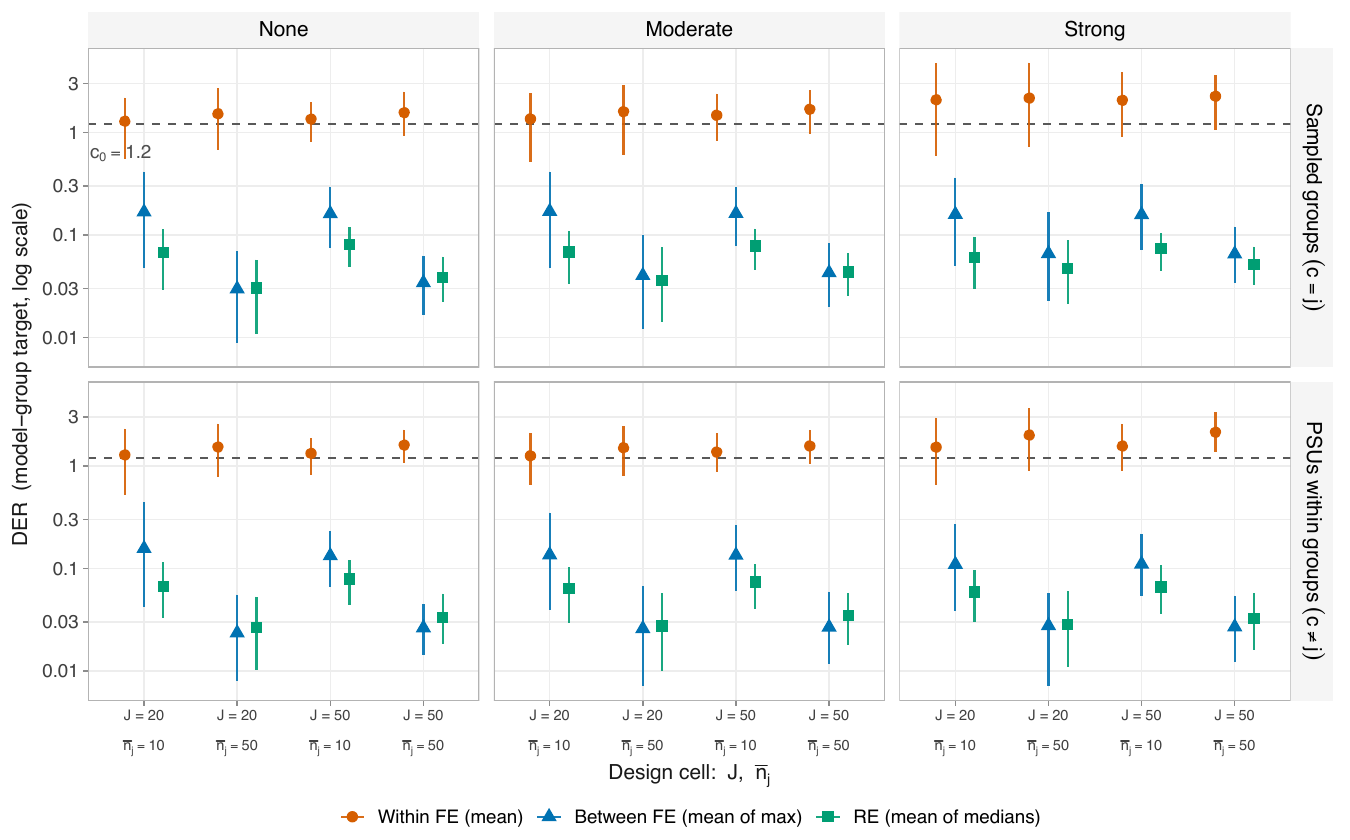}
  \caption{$\DER$ separation profile under the model-group target,
    by informativeness level and structural arm (cell means over
    replications meeting the convergence criteria). The within-group
    coefficient tracks the realized design effect; the between-group coefficients and
    the random-effect medians remain far below the threshold
    $c_0 = 1.2$ (dashed line).}
  \label{fig:der-separation}
\end{figure}

\paragraph{The dual-target comparison.}
In the $c \neq j$ arm, the targets differ as suggested by the NSECE.
The model-group target flags \slotnum{0.5--1.1} parameters per
replication on average, while the design-PSU target flags
\slotnum{2--18}.
The gap widens with $J$ and $\bar n_j$ because the omitted PSU shock
contributes design variability that the group-level posterior smooths
over.
Coverage against the known truth supports this interpretation.
Across all flagged group effects, naive intervals cover the true
$\theta_j$ \slotnum{76\%} of the time; selective correction under the
design-PSU target increases coverage to \slotnum{89\%}.

\paragraph{Selective versus blanket correction.}
Under informative sampling, naive coverage of the exposed
coefficient falls to \slotnum{65--84\%} against the nominal 90\%,
and selective correction increases it to \slotnum{80--89\%},
matching blanket correction on the flagged parameter.
Residual undercoverage at $J = 20$ clusters reflects the finite-cluster
behavior of sandwich estimators.
The strategies differ for the remaining parameters: blanket rescaling
reduces random-effect coverage to \slotnum{0.11--0.37}, while
selective correction leaves those parameters and their coverage
unchanged by construction.

\begin{figure}[H]
  \centering
  \includegraphics[width=\textwidth]{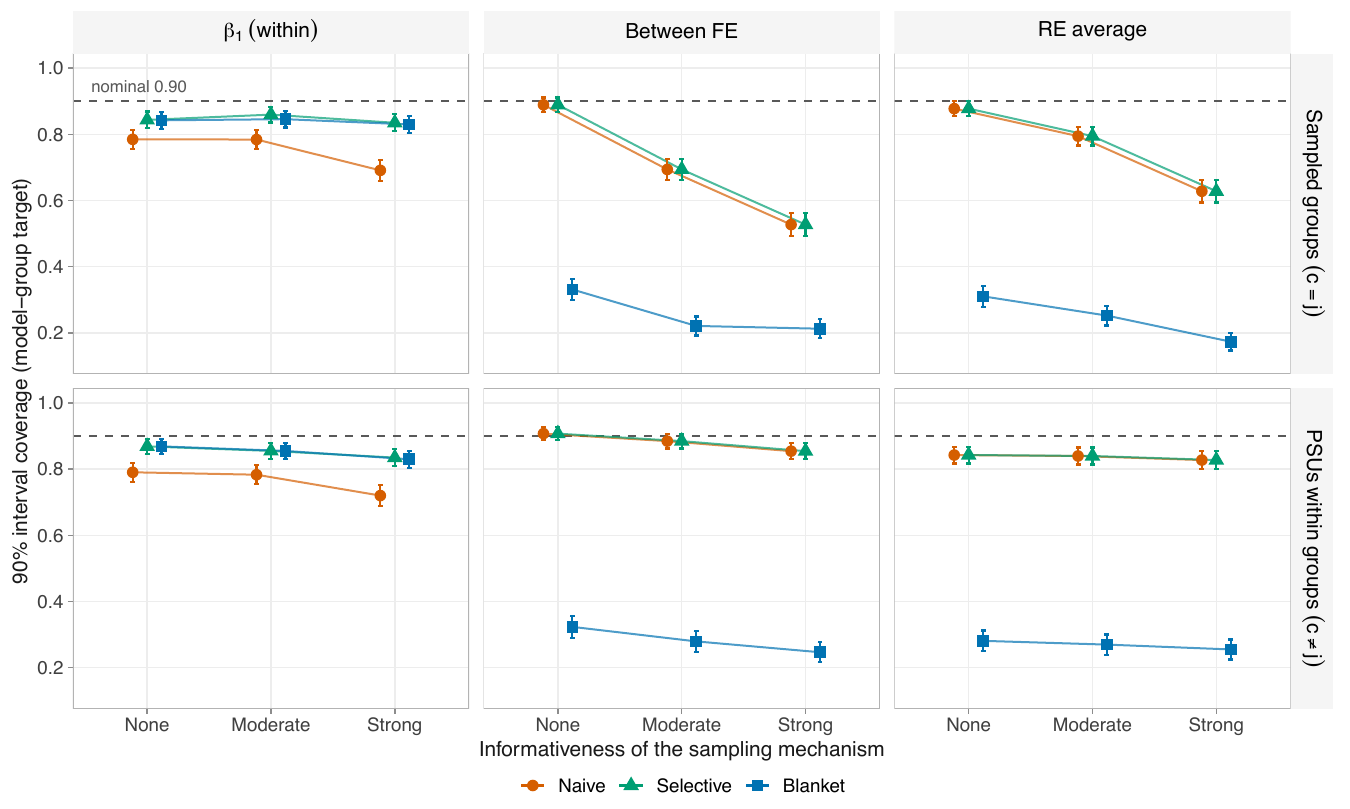}
  \caption{Coverage of nominal 90\% intervals under the three
    strategies (naive, selective at $c_0 = 1.2$, blanket), by
    informativeness level: (a) the exposed within-group coefficient,
    (b) between-group coefficients, (c) random-effect average
    (model-group target).  Error bars show Monte Carlo uncertainty.}
  \label{fig:coverage}
\end{figure}

\paragraph{Separation and threshold sensitivity.}
The within-group cell-mean $\DER$ exceeds every
non-target cell mean by a factor of at least \slotnum{7.1} (range
\slotnum{7.1--27.4} across cells).
The distributions overlap across individual replications: the largest
single non-target $\DER$ across all replications meeting the
convergence criteria is \slotnum{2.15}.
At $c_0 = 1.2$, the per-replication false-positive rate is
\slotnum{2.1}\%.
A false positive is a correction applied to a protected parameter;
its cost is a wider interval, not invalid inference.
Across $c_0 \in \{0.8, 1.0, 1.2, 1.5, 2.0\}$, the false-positive rate falls
\slotnum{16.0} $\to$ \slotnum{5.9} $\to$ \slotnum{2.1} $\to$
\slotnum{0.3} $\to$ \slotnum{0.0}\%.
Decisions on protected parameters are unchanged in \slotnum{96.9}\%
of replications for $c_0 \in [1.1, 1.6]$.
The exposed coefficient's flag changes as the threshold crosses its
$\DER$; the full sweep appears in the supplement.

%% file: tables/tab2_sim_coverage.tex
\begin{table}[H]
\small
\singlespacing
\centering
\caption{Simulation study: 90\% interval coverage under the model-group target by informativeness, structural arm, and interval strategy. Results pool $J \in \{20, 50\}$ and $\bar n_j \in \{10, 50\}$ (four cells $\times$ 200 replications per row block). ``Reps'' counts replications meeting the convergence criteria out of 800; failures (33 of 4{,}800 in total) are excluded from all summaries. The selective strategy corrects only the flagged block (threshold $c_0 = 1.2$); the blanket strategy rescales every parameter marginally to the target.}
\label{tab:sim-coverage}
\begin{tabular}{@{}llcccc@{}}
\toprule
 & & & \multicolumn{3}{c}{Coverage (nominal 0.90)} \\
\cmidrule(lr){4-6}
Informativeness & Reps & Strategy & $\beta_1$ (within) & FE between & RE average \\
\midrule
\multicolumn{6}{@{}l}{\itshape Sampled groups ($c = j$)} \\
None & 799 & Naive & 0.785 & 0.889 & 0.878 \\
 &  & Selective & 0.844 & 0.889 & 0.878 \\
 &  & Blanket & 0.842 & 0.331 & 0.310 \\
\addlinespace
Moderate & 797 & Naive & 0.784 & 0.694 & 0.794 \\
 &  & Selective & 0.859 & 0.694 & 0.794 \\
 &  & Blanket & 0.846 & 0.221 & 0.252 \\
\addlinespace
Strong & 793 & Naive & 0.691 & 0.527 & 0.628 \\
 &  & Selective & 0.835 & 0.527 & 0.628 \\
 &  & Blanket & 0.830 & 0.212 & 0.173 \\
\midrule
\multicolumn{6}{@{}l}{\itshape PSUs within groups ($c \neq j$)} \\
None & 792 & Naive & 0.790 & 0.907 & 0.843 \\
 &  & Selective & 0.869 & 0.907 & 0.843 \\
 &  & Blanket & 0.869 & 0.323 & 0.281 \\
\addlinespace
Moderate & 793 & Naive & 0.783 & 0.885 & 0.840 \\
 &  & Selective & 0.855 & 0.885 & 0.840 \\
 &  & Blanket & 0.855 & 0.279 & 0.270 \\
\addlinespace
Strong & 793 & Naive & 0.720 & 0.854 & 0.827 \\
 &  & Selective & 0.835 & 0.854 & 0.828 \\
 &  & Blanket & 0.830 & 0.247 & 0.255 \\
\bottomrule
\end{tabular}

\vspace{4pt}
\begin{minipage}{\textwidth}
\footnotesize
\textit{Notes.} ``FE between'' averages the intercept and the between-group coefficient $\beta_2$; ``RE average'' averages the $J$ random intercepts. Monte Carlo standard errors are at most about 2 percentage points per cell (200 replications) and about 1 percentage point for the pooled entries shown here ($\approx 800$ replications).
\end{minipage}
\end{table}

%% file: section5_application.tex

\section{Application: NSECE 2019}
\label{sec:application}

We apply the framework to the 2019 National Survey of Early Care and
Education \citep[NSECE;][]{NSECE2022}, a nationally representative
study of childcare providers
\citep[data and model specification from][]{Lee2026}.
The stratified multistage design includes 415 PSUs in 30 strata,
$N = 6{,}785$ center-based providers, and $J = 51$ states (including DC).
The model groups differ from the design PSUs, and every state is in
the sample by construction. We therefore specify both the variance
target and weight convention.

\subsection{Data, model, and declared choices}
\label{sec:app_data}

The outcome is binary participation in the state Infant--Toddler
quality improvement system (prevalence 64.7\%).
We include a within-state poverty measure (centered within state;
\citealp{enders2007centering}) and a between-state tiered-reimbursement
indicator (Tier~I-b).
The CWC construction removes between-state variation from the poverty
measure, placing it in Tier~I-a.
We fit the hierarchical logistic regression
\begin{equation}\label{eq:nsece_model}
  y_i \mid \bbeta, \theta_{j[i]}
  \;\sim\;
  \Bern\!\bigl(
    \expit(\beta_0 + \beta_1 x_{i}^{(\mathrm{cwc})}
    + \beta_2 z_j + \theta_{j[i]})
  \bigr),
  \qquad
  \theta_j \;\sim\; \Norm(0, \sigma_\theta^2),
\end{equation}
via \texttt{cmdstanr} (non-centered parameterization,
pseudo-likelihood weighting, 4 chains $\times$ 2{,}000 post-warmup
iterations, \texttt{adapt\_delta} $= 0.95$).

We scale the raw design weights under the global unit-mean convention
(raw-weight Kish $\DEFF = 3.76$); \Cref{sec:app_convention} examines
the within-state alternative.
We report both aggregation units. The design-PSU target (415 PSUs,
30 strata, centered DF-corrected meat) is our primary design-based
analysis. The model-group target (51 states) supports interpretation
using the closed forms of \Cref{sec:framework}.
We use the posterior mean as the plug-in value,
$\hat{\sigma}_\theta = 0.650$.
The fit meets the convergence criteria: $\widehat{R} \le 1.004$ on
all parameters, bulk ESS $\ge 1{,}499$, and zero divergences.

\subsection{The dual-target diagnostic}
\label{sec:app_results}

\Cref{tab:nsece_der} and \Cref{fig:nsece_der} present the $\DER$
profile of all 54 parameters under both targets.
The full diagnostic and correction procedure runs in a fraction of a
second after MCMC.

\input{tables/tab3_nsece_der}

\begin{figure}[!t]
  \centering
  \includegraphics[width=\textwidth]{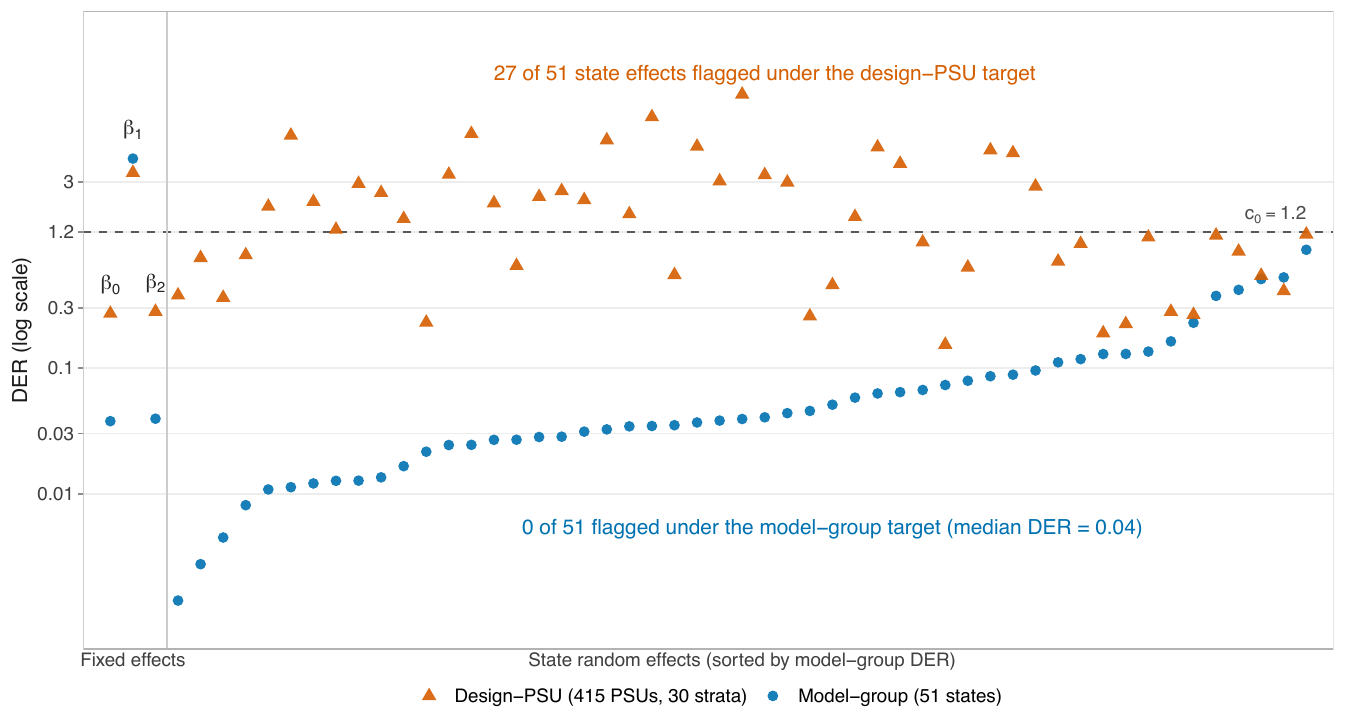}
  \caption{NSECE $\DER$ profile for all 54 parameters under the two
    declared targets (log scale; fixed effects labeled; states sorted
    by model-group-target $\DER$).  Under the model-group target only
    the poverty coefficient crosses $c_0 = 1.2$; under the design-PSU
    target the majority of state effects do.}
  \label{fig:nsece_der}
\end{figure}

\paragraph{Fixed effects.}
Both targets show the within/between contrast of
\Cref{thm:fixed-effect}.
The within-state poverty coefficient $\beta_1$ has
$\DER = 3.55$ under the design-PSU target and $\DER = 4.59$ under the
model-group target. It is flagged under both targets and every
target--convention combination we examine.
The intercept and the tiered-reimbursement coefficient have DERs far
below one under both targets (\slotnum{$\le 0.28$}). These parameters
are identified from between-state variation and are shielded as
Tier~I-b predicts.

\paragraph{Random effects.}
Under the model-group target, no state effect approaches the
threshold: hierarchical shrinkage absorbs the design effect, as
predicted for Tier~II.
Under the design-PSU target, \slotnum{27} of the 51 state effects
exceed $c_0 = 1.2$ (median random-effect $\DER =$ \slotnum{1.53}).
The state-level posteriors do not reflect the within-state,
between-PSU sampling variability measured by this target.
The random effects in \eqref{eq:nsece_model} smooth over the design
layer captured by the 415-PSU meat.

The model-group result indicates that the posterior does not
substantially understate uncertainty about state-level structure.
The design-PSU result shows that it understates the design's sampling
variability for the state effects.
For design-based sampling variability, the usual survey-inference
standard, the design-PSU target is the appropriate primary analysis.
We use it for that reason: the target should match the inferential
question rather than a preferred result.
In total, \slotnum{28} of 54 parameters are flagged under the
design-PSU target versus \slotnum{1} of 54 under the model-group
target.

\paragraph{Selective correction.}
Under the model-group target the flagged set is $\{\beta_1\}$ and the
correction is the scalar case of \Cref{alg:ccc}.
Under the design-PSU target the flagged block has
$|\mathcal{F}| = $ \slotnum{28}. We apply the joint block-Cholesky
correction in Step~3 to match both the target variances and the
covariance structure within the flagged block.
For $\beta_1$, the 90\% credible interval widens from
$[-0.166, -0.072]$ to $[-0.207, -0.030]$ under the design-PSU
target, a factor of \slotnum{1.88} in width.
Under the model-group target, blanket marginal correction would rescale
posterior draws for all 54 parameters toward a rank-deficient
state-level target, narrowing
\slotnum{53} intervals (mean width ratio \slotnum{0.26}, worst
\slotnum{0.04}). Selective correction avoids this narrowing.

\subsection{Sensitivity to the weight-scaling convention}
\label{sec:app_convention}

Because weights enter the pseudo-likelihood, changing their scaling
changes the fit as well as the declared target (\Cref{sec:pseudo}).
We examine this sensitivity in the NSECE.

\input{tables/tab4_nsece_2x2}
Refitting \eqref{eq:nsece_model} with weights rescaled within each
state (the multilevel scaling convention of
\citealp{PfeffermannEtAl1998b}) yields
$\DER(\beta_1) = $ \slotnum{2.70} (model-group) and \slotnum{4.70}
(design-PSU).
The ordering of the targets reverses relative to the global convention
because rescaling reallocates state effective sample sizes
(\slotnum{[17, 1{,}110]} within-state versus \slotnum{[21, 594]} global).
Two findings hold under both conventions: the poverty coefficient is
flagged under every combination, and random effects are flagged only
under the design-PSU target (\slotnum{27--28} states;
\Cref{tab:nsece_convention}).
These results do not establish that one convention is correct and the
other incorrect. They show why the convention and aggregation unit
must be declared: the diagnostic is defined relative to both, and the
pseudo-posterior itself depends on the weight convention.

\subsection{Interpretation checks and threshold sensitivity}
\label{sec:app_checks}

The structural protection factors computed from the working
information split (\Cref{sec:framework}) separate the covariate
types: $R_k =$ \slotnum{0.15} for the within-state poverty covariate
versus \slotnum{$\approx 0.93$} for the between-state parameters,
matching the Tier~I-a/I-b assignment.

Per-state random-effect $\DER$s under the model-group target have the
estimator noise anticipated by \Cref{rem:estimated}: each is built
from one centered score-total contrast per state.
We therefore interpret them by tier rather than by state.
The design-PSU target, with 415 clusters in 30 strata, is estimated
more precisely.

The flagged set is robust to the threshold: across
$c_0 \in [0.8, 2.0]$ the poverty coefficient is flagged at every
value under both targets.
The design-PSU random-effect count varies smoothly with $c_0$;
Section~S-E.5 of the Supporting Information gives the full sweep.
As in the simulation, the diagnostic adds well under a second to an
analysis whose MCMC takes minutes.

%% file: tables/tab3_nsece_der.tex
\begin{table}[H]
\small
\singlespacing
\centering
\caption{NSECE 2019: $\DER$ diagnostics for all 54 parameters under both declared variance targets. Raw design weights use global unit-mean scaling (raw-weight Kish $\DEFF = 3.76$). The fit meets the convergence criteria: $\widehat{R} \le 1.004$ on all parameters and zero divergent transitions. Flags mark $\DER > c_0 = 1.2$.}
\label{tab:nsece_der}
\begin{tabular}{@{}llcccc@{}}
\toprule
 & & \multicolumn{2}{c}{Model-group target} & \multicolumn{2}{c}{Design-PSU target} \\
\cmidrule(lr){3-4} \cmidrule(lr){5-6}
Parameter & Type & $\DER$ & Flag & $\DER$ & Flag \\
\midrule
$\beta_0$ (intercept) & Between & 0.038 & --- & 0.270 & --- \\
$\beta_1$ (poverty, CWC) & Within & 4.590 & $\checkmark$ & 3.549 & $\checkmark$ \\
$\beta_2$ (tiered reimb.) & Between & 0.039 & --- & 0.280 & --- \\
\addlinespace
$\theta_j$ (51 state effects) & RE & median 0.039 & 0 of 51 & median 1.53 & 27 of 51 \\
\midrule
Flagged total ($\DER > 1.2$) & & \multicolumn{2}{c}{1 of 54} & \multicolumn{2}{c}{28 of 54} \\
\bottomrule
\end{tabular}

\vspace{4pt}
\begin{minipage}{\textwidth}
\footnotesize
\textit{Notes.} Model-group target: clusters $=$ 51 states, single stratum. Design-PSU target: clusters $=$ 415 design PSUs within 30 strata. Both use the centered, DF-corrected meat and the declared global unit-mean weight convention; the within-state alternative is examined in Table~\ref{tab:nsece_convention}. For $\beta_1$, the selective correction under the design-PSU target widens the 90\% credible interval from $[-0.166, -0.072]$ to $[-0.207, -0.030]$ (factor 1.88).
\end{minipage}
\end{table}

%% file: tables/tab4_nsece_2x2.tex
\begin{table}[H]
\small
\singlespacing
\centering
\caption{NSECE 2019: sensitivity of the diagnostic to the declared weight convention. The rows represent different pseudo-posteriors: weights enter the pseudo-likelihood, so changing the convention changes the fit as well as the target. Each fit is evaluated under both aggregation units with identical target settings (centered, DF-corrected meat; unit-mean normalization).}
\label{tab:nsece_convention}
\begin{tabular}{@{}lcccc@{}}
\toprule
 & \multicolumn{2}{c}{$\DER(\beta_1)$ (poverty)} & \multicolumn{2}{c}{Flagged (of 54)} \\
\cmidrule(lr){2-3} \cmidrule(lr){4-5}
Weight convention & Model-group & Design-PSU & Model-group & Design-PSU \\
\midrule
Global unit-mean (declared) & 4.59 & 3.55 & 1 & 28 \\
Within-state scaling & 2.70 & 4.70 & 1 & 29 \\
\bottomrule
\end{tabular}

\vspace{4pt}
\begin{minipage}{\textwidth}
\footnotesize
\textit{Notes.} Within-state scaling rescales weights to sum to each state's sample size (the multilevel scaling convention of \citealp{PfeffermannEtAl1998b}). State pseudo-likelihood sample sizes then equal the realized $n_j$ ($[17, 1{,}110]$); the global convention allocates them by weighted share ($[21, 594]$). This reallocation of effective information across states reverses the ordering of the two targets for $\beta_1$. The poverty coefficient is flagged under all four combinations. Random effects remain flagged only under the design-PSU target (27 and 28 of 51 states).
\end{minipage}
\end{table}

%% file: section6.tex

\section{Discussion}
\label{sec:discussion}

The Design Effect Ratio identifies survey-sensitive parameters in
Bayesian hierarchical models by comparing posterior variances with a
declared design-based target.
Our decompositions show how hierarchical shrinkage moderates the
classical design effect (\Cref{thm:fixed-effect,thm:random-effect}).
We evaluated the resulting compute-classify-correct workflow in
\slotnum{4{,}800} simulated data sets with informative sampling and
in the NSECE application.
With global unit-mean weights, the application flagged
\slotnum{1} parameter under the model-group target and
\slotnum{28} of 54 under the design-PSU target.
The poverty coefficient was flagged under every target and convention
examined.

The motivating question of which parameters need design-based
correction is a quantitative restatement of the question posed by
\citet{Pfeffermann1993}.
We follow the calibrated Bayes perspective of \citet{Little2004, Little2012},
building on \citeauthor{Rubin1984}'s (\citeyear{Rubin1984}) principle
that Bayesian procedures should be evaluated by their frequentist
properties.
The $\DER$ states the comparison target and measures each posterior
variance relative to it.
Within its model class, the conservation identity
(\Cref{cor:conservation}) shows that shrinkage moves design
sensitivity from group effects to parameters identified by
between-group information while leaving the total unchanged.
Whether an analogous identity holds more generally is open.

\paragraph{Limitations.}
The $\DER$ excludes hyperparameters by construction
(\Cref{sec:limits}); replicate-weight approaches
\citep{SavitskyGershunskaya2026} may complement it.
The closed forms assume the balanced conjugate Gaussian class
(C1)--(C3) with common within-group design effects.
Under heterogeneous design effects or in multi-covariate models they
give the qualitative ordering.
Exact magnitudes come from the general formula and the ensemble-exact
form in Section~S-B.3 of the Supporting Information.
The theory covers two-level models with scalar random
intercepts; three-level and crossed structures require additional
block algebra.

The threshold $c_0 = 1.2$ is a calibrated default, not the solution of a
decision problem.
Selective correction targets marginal calibration; joint
credible regions for the full parameter vector call for the full correction.
The $\DER$ flags parameters one at a time.
We have not formally addressed how these decisions affect subsequent
inference, a problem related to selective inference
(\citealp{BerkEtAl2013}).

The diagnostic also presupposes a valid variance target.
Under informative within-group selection with small group sizes, the
identification caution of \Cref{sec:scope}
\citep{RaoVerretHidiroglou2013, YiRaoLi2016, Slud2026} applies to the
group-level layer.
Our small-$n_j$ simulation cells show the practical consequences of
this limitation.

\paragraph{Extensions.}
The extension of the closed-form decompositions to non-conjugate
families is developed in a companion paper.
We regard hyperparameter uncertainty under informative sampling as the
most important open problem.
The $\DER$ decomposition also parallels the MSE decompositions of
small area estimation \citep{RaoMolina2015}.
These parallels could yield common design-sensitivity diagnostics for
EBLUP and fully Bayesian estimators.

The $\DER$ makes the design correction interpretable and parameter
specific.
For the NSECE under the model-group target, one parameter was
corrected instead of fifty-four; under the design-PSU target, the
diagnostic showed that the state effects themselves carry design
variability that the model understates.
The declared target specifies the variability against which we assess
each parameter and decide whether to correct its interval.

%% file: osm_body.tex

\input{sm_a}

\input{sm_b_proofs}

\input{sm_c_secondary}

\input{sm_b}

\input{sm_d}

\input{sm_e}

%% file: sm_a.tex

\section{Conditions for the DER Results}
\label{app:conditions}

We state the conditions for Theorems~1 and~2 and the conservation
identity, Corollary~1, of the main article.
We give their sources and identify where they enter the proofs in
\Cref{app:proofs}.
The \emph{design and weight conditions} (A1)--(A4) are adapted from
\citet{SavitskyToth2016} and concern the sampling design and survey
weights.
The \emph{model-structure conditions} (R1)--(R3) and (R5) specify the
hierarchical model requirements for the closed-form $\DER$ results.
The separate \emph{target-validity condition} (TV) gives the declared
variance target its frequentist interpretation.

Conditions (C1)--(C3) in Section~3.2 of the main article summarize
these requirements.
The conjugate Gaussian specification (C1) combines (R1), (R3), and
balance; it also includes (R5), which holds when the covariates are
not collinear.
The weight condition (C2) combines (A1) and (A4) with global unit-mean
scaling.
The model-group target (C3) combines (R2) and Condition~TV.

\subsection{Design and Weight Conditions (A1)--(A4)}
\label{sec:a-conditions}

The following conditions are adapted from \citet{SavitskyToth2016},
whose framework establishes consistency of pseudo-posterior inference
under informative sampling.
We specialize that framework to hierarchical models fitted with
normalized survey weights.
Savitsky and Toth label their conditions (A1)--(A6);
\Cref{sec:condition-sources} relates our conditions to theirs.

\paragraph{Condition (A1): Informative sampling model.}
The finite population $\mathcal{U}$ of size $N_{\mathrm{pop}}$ is
generated from a superpopulation model $P_0$ indexed by a true
parameter value $\bm{\phi}_0$.
The sampling mechanism selects units into the sample $S$ of size $N$
with inclusion probabilities $\pi_i > 0$ for all $i \in \mathcal{U}$,
and the normalized survey weights
$\tilde{w}_i = w_i / \bigl(\sum_{k \in S} w_k / N\bigr)$ satisfy
$\sum_{i \in S} \tilde{w}_i = N$, the global unit-mean convention
declared in Section~2.2 of the main article.

\paragraph{Condition (A2): Pseudo-likelihood regularity.}
The log pseudo-likelihood
$\ell_{\text{lik}}(\bm{\phi}) = \sum_{i \in S} \tilde{w}_i \log p(y_i \mid \bm{\phi})$
satisfies:
\begin{enumerate}[label=(\roman*)]
  \item $\ell_{\text{lik}}(\bm{\phi})$ is twice continuously
    differentiable in $\bm{\phi}$ in a neighborhood of the true
    parameter value $\bm{\phi}_0$.
  \item The weighted Fisher information
    $\mathbf{I}_N(\bm{\phi}_0) = -N^{-1} \nabla^2 \ell_{\text{lik}}(\bm{\phi}_0)$
    converges to a positive definite limit
    $\mathbf{I}(\bm{\phi}_0) \succ 0$.
  \item A uniform law of large numbers holds for
    $N^{-1} \nabla^2 \ell_{\text{lik}}(\bm{\phi})$ over compact subsets
    of the parameter space.
\end{enumerate}

\paragraph{Condition (A3): Prior support.}
The prior density $\pi(\bm{\phi})$ is continuous and strictly positive
in a neighborhood of $\bm{\phi}_0$.

\paragraph{Condition (A4): Weight boundedness.}
The normalized weights satisfy $\tilde{w}_i / N \to 0$ uniformly as
$N \to \infty$.
This prevents any single observation from dominating the
pseudo-likelihood.

\begin{remark}[Finite-sample identities and target validity]
\label{rem:algebra-and-target}
Theorems~1 and~2 of the main article are finite-sample algebraic
identities; their proofs invoke no central limit theorem
(\Cref{app:proofs}).
Condition~TV below supplies the asymptotic justification for
interpreting the variance target in repeated sampling.
Throughout, we use the stratified, centered, degrees-of-freedom-corrected
meat defined in Equation~(2.5) of the main article.
\end{remark}

\subsection{Model-Structure Conditions (R1)--(R3), (R5)}
\label{sec:r-conditions}

The following conditions are required for the exact closed-form $\DER$
formulas in Theorems~1 and~2 of the main article.
We retain the condition labels used in the proof annotations.

\paragraph{Condition (R1): Proper random-effect prior.}
$\tau = 1/\sigma_\theta^2 > 0$ (the random-effect prior precision is
finite and positive) and $a_j = \widetilde{n}_j / \sigma_e^2 > 0$ for
all $j$ (each group has positive weighted sample size).
Consequently, $B_j = a_j/(a_j + \tau) \in (0,1)$ and
$\Sigma_B = \sum_k B_k > 0$.

\paragraph{Condition (R2): Group--PSU nesting.}
No sandwich aggregation unit crosses model-group boundaries: each
design PSU either coincides with a model group or is nested within
one.
Formally, for all PSUs $c$ and groups $j \neq k$: if PSU $c$ contains
observations from group $j$, then PSU $c$ contains no observations
from group $k$.
This ensures $[\Jclust]_{\theta_j, \theta_k} \approx 0$ for
$j \neq k$.

\paragraph{Condition (R3): Flat prior on fixed effects.}
$\mathbf{H}_{\text{prior}}^{\beta} = -\nabla^2 \log \pi(\bbeta) \approx 0$.
The prior on $\bbeta$ is non-informative (flat or very weakly
informative), so that the observed information for $\bbeta$ is
determined by the likelihood alone.
This is the diffuse-$\bbeta$ convention built into the bread
definition, eq.~(2.4) of the main article.

\paragraph{Condition (R5): Non-degeneracy of Schur complement.}
The Schur complement $\mathbf{S}_\beta = \sum_j \mathbf{C}_j$ is
positive definite, where
$\mathbf{C}_j = \mathbf{W}_j + \frac{(1-B_j)}{a_j}\mathbf{b}_j\mathbf{b}_j^\t$
(notation as in \Cref{sec:theorem1}).
This ensures identifiability of all fixed-effect parameters.

\begin{remark}[Design effects for group score totals]
\label{rem:group-score-variance}
In the conjugate Gaussian setting, the scalar within-group relation
$[\Jclust]_{\text{within } j} = \DEFF_j \cdot
[\mathbf{H}_{\text{lik}}]_{\text{within } j}$ follows from the
definition of the group design effect.
With model groups as the aggregation unit, we define $\DEFF_j$ as the
ratio of the design variance of a group score total to its
corresponding data-information contribution.
The scalar form therefore specifies the theorem's parameterization;
it is not an additional design assumption.
When covariates within a group have different design effects, the
closed forms interpret $\DEFF_j$ as the score-weighted average within
the group.
Algorithm~1 of the main article computes the $\DER$ exactly from the
full matrices $\Hobs$ and $\Jclust$ without requiring scalar structure.
The corresponding non-conjugate theory is developed in a companion paper.
\end{remark}

\subsection{Condition TV: Target Validity}
\label{sec:tv-condition}

The decomposition theorems need no asymptotics.
Asymptotic theory instead justifies the variance target and the
frequentist coverage of the corrected posterior.
We state those requirements as Condition~TV, relying on published
pseudo-posterior theory.

\paragraph{Condition (TV): Target validity.}
At the declared aggregation unit, the stratified, centered,
DF-corrected cluster covariance $J_{\mathrm{c}}$ in Equation~(2.5)
of the main article consistently estimates the design covariance of
the weighted score totals, and the corrected posterior (the
pseudo-posterior mapped through the
$\mathbf{R}_2\mathbf{R}_1^{-1}$ transformation described in
Section~2.4 of the main article) attains asymptotically nominal
frequentist coverage.
Sufficient conditions are given by \citet{SavitskyToth2016} and
\citet{WilliamsSavitsky2021}.
Theorem~3.1 of \citet{SavitskyToth2016} establishes contraction of
the pseudo-posterior on the true generating distribution.
Theorems~2 and~3 of \citet{WilliamsSavitsky2021} give the sandwich
form of the asymptotic covariance of the weighted score totals and
the asymptotic Gaussian form of the pseudo-posterior that the
correction rescales.
Their empirical evaluations also document near-nominal coverage of
the corrected intervals across a range of sample designs
\citep{WilliamsSavitsky2021}.

As explained in Remark~\ref{rem:algebra-and-target}, Condition~TV
justifies the design-based interpretation of $\Vtarget$ and the
$\DER$; it is not used in the algebraic proofs in \Cref{app:proofs}.

\subsection{Sources and Use of the Conditions}
\label{sec:condition-sources}

\Cref{tab:condition-sources} maps each design/weight condition to its
counterpart in \citet{SavitskyToth2016} (labels ST-A1 to ST-A6 refer
to their conditions (A1)--(A6)) and records where each condition
is used in the proofs in \Cref{app:proofs}.
Condition (A1) enters the proofs through the positive weights required
by (R1).
Conditions (A2)--(A4) support the pseudo-posterior and the asymptotic
interpretation of the target through Condition~TV.
Each proof in \Cref{app:proofs} identifies the steps that require
(R1)--(R3) and (R5).

\begin{table}[htbp]
  \centering
  \caption{Sources of the design/weight conditions and Condition~TV,
    and their use in this supplement.
    ST = \citet{SavitskyToth2016}; WS = \citet{WilliamsSavitsky2021}.
    ST label their conditions (A1)--(A6); WS label theirs (A1)--(A7).}
  \label{tab:condition-sources}
  \small
  \begin{tabularx}{\textwidth}{@{}>{\raggedright\arraybackslash}p{1.4cm} >{\raggedright\arraybackslash}p{4.2cm} >{\raggedright\arraybackslash}p{4.6cm} >{\raggedright\arraybackslash}X@{}}
    \toprule
    \textbf{Cond.} & \textbf{Published counterpart}
      & \textbf{Where used (this supplement)}
      & \textbf{Note} \\
    \midrule
    (A1) & ST-A4 (non-zero inclusion probabilities), with the
      sampling framework of ST Section~2
      & \Cref{sec:theorem1} Step~1 and \Cref{sec:theorem2} Step~1,
        through (R1): $\tilde w_{ij} > 0$ gives $a_j > 0$, so
        $\mathbf{D}_H$ is invertible
      & WS counterpart: WS-A5. \\
    \addlinespace
    (A2) & ST-A1--ST-A2 (entropy and sieve regularity of the model
      space)
      & Existence and finiteness of $\Hobs$ and of the score totals
        (setup of \Cref{sec:theorem1,sec:theorem2}); not used in
        individual algebraic steps
      & Interpretive mapping: ST state regularity nonparametrically;
        the parametric analogue is WS-A1--WS-A2. \\
    \addlinespace
    (A3) & ST-A3 (prior mass covering the truth)
      & Properness of the pseudo-posterior (setup); not used in
        algebraic steps
      & \\
    \addlinespace
    (A4) & ST-A4 (inclusion probabilities bounded away from zero
      bound the normalized weights), with ST-A6 (growing sample)
      & No algebraic step; supports the asymptotics required by
        Condition~TV
      & WS counterparts: WS-A5, WS-A7. \\
    \addlinespace
    (TV) & Theorem~3.1 of ST (contraction);
      Theorems~2--3 of WS (sandwich covariance of the weighted score;
      asymptotic form of the pseudo-posterior)
      & Gives $\Vtarget$, hence the $\DER$, its design-based meaning
        (Section~2.3 of the main article); used at no proof step
      & \\
    \bottomrule
  \end{tabularx}
\end{table}

\paragraph{Where the R-conditions are used.}
For convenience, the step-level map assembled in \Cref{app:proofs} is:
(R1) at \Cref{sec:theorem1} Steps~1, 4, 5 and \Cref{sec:theorem2}
Steps~1, 4--6;
(R2) at the block-structure setup and Step~2 of \Cref{sec:theorem1}
and Step~2 of \Cref{sec:theorem2};
(R3) at the block-structure setup of \Cref{sec:theorem1} and Step~1 of
\Cref{sec:theorem2};
(R5) at \Cref{sec:theorem1} Steps~1 and~3.
Balance (the (C1) hypothesis) is used only at
\Cref{sec:theorem1} Step~5 (upper bound) and \Cref{sec:theorem2}
Step~7; the unbalanced per-group form is treated in
\Cref{app:unbalanced}.

%% file: sm_b_proofs.tex

\section{Proofs of the Decomposition Theorems}
\label{app:proofs}

This section proves Theorems~1 and~2 of the main article.
Each proof lists the conditions from \Cref{app:conditions} used in
each step. Inline annotations such as ``\textup{[uses (R1)]}'' identify
where a condition is used.
The proofs use finite-sample algebra, with no asymptotic arguments.
The design conditions
(A1)--(A4) enter only through the positivity of the weights (via (R1))
and through the interpretation of the target (Condition~TV); see the
condition-source table in \Cref{sec:condition-sources}.

\paragraph{Notation and variance target.}
Throughout this section, $\Hobs$ denotes the penalized (posterior)
Hessian $H$ of eq.~(2.4) of the main article---the curvature of the
log-pseudo-likelihood plus the random-effect prior precision
$\tau = 1/\sigma^2_\theta$---evaluated at the plug-in point.
$\Jclust$ denotes the meat at the model-group aggregation unit, and
$\Vsand = \Hobs^{-1} \Jclust \Hobs^{-1}$ is the declared target
$\Vtarget$ of eq.~(2.3) of the main article.
The results describe the \emph{population} counterpart of
the diagnostic (Remark~2 of the main article): $\Jclust$ is
parameterized as the design covariance of the weighted score totals,
with the within-group design effects $\DEFF_j$ entering as defined in
Section~2.5 of the main article and \Cref{sec:r-conditions}
(Remark~\ref{rem:group-score-variance}).

\subsection{Proof of Theorem~1: Fixed Effects}
\label{sec:theorem1}

\subsubsection{Statement}

\textbf{Theorem 1} (DER for fixed effects; restating Theorem~1 in the main article)\textbf{.}
\textit{Consider the normal hierarchical model with covariates:}
\begin{equation*}
  y_{ij} = \mathbf{x}_{ij}^\t \bbeta + \theta_j + \varepsilon_{ij}, \quad
  \varepsilon_{ij} \sim N(0, \sigma_e^2), \quad
  \theta_j \sim N(0, \sigma_\theta^2)
\end{equation*}
\textit{where $\bbeta \in \bR^p$, flat prior on $\bbeta$, $j = 1, \ldots, J$.}

\textit{Under conditions (R1)--(R3) and (R5)
(\Cref{app:conditions}), the $\DER$ for $\beta_k$ is:}
\begin{equation*}
  \DER_{\beta_k} = \frac{[\mathbf{S}_\beta^{-1} \mathbf{M} \, \mathbf{S}_\beta^{-1}]_{kk}}{[\mathbf{S}_\beta^{-1}]_{kk}}
\end{equation*}
\textit{where $\mathbf{S}_\beta = \sum_{j=1}^J \mathbf{C}_j$ is the Schur complement (bread), $\mathbf{M} = \sum_{j=1}^J \DEFF_j \, \mathbf{C}_j^{(d)}$ is the effective meat, and}
\begin{equation*}
  \mathbf{C}_j = \mathbf{W}_j + \frac{(1-B_j)}{a_j} \mathbf{b}_j \mathbf{b}_j^\t, \qquad
  \mathbf{C}_j^{(d)} = \mathbf{W}_j + \frac{(1-B_j)^2}{a_j} \mathbf{b}_j \mathbf{b}_j^\t.
\end{equation*}

\textit{When $\DEFF_j = \DEFF$ for all $j$:}
\begin{equation*}
  \DER_{\beta_k} = \DEFF \cdot (1 - R_k), \qquad R_k \in [0, B]
\end{equation*}
\textit{Here $R_k$ is the shrinkage attenuation factor, with}
\begin{align*}
  R_k
  &= \frac{[\mathbf{S}_\beta^{-1} \bDelta \, \mathbf{S}_\beta^{-1}]_{kk}}
          {[\mathbf{S}_\beta^{-1}]_{kk}}, \\
  \bDelta
  &= \sum_j \frac{\tau}{(a_j + \tau)^2}
     \mathbf{b}_j \mathbf{b}_j^\t.
\end{align*}

\smallskip
\noindent
The Schur-complement sandwich
$\mathbf{S}_\beta^{-1} \mathbf{M} \, \mathbf{S}_\beta^{-1}$ is the
marginal fixed-effect target $V_{\bbeta}$ of eq.~(2.6) of the main article, in the population parameterization.

\paragraph{Conditions used.}
(R1)---Steps~1, 4, and~5 (invertibility of $\mathbf{D}_H$;
$\bDelta \succeq 0$);
(R2)---block-structure setup and Step~2 (no cross-group blocks in the
meat);
(R3)---block-structure setup (no prior curvature in the
$\bbeta$ block);
(R5)---Steps~1 and~3 ($\mathbf{S}_\beta$ invertible with positive
diagonal).
The balanced-case hypothesis ((C1)) is used only for the upper bound
in Step~5 and the special cases in \Cref{sec:theorem1-special}.
The common-$\DEFF$ hypothesis is used only in Step~4, which gives
$\DEFF(1 - R_k)$.
Theorem~1 of the main article states the common-$\DEFF$ hypothesis
separately: balance makes the shrinkage common, but not the design
effects.
Here we prove the general effective-meat form; the main article
states its common-design-effect special case.
No algebraic step directly uses Conditions (A1)--(A4).
Condition (A1) guarantees the strictly positive weights underlying
(R1), and (A2)--(A4) support the pseudo-posterior framework and Condition~TV
(\Cref{sec:condition-sources}).

\subsubsection{Notation}

\begin{center}
\small
\begin{tabularx}{\textwidth}{@{}lXl@{}}
\hline
Symbol & Definition & Dim.\ \\
\hline
$\mathbf{A}_j = \frac{1}{\sigma_e^2} \sum_{i \in j} \tilde{w}_{ij} \mathbf{x}_{ij} \mathbf{x}_{ij}^\t$ &
  Weighted Fisher information from group $j$ for $\bbeta$ & $p \times p$ \\[6pt]
$\mathbf{b}_j = \frac{1}{\sigma_e^2} \sum_{i \in j} \tilde{w}_{ij} \mathbf{x}_{ij}$ &
  Cross-term vector ($\bbeta$--$\theta_j$ coupling) & $p \times 1$ \\[6pt]
$\bar{\mathbf{x}}_j = \mathbf{b}_j / a_j$ &
  Weighted group mean covariate vector & $p \times 1$ \\[6pt]
$\mathbf{W}_j = \frac{1}{\sigma_e^2} \sum_{i \in j} \tilde{w}_{ij} (\mathbf{x}_{ij} - \bar{\mathbf{x}}_j)(\mathbf{x}_{ij} - \bar{\mathbf{x}}_j)^\t$ &
  Within-group covariate variation & $p \times p$ \\[6pt]
\hline
\end{tabularx}
\end{center}

\noindent
We have $\mathbf{A}_j = \mathbf{W}_j + a_j \bar{\mathbf{x}}_j \bar{\mathbf{x}}_j^\t$, separating total variation into within- and between-group components for group $j$.

\subsubsection{Block structure}

\noindent\textit{Notational convention.}  Throughout this proof, $\Hobs$
denotes the full posterior Hessian $\Hobs^{\mathrm{post}}$ (including
prior precision terms $\tau = 1/\sigma^2_\theta$ in the
random-effect block), as defined in eq.~(2.4) of the main article.

\smallskip
With parameter vector $\bm{\phi} = (\bbeta^\t, \theta_1, \ldots, \theta_J)^\t$:
\begin{equation*}
  \Hobs = \begin{pmatrix} \mathbf{A} & \mathbf{G} \\ \mathbf{G}^\t & \mathbf{D}_H \end{pmatrix}, \qquad
  \Jclust = \begin{pmatrix} \sum_j \DEFF_j \mathbf{A}_j & \DEFF_1 \mathbf{b}_1 & \cdots \\ \DEFF_1 \mathbf{b}_1^\t & d_1 & \cdots \\ \vdots & \vdots & \ddots \end{pmatrix}
\end{equation*}
where $\mathbf{A} = \sum_j \mathbf{A}_j$ ($p \times p$), $\mathbf{G} = (\mathbf{b}_1, \ldots, \mathbf{b}_J)$ ($p \times J$), and $\mathbf{D}_H = \diag(a_j + \tau)$ ($J \times J$).

The $\mathbf{A}$-block of $\Hobs$ has no prior curvature because
the prior on $\bbeta$ is flat \textup{[uses (R3)]}, while the
random-effect block $\mathbf{D}_H$ adds the prior precision $\tau$ to
each diagonal entry, with $a_j, \tau > 0$ \textup{[uses (R1)]}.
In $\Jclust$, the absence of $\theta_j$--$\theta_k$ blocks for
$j \neq k$ reflects the group--PSU nesting \textup{[uses (R2)]}, and
the scalar within-group design-effect parameterization
($d_j = \DEFF_j\, a_j$ and the $\DEFF_j$-scaled coupling blocks) is
the definition of $\DEFF_j$ for the conjugate model in
Remark~\ref{rem:group-score-variance} of \Cref{app:conditions}.

\subsubsection{Proof}

\paragraph{Step 1: Inverting $\Hobs$.}
The Schur complement of $\mathbf{D}_H$ in $\Hobs$ gives the marginal precision for $\bbeta$
\textup{[uses (R1): $a_j + \tau > 0$, so $\mathbf{D}_H$ is invertible]}:
\begin{equation*}
  \mathbf{S}_\beta = \mathbf{A} - \mathbf{G} \mathbf{D}_H^{-1} \mathbf{G}^\t
  = \sum_j \mathbf{A}_j - \sum_j \frac{\mathbf{b}_j \mathbf{b}_j^\t}{a_j + \tau}
  = \sum_j \mathbf{C}_j
\end{equation*}

\textbf{Decomposing $\mathbf{C}_j$}:
\begin{align*}
  \mathbf{C}_j &= \mathbf{A}_j - \frac{\mathbf{b}_j \mathbf{b}_j^\t}{a_j + \tau}
  = \pr{\mathbf{W}_j + \frac{\mathbf{b}_j \mathbf{b}_j^\t}{a_j}} - \frac{\mathbf{b}_j \mathbf{b}_j^\t}{a_j + \tau} \\
  &= \mathbf{W}_j + \mathbf{b}_j \mathbf{b}_j^\t \pr{\frac{1}{a_j} - \frac{1}{a_j + \tau}}
  = \mathbf{W}_j + \frac{\tau}{a_j(a_j + \tau)} \mathbf{b}_j \mathbf{b}_j^\t \\
  &= \mathbf{W}_j + \frac{(1-B_j)}{a_j} \mathbf{b}_j \mathbf{b}_j^\t
\end{align*}
using $\tau / [a_j(a_j + \tau)] = (1 - B_j)/a_j$.

By the standard block inversion formula
\textup{[uses (R5): $\mathbf{S}_\beta \succ 0$, so
$\mathbf{S}_\beta^{-1}$ exists]}:
\begin{equation*}
  [\Hobs^{-1}]_{\beta\beta} = \mathbf{S}_\beta^{-1}
\end{equation*}

The cross-block is $[\Hobs^{-1}]_{\beta, \theta_j} = -\mathbf{S}_\beta^{-1} \mathbf{b}_j / (a_j + \tau) = -(B_j / a_j) \mathbf{S}_\beta^{-1} \mathbf{b}_j$.

\paragraph{Step 2: Effective meat $\mathbf{M}$.}
\textup{[This step uses the block structure of $\Jclust$: (R2) for the
absence of cross-group blocks, and the within-group design-effect
parameterization of Remark~\ref{rem:group-score-variance}.]}
Write the first $p$ rows of $\Hobs^{-1}$ as $\mathbf{E} = (\mathbf{P}, \mathbf{Q}_1, \ldots, \mathbf{Q}_J)$, where $\mathbf{P} = \mathbf{S}_\beta^{-1}$ and $\mathbf{Q}_j = -(B_j/a_j) \mathbf{P} \mathbf{b}_j$.
Then:
\begin{equation*}
  [\Vsand]_{\beta\beta} = \mathbf{E} \, \Jclust \, \mathbf{E}^\t
\end{equation*}

Expanding using the block structure of $\Jclust$:
\begin{align*}
  [\Vsand]_{\beta\beta}
  &= \underbrace{\mathbf{P} \pr{\textstyle\sum_j \DEFF_j \mathbf{A}_j} \mathbf{P}}_{\text{Term 1}}
   + \underbrace{2 \sum_j \mathbf{P}(\DEFF_j \mathbf{b}_j) \mathbf{Q}_j^\t}_{\text{Terms 2+3}}
   + \underbrace{\sum_j d_j \mathbf{Q}_j \mathbf{Q}_j^\t}_{\text{Term 4}}
\end{align*}

Computing each term:
\begin{itemize}
  \item \textbf{Term 1}: $\mathbf{P} \pr{\sum_j \DEFF_j \mathbf{A}_j} \mathbf{P}$
  \item \textbf{Terms 2+3}: $-2 \sum_j \frac{\DEFF_j B_j}{a_j} \mathbf{P} \mathbf{b}_j \mathbf{b}_j^\t \mathbf{P}$
  \item \textbf{Term 4}: $\sum_j \frac{\DEFF_j B_j^2}{a_j} \mathbf{P} \mathbf{b}_j \mathbf{b}_j^\t \mathbf{P}$
\end{itemize}

Combining and factoring out $\mathbf{P} = \mathbf{S}_\beta^{-1}$:
\begin{equation*}
  [\Vsand]_{\beta\beta} = \mathbf{S}_\beta^{-1} \bk{\sum_j \DEFF_j \mathbf{A}_j - \sum_j \frac{\DEFF_j(2B_j - B_j^2)}{a_j} \mathbf{b}_j \mathbf{b}_j^\t} \mathbf{S}_\beta^{-1}
\end{equation*}

Using $2B_j - B_j^2 = 1 - (1 - B_j)^2$ and $\mathbf{A}_j = \mathbf{W}_j + \mathbf{b}_j \mathbf{b}_j^\t / a_j$:
\begin{equation*}
  \mathbf{M} = \sum_j \DEFF_j \mathbf{W}_j + \sum_j \DEFF_j \frac{(1 - B_j)^2}{a_j} \mathbf{b}_j \mathbf{b}_j^\t = \sum_j \DEFF_j \, \mathbf{C}_j^{(d)}
\end{equation*}

\paragraph{Step 3: The DER.}
The $\DER$ for $\beta_k$ is
\textup{[uses (R5): $[\mathbf{S}_\beta^{-1}]_{kk} > 0$]}:
\begin{equation*}
  \DER_{\beta_k} = \frac{[\mathbf{S}_\beta^{-1} \mathbf{M} \, \mathbf{S}_\beta^{-1}]_{kk}}{[\mathbf{S}_\beta^{-1}]_{kk}}
\end{equation*}

The $\DER$ reflects how the bread and effective meat weight between-group information:
\begin{itemize}
  \item Bread: $\mathbf{S}_\beta = \sum_j \mathbf{W}_j + \sum_j \frac{(1-B_j)}{a_j} \mathbf{b}_j \mathbf{b}_j^\t$ --- factor $(1-B_j)$ on between-group terms
  \item Effective meat: $\mathbf{M} = \sum_j \DEFF_j \mathbf{W}_j + \sum_j \DEFF_j \frac{(1-B_j)^2}{a_j} \mathbf{b}_j \mathbf{b}_j^\t$ --- factor $(1-B_j)^2$ on between-group terms
\end{itemize}
The extra $(1-B_j)$ factor in the meat attenuates the design sensitivity of between-group information.

\paragraph{Step 4: Common design effects.}
When $\DEFF_j = \DEFF$ for all $j$ \textup{[hypothesis of the
common-$\DEFF$ clause]}:
\begin{equation*}
  \mathbf{M} = \DEFF \bk{\mathbf{S}_\beta - \bDelta}
\end{equation*}
where $\bDelta = \sum_j \frac{\tau}{(a_j + \tau)^2} \mathbf{b}_j \mathbf{b}_j^\t = \sum_j \frac{B_j(1-B_j)}{a_j} \mathbf{b}_j \mathbf{b}_j^\t \succeq 0$
\textup{[uses (R1): $\tau > 0$ and $B_j \in (0,1)$]}.

Therefore:
\begin{equation*}
  [\Vsand]_{\beta\beta} = \DEFF \cdot \pr{\mathbf{S}_\beta^{-1} - \mathbf{S}_\beta^{-1} \bDelta \, \mathbf{S}_\beta^{-1}}
\end{equation*}
and
\begin{equation*}
  \DER_{\beta_k} = \DEFF \cdot \pr{1 - \frac{[\mathbf{S}_\beta^{-1} \bDelta \, \mathbf{S}_\beta^{-1}]_{kk}}{[\mathbf{S}_\beta^{-1}]_{kk}}} = \DEFF \cdot (1 - R_k). \qquad \square
\end{equation*}

\paragraph{Step 5: Bounds on $R_k$.}
\textbf{Lower bound}: Since $\bDelta \succeq 0$ \textup{[uses (R1), as
in Step~4]}, we have $[\mathbf{S}_\beta^{-1} \bDelta \, \mathbf{S}_\beta^{-1}]_{kk} \geq 0$, so $R_k \geq 0$.

\textbf{Upper bound}: In the balanced case \textup{[uses the
balanced-case hypothesis, (C1)]},
$\bDelta = J \cdot B(1-B)/a \cdot \mathbf{b}\mathbf{b}^\t$ and $\mathbf{S}_\beta \succeq J(1-B)/a \cdot \mathbf{b}\mathbf{b}^\t$ (from the between-group component).
The ratio of the between-group parts is $\bDelta / \mathbf{S}_\beta|_{\text{between}} = B$.
Since $\mathbf{W}_j$ contributes additional positive definite terms to $\mathbf{S}_\beta$ but not to $\bDelta$ (which only involves $\mathbf{b}_j\mathbf{b}_j^\t$), we have $R_k \leq B$.

Therefore: $\DEFF(1-B) \leq \DER_{\beta_k} \leq \DEFF$. $\square$

\subsubsection{Special Cases}
\label{sec:theorem1-special}

\textbf{Pure within-group covariate} ($\mathbf{b}_j = \mathbf{0}$ for all $j$):
Then $\bDelta = 0$, so $R_k = 0$ and $\DER_{\beta_k} = \DEFF$.

\textbf{Pure between-group/intercept} (balanced, centered covariates):
$\mathbf{W}_j = 0$ in the $k$-th direction, so $R_k = B$ and $\DER_{\beta_0} = \DEFF(1-B)$.

\textbf{Balanced $p = 2$ slope (strong balanced condition).}
For the slope coefficient $\beta_1$ in the model $y_{ij} = \beta_0 + \beta_1 x_{ij} + \theta_j + \varepsilon_{ij}$, we have $\DER_{\beta_1} = \DEFF$ provided the following strong balanced condition holds:

\begin{remark}[Strong balanced condition]
All groups share identical covariate summaries: $\mathbf{A}_j = \mathbf{A}_0$, $\mathbf{b}_j = \mathbf{b}_0$, $\mathbf{W}_j = \mathbf{W}_0$ for all $j = 1, \ldots, J$.
In particular, $a_j = a$, $B_j = B$, and $\bar{x}_j = \bar{x}$ for all $j$.
\end{remark}

\begin{proof}
Under the strong balanced condition with $p = 2$
\textup{[uses the strong balanced condition]}, write
$\mathbf{b}_j = \mathbf{b}_0 = (a,\; b)^\t$ where $b = \sum_{i \in j} \tilde{w}_{ij} x_{ij} / \sigma_e^2 = a \bar{x}$.
The rank-1 matrix $\mathbf{b}_j \mathbf{b}_j^\t$ is:
\begin{equation*}
  \mathbf{b}_0 \mathbf{b}_0^\t = a^2 \begin{pmatrix} 1 & \bar{x} \\ \bar{x} & \bar{x}^2 \end{pmatrix}
\end{equation*}

This matrix has rank~1, with its column space spanned by $(1, \bar{x})^\t$---the intercept direction.
The slope direction $\mathbf{e}_1 - \bar{x}\, \mathbf{e}_0$ (after orthogonalization against the intercept) lies in the null space of $\mathbf{b}_0 \mathbf{b}_0^\t$:
\begin{equation*}
  \mathbf{b}_0 \mathbf{b}_0^\t \begin{pmatrix} -\bar{x} \\ 1 \end{pmatrix}
  = a^2 \begin{pmatrix} 1 & \bar{x} \\ \bar{x} & \bar{x}^2 \end{pmatrix} \begin{pmatrix} -\bar{x} \\ 1 \end{pmatrix}
  = a^2 \begin{pmatrix} 0 \\ 0 \end{pmatrix}
\end{equation*}

Since all groups contribute identically,
\begin{align*}
  \bDelta
  &= J \cdot \frac{\tau}{(a+\tau)^2}
     \mathbf{b}_0 \mathbf{b}_0^\t, \\
  \mathbf{S}_\beta
  &= J\pr{\mathbf{W}_0 + \frac{(1-B)}{a}
     \mathbf{b}_0 \mathbf{b}_0^\t}.
\end{align*}
The correction $\mathbf{S}_\beta^{-1} \bDelta\, \mathbf{S}_\beta^{-1}$ is a rank-1 matrix in the intercept direction.
Its $(1,1)$ entry (the slope--slope element) vanishes:
\begin{equation*}
  [\mathbf{S}_\beta^{-1} \bDelta\, \mathbf{S}_\beta^{-1}]_{11} = 0
\end{equation*}

Therefore $R_1 = 0$ and $\DER_{\beta_1} = \DEFF$, regardless of $\bar{x}$ (i.e., regardless of the between-group fraction $\rho = b^2/(ac)$ of covariate variation).
The slope is invariant because the Schur complement projection automatically separates within-group information (which determines slopes) from between-group information (which is confounded with $\theta_j$).
\end{proof}

\textbf{Counterexample: unequal group means.}
If $\bar{x}_j$ varies across groups (violating the strong balanced condition), then $\sum_j \mathbf{b}_j \mathbf{b}_j^\t$ is generally rank~2, and the slope direction is no longer in the null space of $\bDelta$.
In this case, $R_1 > 0$ and $\DER_{\beta_1} < \DEFF$.

For example, with $J = 2$ groups where $\bar{x}_1 = 0$ and $\bar{x}_2 = 1$ (maximal between-group variation in the covariate), balanced otherwise:
\begin{equation*}
  \bDelta = \frac{\tau}{(a+\tau)^2}\bk{\mathbf{b}_1 \mathbf{b}_1^\t + \mathbf{b}_2 \mathbf{b}_2^\t}
\end{equation*}
where $\mathbf{b}_1 = a(1, 0)^\t$ and $\mathbf{b}_2 = a(1, 1)^\t$.
Then $\bDelta$ has rank~2 (full rank for $p = 2$), and neither the slope nor the intercept direction lies in its null space.
Consequently, $R_1 > 0$ and $\DER_{\beta_1} < \DEFF$.

\subsection{Proof of Theorem~2: Random Effects}
\label{sec:theorem2}

\subsubsection{Statement}

\textbf{Theorem 2} (DER for random effects; restating Theorem~2 in the main article)\textbf{.}
\textit{In the normal--normal conjugate model with parameter vector $(\mu, \theta_1, \ldots, \theta_J)$, under conditions (R1)--(R3) (\Cref{app:conditions}), the $\DER$ for the random effect $\theta_j$ is:}
\begin{equation*}
  \DER_j = \frac{B_j \bk{\DEFF_j (1 - B_j) \Sigma_{-j}^2 + B_j \sum_{k \neq j} \DEFF_k B_k (1 - B_k)}}{\Sigma_B \cdot \Gamma_j}
\end{equation*}
\textit{where $B_k = a_k/(a_k + \tau)$, $\Sigma_B = \sum_k B_k$, $\Sigma_{-j} = \Sigma_B - B_j$, and $\Gamma_j = B_j + \Sigma_{-j}(1 - B_j)$.}

\textit{The formula decomposes as $\DER_j = B_j \cdot \DEFF_j \cdot \kappa_j(J)$, where $\kappa_j(J)$ is the exact finite-$J$ coupling factor recorded in \eqref{eq:unbalanced-kappa} of \Cref{app:unbalanced}, satisfying $\kappa_j(J) \to 1$ as $J \to \infty$.  Note $\kappa_j(J)$ depends on the whole ensemble $\{(B_k, \DEFF_k)\}_{k=1}^{J}$, not on $(B_j, \DEFF_j, J)$ alone, unless the ensemble is homogeneous.}

\textit{In the balanced, common-$\DEFF$ case ($B_k = B$, $\DEFF_k = \DEFF$ for all $k$---the hypotheses of Theorem~2 in the main article):}
\begin{equation*}
  \DER_j = \DEFF \cdot B \cdot \frac{(J-1)(1-B)}{J(1-B) + B}
\end{equation*}

\paragraph{Conditions used.}
(R1)---Steps~1, 4, 5, and~6 (invertibility of $\mathbf{D}$;
$B_k \in (0,1)$; $\Sigma_B, \Gamma_j > 0$);
(R2)---Step~2 (rank-one structure of the meat);
(R3)---Step~1 (flat prior on $\mu$: no prior curvature in the $\mu$
entry).
Balance and the common-$\DEFF$ hypothesis are used only at Step~7
(balance from (C1) gives common $B$; common $\DEFF$ is the separate
hypothesis of Theorem~2 in the main article).
Steps~1--6 hold for arbitrary group-specific $(B_j, \DEFF_j, n_j)$,
as detailed in \Cref{app:unbalanced}.
Conditions (A1)--(A4): as in \Cref{sec:theorem1}---(A1) underlies
$a_j > 0$ via (R1); no algebraic step uses them directly.

\subsubsection{Proof}

\noindent\textit{Notational convention.}  As in the Theorem~1 proof,
$\Hobs$ denotes the full posterior Hessian $\Hobs^{\mathrm{post}}$
including prior precision $\tau = 1/\sigma^2_\theta$.

\smallskip
\paragraph{Step 1: Inverting $\Hobs$.}
$\Hobs$ has the bordered diagonal (arrow) structure:
\begin{equation*}
  \mathbf{H} = \begin{pmatrix} c & \mathbf{b}^\t \\ \mathbf{b} & \mathbf{D} \end{pmatrix}
\end{equation*}
where $c = \sum_j a_j$, $\mathbf{b} = (a_1, \ldots, a_J)^\t$, and $\mathbf{D} = \diag(a_j + \tau)$.
The $(\mu, \mu)$ entry $c$ carries no prior term because the prior on
$\mu$ is flat \textup{[uses (R3)]}.

The Schur complement of $\mathbf{D}$
\textup{[uses (R1): $a_k + \tau > 0$, so $\mathbf{D}$ is invertible,
and $\Sigma_B > 0$]}:
\begin{equation*}
  S = c - \mathbf{b}^\t \mathbf{D}^{-1} \mathbf{b}
  = \sum_k a_k - \sum_k \frac{a_k^2}{a_k + \tau}
  = \sum_k a_k(1 - B_k)
  = \tau \sum_k B_k
  = \tau \Sigma_B
\end{equation*}
using $[\mathbf{D}^{-1}\mathbf{b}]_k = a_k/(a_k + \tau) = B_k$.

By standard block inversion:

\begin{center}
\begin{tabular}{ll}
\hline
Entry & Formula \\
\hline
$[\mathbf{H}^{-1}]_{0,0}$ (for $\mu$) & $1/(\tau \Sigma_B)$ \\[4pt]
$[\mathbf{H}^{-1}]_{j+1,0}$ (cross $\theta_j$--$\mu$) & $-B_j/(\tau \Sigma_B)$ \\[4pt]
$[\mathbf{H}^{-1}]_{j+1,j+1}$ (diagonal for $\theta_j$) & $(1-B_j)/\tau + B_j^2/(\tau \Sigma_B) = \Gamma_j / (\tau \Sigma_B)$ \\[4pt]
$[\mathbf{H}^{-1}]_{j+1,k+1}$ ($j \neq k$) & $B_j B_k / (\tau \Sigma_B)$ \\[4pt]
\hline
\end{tabular}
\end{center}

\noindent
where $\Gamma_j = B_j + \Sigma_{-j}(1 - B_j)$.

In the balanced case ($B_k = B$ and $a_k = a$ for all $k$), $\Sigma_{-j} = (J-1)B$ and $\Gamma_j$ reduces to $\Gamma = B + (J-1)B(1-B) = B[1 + (J-1)(1-B)]$ for all $j$.
The denominator factor $\Sigma_B \cdot \Gamma_j$ in the $\DER$ formula then becomes $JB^2[1 + (J-1)(1-B)]$, leading to the balanced simplification $\DER_j = \DEFF \cdot B \cdot (J-1)(1-B)/[J(1-B) + B]$.

\paragraph{Step 2: Rank-1 decomposition of $\Jclust$.}
The meat matrix decomposes as a sum of rank-1 matrices
\textup{[uses (R2): each group's score total loads only on its own
$\theta_k$ entry and the $\mu$ entry, with no cross-group blocks]}:
\begin{equation*}
  \mathbf{J} = \sum_{k=1}^J d_k \, \mathbf{f}_k \mathbf{f}_k^\t
\end{equation*}
where $\mathbf{f}_k = \mathbf{e}_0 + \mathbf{e}_k$ (standard basis vectors in $\bR^{1+J}$).
This follows from the block structure: each group $k$ contributes $d_k$ to $[\mathbf{J}]_{0,0}$, $[\mathbf{J}]_{0,k+1}$, $[\mathbf{J}]_{k+1,0}$, and $[\mathbf{J}]_{k+1,k+1}$.

Therefore:
\begin{equation*}
  [\Vsand]_{j+1,j+1} = \mathbf{h}_j^\t \mathbf{J} \, \mathbf{h}_j = \sum_{k=1}^J d_k (g_j[0] + g_j[k])^2
\end{equation*}
where $\mathbf{h}_j$ is the $(j+1)$-th row of $\mathbf{H}^{-1}$ and $g_j[m] = \tau \cdot [\mathbf{H}^{-1}]_{j+1, m+1}$.

\paragraph{Step 3: Computing $g_j[0] + g_j[k]$.}
Substitute the entries from Step~1; no additional condition is needed.

\textbf{Case $k = j$}:
\begin{align*}
  g_j[0] + g_j[j] &= -\frac{B_j}{\Sigma_B} + (1 - B_j) + \frac{B_j^2}{\Sigma_B}
  = (1 - B_j)\pr{1 - \frac{B_j}{\Sigma_B}}
  = \frac{(1 - B_j)\Sigma_{-j}}{\Sigma_B}
\end{align*}

\textbf{Case $k \neq j$}:
\begin{equation*}
  g_j[0] + g_j[k] = -\frac{B_j}{\Sigma_B} + \frac{B_j B_k}{\Sigma_B} = -\frac{B_j(1 - B_k)}{\Sigma_B}
\end{equation*}

\paragraph{Step 4: Target variance $[\Vsand]_{j+1,j+1}$.}
\begin{align*}
  [\Vsand]_{j+1,j+1}
  &= \frac{1}{\tau^2} \bk{\frac{d_j(1-B_j)^2 \Sigma_{-j}^2}{\Sigma_B^2} + \frac{B_j^2 \sum_{k \neq j} d_k(1-B_k)^2}{\Sigma_B^2}}
\end{align*}

Using $d_k(1 - B_k)^2 = \DEFF_k \cdot a_k(1 - B_k)^2 = \DEFF_k \cdot \tau B_k(1 - B_k)$ (since $a_k = \tau B_k/(1-B_k)$)
\textup{[uses (R1): $B_k < 1$; and $d_k = \DEFF_k\, a_k$, the
within-group design-effect parameterization of
Remark~\ref{rem:group-score-variance}]}:
\begin{equation*}
  [\Vsand]_{j+1,j+1} = \frac{B_j}{\tau \Sigma_B^2} \bk{\DEFF_j(1-B_j)\Sigma_{-j}^2 + B_j \sum_{k \neq j} \DEFF_k B_k(1-B_k)}
\end{equation*}

\paragraph{Step 5: The DER.}
\textup{[Uses (R1): $\Sigma_B > 0$ and $\Gamma_j > 0$, so the ratio is
well defined.]}
\begin{equation*}
  \begin{aligned}
  \DER_j &= \frac{[\Vsand]_{j+1,j+1}}{[\mathbf{H}^{-1}]_{j+1,j+1}}\\
  &= \frac{B_j \bk{\DEFF_j(1-B_j)\Sigma_{-j}^2 + B_j \sum_{k \neq j} \DEFF_k B_k(1-B_k)}}{\Sigma_B \cdot \Gamma_j}
  \qquad \square
  \end{aligned}
\end{equation*}

\paragraph{Step 6: Conditional DER.}
Conditional on $\mu$ being known, the $\DER$ for $\theta_j$ reduces to the scalar problem
\textup{[uses (R1) and $d_j = \DEFF_j\, a_j$, as in Step~4]}:
\begin{equation*}
  \DER_j^{\text{cond}} = \frac{d_j/(a_j + \tau)^2}{1/(a_j + \tau)} = \frac{d_j}{a_j + \tau} = \frac{\DEFF_j \cdot a_j}{a_j + \tau} = B_j \cdot \DEFF_j
\end{equation*}

The exact formula can therefore be written as $\DER_j = B_j \cdot \DEFF_j \cdot \kappa_j(J)$.
The factor satisfies $\kappa_j(J) < 1$ under balance and common design
effects, but this inequality need not hold with heterogeneous design
effects.
It captures finite-$J$ coupling through uncertainty in $\mu$; its
exact ensemble form is \eqref{eq:unbalanced-kappa} in
\Cref{app:unbalanced}.

\Needspace{14\baselineskip}
\paragraph{Step 7: Balanced, common-$\DEFF$ case.}
Setting $B_k = B$, $\DEFF_k = \DEFF$ for all $k$
\textup{[uses balance from (C1), which gives common $B$, \emph{and}
the common-$\DEFF$ hypothesis of Theorem~2 in the main article; balance alone
does not make the $\DEFF_k$ common]}:
\begin{itemize}
  \item $\Sigma_B = JB$, $\Sigma_{-j} = (J-1)B$
  \item Numerator: $B[\DEFF(1-B)(J-1)^2 B^2 + B(J-1)\DEFF \cdot B(1-B)] = \DEFF \cdot B^3(1-B)J(J-1)$
  \item Denominator: $JB \cdot [B + (J-1)B(1-B)] = JB^2[J(1-B) + B]$
\end{itemize}
\begin{equation*}
  \DER_j = \DEFF \cdot B \cdot \frac{(J-1)(1-B)}{J(1-B) + B} = B \cdot \DEFF \cdot \kappa(J)
\end{equation*}
where $\kappa(J) = 1 - 1/[J(1-B) + B]$.
As $J \to \infty$, $\kappa(J) \to 1$ and $\DER_j \to B \cdot \DEFF$. $\square$

\subsection{Theorem~2 with Unequal Group Sizes and Design Effects}
\label{app:unbalanced}

Theorem~2 of the main article assumes balance (C1) and common
within-group design effects
($\DEFF_j \equiv \DEFF$), where the closed form
$\kappa(J) = 1 - 1/[J(1-B) + B]$ explains the coupling between groups.
We now give the exact form without either restriction, allowing
group-specific data precisions
$a_j = \widetilde{n}_j/\sigma^2_e$ (hence group-specific shrinkage
$B_j = a_j/(a_j + \tau)$ and sample sizes $n_j$) and group-specific
within-group design effects $\DEFF_j$.
We establish the result algebraically and compare it with direct
matrix calculations.

\paragraph{Exact unbalanced formula.}
In the normal--normal conjugate model of \Cref{sec:theorem2}, under
conditions (R1)--(R3) and \emph{without} any balance assumption, the
$\DER$ of $\theta_j$ is exactly
\begin{equation}\label{eq:unbalanced-der}
  \DER_j
  \;=\;
  \frac{B_j \bk{\DEFF_j (1 - B_j)\, \Sigma_{-j}^{2}
        \;+\; B_j \sum_{k \neq j} \DEFF_k\, B_k (1 - B_k)}}
       {\Sigma_B \cdot \Gamma_j},
\end{equation}
where $B_k = a_k/(a_k + \tau)$, $\Sigma_B = \sum_{k=1}^{J} B_k$,
$\Sigma_{-j} = \Sigma_B - B_j$, and
$\Gamma_j = B_j + \Sigma_{-j}(1 - B_j)$.
Equivalently, factoring out the conditional $\DER$ of Step~6,
\begin{equation}\label{eq:unbalanced-kappa}
  \begin{aligned}
  \DER_j &\;=\; B_j \cdot \DEFF_j \cdot \kappa_j(J),\\
  \kappa_j(J)
  &\;=\;
  \frac{(1 - B_j)\, \Sigma_{-j}^{2}
        \;+\; \dfrac{B_j}{\DEFF_j} \displaystyle\sum_{k \neq j}
        \DEFF_k\, B_k (1 - B_k)}
       {\Sigma_B \cdot \Gamma_j},
  \end{aligned}
\end{equation}
which defines the exact finite-$J$ coupling factor in the unbalanced
case.

\paragraph{Reductions.}
When all groups share a common design effect
($\DEFF_k = \DEFF_j$ for all $k$), the coupling factor simplifies to
$\kappa_j(J) = \bigl[(1-B_j)\Sigma_{-j}^2 + B_j \sum_{k \neq j}
B_k(1-B_k)\bigr] / (\Sigma_B \Gamma_j)$, which depends on the design
only through the shrinkage profile $\{B_k\}$.
When in addition the shrinkage is common ($B_k = B_j$ for all $k$),
substituting $\Sigma_{-j} = (J-1)B_j$ and
$\Gamma_j = B_j[1 + (J-1)(1-B_j)]$ recovers
\begin{equation*}
  \kappa_j(J)
  = \frac{(J-1)(1-B_j)}{J(1-B_j) + B_j}
  = 1 - \frac{1}{J(1-B_j) + B_j},
\end{equation*}
the form displayed in Theorem~2 of the main article.
When either shrinkage or design effects vary across groups, the exact
$\kappa_j(J)$ of \eqref{eq:unbalanced-kappa} couples group~$j$
to the entire ensemble through $\Sigma_{-j}$ and the
$\sum_{k \neq j}$ term, and the naive per-group product
$B_j \DEFF_j \kappa_j^{\mathrm{bal}}(J)$ (the balanced-form
expression evaluated at group-specific arguments) is not exact.

\paragraph{Example: heterogeneous $\DEFF$, common $B$.}
Balance makes shrinkage common but does not constrain the design
effects. Heterogeneous $\DEFF$ alone can therefore invalidate the
balanced closed form. Take $J = 3$,
$B_k \equiv 0.5$, and $(\DEFF_1, \DEFF_2, \DEFF_3) = (1, 9, 1)$.
Then \eqref{eq:unbalanced-der} gives, for group~1,
$\Sigma_B = 1.5$, $\Sigma_{-1} = 1$, $\Gamma_1 = 1$, and
\begin{equation*}
  \DER_1
  = \frac{0.5\,[\,1 \cdot 0.5 \cdot 1 + 0.5\,(9 \cdot 0.25 + 1 \cdot
    0.25)\,]}{1.5 \cdot 1}
  = \frac{0.5 \cdot 1.75}{1.5}
  = 0.583,
\end{equation*}
which agrees with direct matrix computation. The naive product
$B_1 \DEFF_1 \kappa(3) = 0.5 \cdot 1 \cdot 0.5 = 0.25$
understates it by a factor $2.3$: group~1's marginal $\DER$ inherits
group~2's large design effect through the grand-mean coupling.
This is why the common-$\DEFF$ hypothesis appears explicitly in
Theorems~1--2 in the main article.
All computations in the paper and in the \texttt{svyder} package
evaluate the exact matrix formula (equivalently
\eqref{eq:unbalanced-der}), rather than substituting group-specific
values into the balanced formula.

\paragraph{Derivation.}
Steps~1--6 of the Theorem~2 proof in \Cref{sec:theorem2} allow
arbitrary group-specific $(a_j, \DEFF_j)$; balance enters only at
Step~7. Thus, \eqref{eq:unbalanced-der} follows directly from that
proof, with the per-group quantities written explicitly.
The corollary describing qualitative properties of the unbalanced
formula (non-negativity, the $B_j = 0$ and $J = 1$ degeneracies, and
cross-group dependence) is given in \Cref{app:secondary}
(Corollary~\ref{cor:unbalanced}).

\paragraph{Numerical comparison.}
We compared \eqref{eq:unbalanced-der} with direct construction and
inversion of the full $(1+J) \times (1+J)$ matrices $\Hobs$ and
$\Jclust$ in seven sets of cases: balanced designs with
$J \in \{2, 5, 10, 20, 50, 100, 500\}$; an unbalanced $J = 5$ design
with heterogeneous $(a_j, \DEFF_j)$; the boundary cases $B \to 0$,
$B \to 1$, and $\DEFF = 1$ (SRS); the smallest non-trivial case
$J = 2$; and a randomly generated unbalanced $J = 50$ design.
All seven comparisons agreed, with a maximum absolute discrepancy of
$5.2 \times 10^{-12}$ overall (attained in the near-degenerate
$B \to 1$ case) and at most $1.4 \times 10^{-13}$ elsewhere.
A 75-point $(B \times \DEFF \times J)$ grid comparison also agreed.
The reproducibility materials referenced in the main article's
software statement include the script
\texttt{015\_verify\_theorem2.R}.

%% file: sm_c_secondary.tex

\section{Secondary Results}
\label{app:secondary}

We give further properties of the random-effect $\DER$ (Theorem~2 of
the main article), the marginal conservation identity, and corollaries
supporting the four-tier classification (Table~1 of the main article).
We then extend the diagnostic to multivariate random effects, explain
the Tier~III exclusion and reparameterization, and compare the
diagnostic with other design-sensitivity tools.

\subsection{Boundary Behavior and Shape of the Random-Effect DER}
\label{sec:boundary}

The following corollary describes the random-effect $\DER$ in the
balanced conjugate model. Theorem~2 in the main article refers to it
as Corollary~S-C.1.

\begin{corollary}[Boundary behavior and shape]
\label{cor:boundary}
In the balanced normal--normal conjugate model used in Theorem~2 of the main article, with $J \geq 2$, $B \in (0,1)$ the common shrinkage
factor, and $\DEFF$ the common within-group design effect:
\begin{enumerate}[label=(\roman*)]
  \item \textit{(Complete shielding.)} $\lim_{B \to 0} \DER_j = 0$:
    under complete shrinkage, the group effects have no design
    sensitivity.
  \item \textit{(Exposure at the other boundary.)} As $B \to 1$, the
    conditional $\DER$ tends to the full design effect,
    $\DER_j^{\mathrm{cond}} = B \cdot \DEFF \to \DEFF$, while the
    marginal $\DER_j \to 0$ for finite $J$---consistent with the
    marginal form of the conservation identity
    (Corollary~\ref{cor:conservation-marginal}), whose total also
    vanishes in this limit.
  \item \textit{(Monotonicity in the design effect.)} $\DER_j$ is
    strictly increasing, and linear, in $\DEFF$.
  \item \textit{(Inverted U in shrinkage.)} $\DER_j$ is non-monotone
    in $B$, with a unique interior maximum at
    $B^* = \sqrt{J}/(\sqrt{J} + 1)$ and maximum value
    $\DER^{\max} = \DEFF \cdot (J-1)/(\sqrt{J}+1)^2
    = \DEFF \cdot (\sqrt{J}-1)/(\sqrt{J}+1)$.
  \item \textit{(SRS baseline.)} Under simple random sampling
    ($\DEFF = 1$):
    $\DER_j = B\, \kappa(J)
    = B(J-1)(1-B)/[J(1-B)+B] < 1$.
\end{enumerate}
\end{corollary}

\subsubsection{Proofs}

\paragraph{(i) $B \to 0$.}
As $B \to 0$, the factor $B$ in $B \cdot \DEFF \cdot \kappa(J)$ tends to $0$, while $\kappa(J) \to (J-1)/J$ remains bounded.
Therefore $\DER_j \to 0$.

When $B \to 0$, the prior $\pi(\theta_j \mid \sigma_\theta^2)$ determines the posterior, so the survey design has no influence on $\theta_j$. $\square$

\paragraph{(ii) $B \to 1$.}
Part~(a) of Theorem~2 in the main article gives the conditional limit: $\DER_j^{\mathrm{cond}} = B \cdot \DEFF \to \DEFF$.
For the marginal limit, as $B \to 1$, $(1-B) \to 0$ makes the numerator factor $(J-1)(1-B)$ vanish, while the denominator $J(1-B) + B \to 1$.
Thus,
\begin{equation*}
  \DER_j \to \DEFF \cdot 1 \cdot \frac{(J-1) \cdot 0}{0 + 1} = 0
\end{equation*}

When $B \to 1$ uniformly ($\tau \to 0$), the random effects become unpenalized.
In the joint parameterization $(\mu, \theta_1, \ldots, \theta_J)$, $\theta_j$ and $\mu$ become non-identifiable, causing the model-based posterior variance to grow faster than the design-based variance.
The marginal conservation identity
(Corollary~\ref{cor:conservation-marginal}(ii)) gives the same limit:
$\DER_\mu + \DER_\theta = \DEFF \cdot J(1-B)/[J(1-B)+B] \to 0$ in this
limit. $\square$

\paragraph{(iii) Linearity in DEFF.}
Define $C(B, J) = B(J-1)(1-B)/[J(1-B) + B]$.
For $B \in (0,1)$ and $J \geq 2$, $C(B,J) > 0$, so $\DER_j = \DEFF \cdot C(B,J)$ is linear in $\DEFF$ with positive slope. $\square$

\paragraph{(iv) Non-monotonicity and maximum location.}
Define $f(B) = B(1-B)/[J(1-B) + B]$.
The quotient rule, with $u = B(1-B)$ and $v = J - (J-1)B$, gives
\begin{equation*}
  f'(B) = \frac{(1-2B)[J - (J-1)B] + (J-1)B(1-B)}{[J-(J-1)B]^2}
\end{equation*}

The numerator expands to
\begin{equation*}
  N = J - 2BJ + (J-1)B^2
\end{equation*}

Setting $N = 0$ gives $(J-1)B^2 - 2JB + J = 0$.
The quadratic formula then gives
\begin{equation*}
  B = \frac{2J \pm 2\sqrt{J}}{2(J-1)} = \frac{J \pm \sqrt{J}}{J-1} = \frac{\sqrt{J}}{\sqrt{J} \mp 1}
\end{equation*}

The root in $(0,1)$ is $B^* = \sqrt{J}/(\sqrt{J} + 1)$.
Since $f(0) = f(1) = 0$ and $f > 0$ on $(0,1)$, this unique critical point is a maximum.

To obtain the maximum, we substitute $B^* = \sqrt{J}/(\sqrt{J}+1)$:
\begin{itemize}
  \item $1 - B^* = 1/(\sqrt{J}+1)$
  \item $J(1-B^*) + B^* = J/(\sqrt{J}+1) + \sqrt{J}/(\sqrt{J}+1) = \sqrt{J}$
  \item $B^*(1-B^*) = \sqrt{J}/(\sqrt{J}+1)^2$
\end{itemize}
\begin{equation*}
  \DER^{\max} = \DEFF \cdot (J-1) \cdot \frac{\sqrt{J}/(\sqrt{J}+1)^2}{\sqrt{J}}
  = \DEFF \cdot \frac{J-1}{(\sqrt{J}+1)^2}
  = \DEFF \cdot \frac{\sqrt{J}-1}{\sqrt{J}+1}
\end{equation*}

As $J \to \infty$, $B^* \to 1$ and $\DER^{\max} \to \DEFF$. $\square$

\paragraph{(v) SRS special case and the DER baseline.}
Under SRS ($\DEFF = 1$), $\DER_j = B(J-1)(1-B)/[J(1-B)+B] < 1$ for $B \in (0,1)$;
this equals $B\,\kappa(J)$ with
$\kappa(J) = 1 - 1/[J(1-B)+B] = (J-1)(1-B)/[J(1-B)+B]$.
In the large-$J$ limit, $\DER_j^{\text{SRS}} \to B$.

\textbf{Why $\DER_j \neq 1$ under SRS.}
The classical design effect compares the complex-design variance with the SRS variance, so $\DEFF = 1$ under SRS.
This does not imply $\DER = 1$: the $\DER$ compares the design-corrected sandwich variance with the model-based posterior variance.
Its denominator includes the prior's contribution.

The two sandwich components treat the prior differently:
\begin{itemize}
  \item The bread $\Hobs$ is the negative Hessian of the full log pseudo-posterior, including the prior: $\Hobs = \mathbf{H}_{\text{lik}} + \mathbf{H}_{\text{prior}}$.
  \item The meat $\Jclust$ uses only pseudo-likelihood scores. The prior is a fixed function of the parameters with zero sampling variance, so it contributes no scores.
\end{itemize}

The posterior mode solves the penalized estimating equation $\sum_i \mathbf{s}_i(\bm{\phi}) \allowbreak + \nabla \log \pi(\bm{\phi}) \allowbreak = 0$.
The prior gradient $\nabla \log \pi$ is a fixed function of $\bm{\phi}$ across repeated samples from the design.
The estimator's design-based variability therefore comes from the data scores $\sum_i \mathbf{s}_i$, whose variability $\Jclust$ estimates.
The curvature $\Hobs$ also includes the prior.

Under SRS, the design-based variance of the scores matches the model-based (likelihood-only) variance: $\Jclust = \mathbf{H}_{\text{lik}}$.
For random effects, the posterior variance is $\Hobs^{-1} = (\mathbf{H}_{\text{lik}} + \mathbf{H}_{\text{prior}})^{-1} < \mathbf{H}_{\text{lik}}^{-1}$ because the prior adds precision.
Hence,
the $\DER$ is
\begin{equation*}
  \DER_j^{\text{SRS}} = \frac{[\Hobs^{-1} \mathbf{H}_{\text{lik}} \, \Hobs^{-1}]_{jj}}{[\Hobs^{-1}]_{jj}} < 1
\end{equation*}
because the sandwich numerator propagates data variability, while the denominator includes the prior's contribution to posterior precision.

The $\DER$ for random effects thus measures the design-inflated fraction of data-driven posterior variance.
Under SRS, $\DER = B$ in the large-$J$ limit: the fraction $B$ of posterior precision comes from data, and the remaining fraction $(1 - B)$ comes from the prior and is unaffected by the survey design.

For fixed effects with flat priors, $\mathbf{H}_{\text{prior}}^{\beta} \approx 0$, so $\Hobs \approx \mathbf{H}_{\text{lik}}$ and $\DER_{\beta}^{\text{SRS}} = 1$, recovering the \citet{Binder1983} identity. $\square$

\subsection{The Conservation Identity: Marginal Form and
  Finite-\texorpdfstring{$J$}{J} Deficit}
\label{sec:conservation-marginal}

The conditional conservation identity
$\DER_\mu + \DER_\theta^{\mathrm{cond}} = \DEFF$ is
Corollary~1, Equation~(3.5), of the main article. It follows from
Theorem~1 and part~(a) of Theorem~2.
We now give its marginal counterpart, where finite-$J$ coupling
through the grand mean produces a deficit.

\begin{corollary}[Marginal conservation and its finite-$J$ deficit]
\label{cor:conservation-marginal}
In the balanced, common-$\DEFF$ normal--normal model with centered
covariates:
\begin{enumerate}[label=(\roman*)]
  \item \textit{(Conditional conservation---exact; Corollary~1 of the main article):}
  \begin{equation*}
    \DER_\mu + \DER_\theta^{\text{cond}} = \DEFF(1-B) + B \cdot \DEFF = \DEFF
  \end{equation*}

  \item \textit{(Marginal finite-$J$ deficit): For finite $J \geq 2$,}
  \begin{equation*}
    \DER_\mu + \DER_\theta = \DEFF \cdot \frac{J(1-B)}{J(1-B)+B} < \DEFF
  \end{equation*}

  \item \textit{The deficit is:}
  \begin{equation*}
    \DEFF - [\DER_\mu + \DER_\theta] = \frac{\DEFF \cdot B}{J(1-B)+B} = O(1/J)
  \end{equation*}
\end{enumerate}
\end{corollary}

\begin{proof}
\textit{(i) Conditional conservation (exact).}
For any $J$, we use the conditional $\DER$ (\Cref{sec:theorem2}, Step~6): $\DER_\mu = \DEFF(1-B)$ follows from Theorem~1 of the main article for the intercept-only model, and $\DER_\theta^{\text{cond}} = B \cdot \DEFF$ is the conditional $\DER$, exact for all $J$.
Adding these terms gives
\begin{equation*}
  \DEFF(1-B) + B \cdot \DEFF = \DEFF[(1-B) + B] = \DEFF.
\end{equation*}

\textit{(ii) Marginal finite-$J$ deficit.}
The full matrix gives a marginal $\DER_\theta$ smaller than the conditional $\DER_\theta^{\text{cond}} = B \cdot \DEFF$ by the coupling factor $\kappa(J) < 1$.
The resulting sum $\DER_\mu + \DER_\theta$ falls short of $\DEFF$ by an amount of order $O(1/J)$.

The exact finite-$J$ formula for $\DER_\theta$ gives
\begin{align*}
  \DER_\mu + \DER_\theta
  &= \DEFF(1-B) + \DEFF \cdot B \cdot \frac{(J-1)(1-B)}{J(1-B)+B} \\
  &= \DEFF(1-B) \bk{1 + \frac{B(J-1)}{J(1-B)+B}} \\
  &= \DEFF(1-B) \cdot \frac{J(1-B)+B + B(J-1)}{J(1-B)+B} \\
  &= \DEFF(1-B) \cdot \frac{J(1-B) + BJ}{J(1-B)+B} \\
  &= \DEFF(1-B) \cdot \frac{J}{J(1-B)+B} \\
  &= \DEFF \cdot \frac{J(1-B)}{J(1-B)+B}.
\end{align*}

\textit{(iii) Deficit.}
\begin{equation*}
  \DEFF - \DEFF \cdot \frac{J(1-B)}{J(1-B)+B} = \DEFF \cdot \frac{B}{J(1-B)+B} > 0
\end{equation*}
for $B > 0$.
For large $J$, $B/[J(1-B)+B] \sim B/[(1-B)J] = O(1/J)$.

The deficit arises from the posterior correlation between $\mu$ and $\theta_j$.
Uncertainty in $\mu$ absorbs some design effects that would otherwise appear in $\theta_j$'s $\DER$.
The deficit vanishes as $J \to \infty$, when $\mu$ becomes perfectly determined.
\end{proof}

\subsection{Additional Corollaries}
\label{sec:additional-corollaries}

The following corollaries derive from Theorems~1 and~2 of the main
article and support its four-tier classification in Table~1.
Throughout this subsection, ``the balanced case'' means the full
hypotheses of those theorems: balance (hence common shrinkage
$B$) \emph{and} a common within-group design effect $\DEFF$.
Balance alone does not imply common $\DEFF$; under a heterogeneous
ensemble the scalar forms below are the balanced/common-$\DEFF$
approximation, and the ensemble-exact statements are
Section~S-B.3 and Corollary~\ref{cor:unbalanced}.

\subsubsection{Bounds on the random-effect DER}

\begin{corollary}[Upper bound]
\label{cor:upper-bound}
In the balanced, common-$\DEFF$ case with $J \geq 2$:
\begin{enumerate}[label=(\roman*)]
  \item $0 \leq \DER_j \leq B \cdot \DEFF$, with equality in the upper bound as $J \to \infty$.
  \item The finite-$J$ relative gap is exactly:
  \begin{equation*}
    \frac{B \cdot \DEFF - \DER_j}{B \cdot \DEFF} = \frac{1}{J(1-B)+B}
  \end{equation*}
\end{enumerate}
\end{corollary}

\begin{proof}
We use $\DER_j = B \cdot \DEFF \cdot \kappa(J)$ with $\kappa(J) = 1 - 1/[J(1-B)+B]$. Since $J(1-B)+B > 1$ for $J \geq 2$ and $B \in (0,1)$, we have $0 < \kappa(J) < 1$.
\end{proof}

\subsubsection{Design sensitivity ordering}

\begin{corollary}[Ordering]
\label{cor:ordering}
In the balanced case with common $\DEFF$ and centered covariates, the intercept and random-effect DERs have the following large-$J$ ordering:
\begin{itemize}
  \item $B < 1/2$: $\DER_\theta < \DER_\mu$ (random effects less design-sensitive).
  \item $B = 1/2$: $\DER_\theta = \DER_\mu = \DEFF/2$ (equal sensitivity).
  \item $B > 1/2$: $\DER_\theta > \DER_\mu$ (random effects more design-sensitive).
\end{itemize}
Within-group slopes always satisfy $\DER_{\beta_1} = \DEFF \geq \DER_\theta$ (slopes are the most design-sensitive parameters).
\end{corollary}

\begin{proof}
We compare $\DER_\mu = \DEFF(1-B)$ with $\DER_\theta \approx B \cdot \DEFF$. Thus $\DER_\theta > \DER_\mu$ if and only if $B > 1-B$, or $B > 1/2$. For slopes, $\DEFF \geq B \cdot \DEFF$ since $B \leq 1$.
\end{proof}

\subsubsection{Monotonicity in \texorpdfstring{$J$}{J}}

\begin{corollary}
\label{cor:monotonicity}
For fixed $B \in (0,1)$ and $\DEFF > 0$, the random-effect DER is strictly increasing in $J$.
\end{corollary}

\begin{proof}
Differentiating $g(J) = (J-1)/[J(1-B)+B]$ gives
\begin{equation*}
  g'(J) = \frac{[J(1-B)+B] - (J-1)(1-B)}{[J(1-B)+B]^2} = \frac{1}{[J(1-B)+B]^2} > 0. \qedhere
\end{equation*}
\end{proof}

\subsubsection{Convergence rate to the large-\texorpdfstring{$J$}{J} limit}

\begin{corollary}
\label{cor:convergence-rate}
The coupling factor satisfies:
\begin{equation*}
  1 - \kappa(J) = \frac{1}{J(1-B)+B} \sim \frac{1}{(1-B)J} \quad \text{as } J \to \infty
\end{equation*}
The convergence rate is $O(1/J)$, faster for small $B$ (strong shrinkage).
\end{corollary}

\begin{proof}
This follows from $\kappa(J) = 1 - 1/[J(1-B)+B]$ and $J(1-B)+B \approx (1-B)J$ for large $J$.
\end{proof}

\subsubsection{Effective sample size interpretation}

\begin{corollary}[Effective sample size]
\label{cor:ess}
For a fixed-effect parameter $\beta_k$ with $\DER_{\beta_k} = d$, the design-relevant effective sample size is $N_{\mathrm{eff}}^{\mathrm{design}}(\beta_k) = N/d$.

For a random effect $\theta_j$, using the balanced/common-$\DEFF$ approximation $\DER_j \approx B_j \cdot \DEFF_j$ at large $J$ (the exact ensemble form is \Cref{app:unbalanced}), the design-adjusted effective sample size for the data-driven component is $n_j / \DEFF_j$.

PSIS effective sample size \citep{vehtari2024pareto} instead measures computational efficiency, not survey design efficiency.
\end{corollary}

\subsubsection{Threshold conditions}

\begin{corollary}[Parameter-type-specific thresholds]
\label{cor:thresholds}
In the balanced, common-$\DEFF$ case, for a flagging threshold $c_0$
(Algorithm~1 of the main article):
\begin{itemize}
  \item \textbf{Fixed effects (within-group)}: We correct when $\DEFF > c_0$.
  \item \textbf{Fixed effects (between-group)}: We correct when $\DEFF > c_0/(1-B)$, a higher threshold.
  \item \textbf{Random effects (large-$J$)}: We correct when $B \cdot \DEFF > c_0$.
  \item \textbf{Random effects (finite-$J$)}: We correct when $B \cdot \DEFF > c_0/\kappa(J)$.
\end{itemize}
Under heterogeneous designs these closed-form trigger points are
heuristics; we always use $\DER_k > c_0$ with the
$\DER$ computed from the exact matrix formula.
\end{corollary}

\subsubsection{Unbalanced case properties}

\begin{corollary}
\label{cor:unbalanced}
For the general (unbalanced) formula of \Cref{app:unbalanced},
eq.~\eqref{eq:unbalanced-der}:
\begin{enumerate}[label=(\roman*)]
  \item $\DER_j \geq 0$ for all $j$.
  \item If $B_j = 0$: $\DER_j = 0$.
  \item If $J = 1$: $\DER_1 = 0$ ($\theta_1$ confounded with $\mu$).
  \item The DER for group $j$ depends on the design effects and shrinkage of all other groups through $\sum_{k \neq j} \DEFF_k B_k(1-B_k)$.
\end{enumerate}
\end{corollary}

\begin{proof}
(i) Both numerator terms are non-negative. (ii) $B_j = 0$ makes the leading factor zero. (iii) $\Sigma_{-1} = 0$ and the sum over $k \neq j$ is empty. (iv) The sum in the formula gives the dependence on other groups.
\end{proof}

\subsection{Multivariate Matrix DER}
\label{sec:proposition3}

\begin{proposition}[Multivariate matrix DER]
\label{prop:matrix-der}
Consider a hierarchical model with $q$-dimensional random effects $\btheta_j \sim N_q(\mathbf{0}, \bSigma_\theta)$ under conditions (R1)--(R3).

\textbf{(a) Conditional matrix DER (exact).} For group $j$,
\begin{equation*}
  \mathbf{DER}_j^{\mathrm{cond}} = \tilde{\mathbf{A}}_j^{-1} \tilde{\mathbf{J}}_j = \mathbf{B}_j \cdot \mathbf{DEFF}_j
\end{equation*}
where
\begin{itemize}
  \item $\mathbf{B}_j = \bSigma_\theta(\bSigma_\theta + \mathbf{V}_j)^{-1}$ is the matrix shrinkage factor with eigenvalues in $(0,1)$
  \item $\mathbf{DEFF}_j = \tilde{\mathbf{H}}_{j,\mathrm{lik}}^{-1} \tilde{\mathbf{J}}_j$ is the matrix design effect
  \item $\tilde{\mathbf{A}}_j = \tilde{\mathbf{H}}_{j,\mathrm{lik}} + \bSigma_\theta^{-1}$ and $\mathbf{V}_j = \tilde{\mathbf{H}}_{j,\mathrm{lik}}^{-1}$
\end{itemize}

\textbf{(b) Spectral decomposition.} In the canonical multivariate normal model with scalar $\DEFF_j$, we express the matrix DER in the eigenbasis of $\bSigma_\theta$:
\begin{equation*}
  \boldsymbol{\mathcal{D}}_j = \mathbf{P} \, \diag(\delta_{j,1}, \ldots, \delta_{j,q}) \, \mathbf{P}^\t
\end{equation*}
where each $\delta_{j,r}$ follows the general scalar formula in \Cref{app:unbalanced} with direction-specific shrinkage $B_{j,r} = a_j \lambda_r / (a_j \lambda_r + 1)$.

\textbf{(c) Large-$J$ limit.} As $J \to \infty$, $\boldsymbol{\mathcal{D}}_j^{\mathrm{cond}} = \DEFF_j \cdot \mathbf{B}_j$ with eigenvalues $\DEFF_j \cdot B_{j,r}$.

\textbf{(d) Scalar reduction.} Setting $q = 1$ recovers $\DER_j^{\mathrm{cond}} = B_j \cdot \DEFF_j$ (Theorem~2 of the main article).
\end{proposition}

\subsubsection{Proof}

\paragraph{Conditional factorization.}
The conditional DER matrix is
\begin{equation*}
  \mathbf{DER}_j^{\mathrm{cond}} = (\tilde{\mathbf{V}}_j)(\tilde{\mathbf{A}}_j^{-1})^{-1} = \tilde{\mathbf{A}}_j^{-1} \tilde{\mathbf{J}}_j
\end{equation*}

Inserting $\tilde{\mathbf{H}}_{j,\mathrm{lik}} \tilde{\mathbf{H}}_{j,\mathrm{lik}}^{-1} = \mathbf{I}_q$ gives
\begin{equation*}
  \tilde{\mathbf{A}}_j^{-1} \tilde{\mathbf{J}}_j = \underbrace{\tilde{\mathbf{A}}_j^{-1} \tilde{\mathbf{H}}_{j,\mathrm{lik}}}_{\mathbf{B}_j} \cdot \underbrace{\tilde{\mathbf{H}}_{j,\mathrm{lik}}^{-1} \tilde{\mathbf{J}}_j}_{\mathbf{DEFF}_j}
\end{equation*}

The factorization is exact. \qed

\paragraph{Matrix shrinkage identity.}
We use the identity
\begin{equation*}
  \mathbf{B}_j = (\mathbf{I}_q + \mathbf{V}_j \bSigma_\theta^{-1})^{-1} = \bSigma_\theta(\bSigma_\theta + \mathbf{V}_j)^{-1}
\end{equation*}

The matrices $\bSigma_\theta$ and $\mathbf{V}_j$ need not commute, since
\begin{equation*}
  \mathbf{I}_q + \mathbf{V}_j \bSigma_\theta^{-1} = (\bSigma_\theta + \mathbf{V}_j)\bSigma_\theta^{-1}
\end{equation*}
Taking inverses gives $[(\bSigma_\theta + \mathbf{V}_j)\bSigma_\theta^{-1}]^{-1} = \bSigma_\theta(\bSigma_\theta + \mathbf{V}_j)^{-1}$ by $(\mathbf{P}\mathbf{Q})^{-1} = \mathbf{Q}^{-1}\mathbf{P}^{-1}$. \qed

This identifies $\mathbf{B}_j$ with the EBLUP shrinkage matrix
$\bgamma_j$ of small area estimation \citep{RaoMolina2015} for all
models.

\paragraph{Spectral decomposition.}
For the canonical multivariate normal model ($\by_{ij} \sim N_q(\bmu + \btheta_j, \sigma_e^2 \mathbf{I}_q)$), all matrices share the eigenbasis of $\bSigma_\theta$.\newline Let $\bSigma_\theta = \mathbf{P} \, \diag(\lambda_1, \ldots, \lambda_q) \, \mathbf{P}^\t$ be the eigendecomposition.

In this basis, the $r$-th eigendirection has
\begin{itemize}
  \item Data precision: $a_j$ (isotropic)
  \item Prior precision: $1/\lambda_r$
  \item Shrinkage: $B_{j,r} = a_j \lambda_r / (a_j \lambda_r + 1)$
\end{itemize}

Each direction is an independent scalar problem. Its spectral DER $\delta_{j,r}$ follows the general scalar formula in \Cref{app:unbalanced} with parameters $(B_{j,r}, \Sigma_{B,r}, \Sigma_{-j,r}, \Gamma_{j,r})$.

\paragraph{Spectral ordering.} $\lambda_1 \geq \cdots \geq \lambda_q \Rightarrow B_{j,1} \geq \cdots \geq B_{j,q} \Rightarrow \delta_{j,1} \geq \cdots \geq \delta_{j,q}$. Directions with larger prior variance are more design-sensitive.

The ordering $\delta_{j,1} \geq \cdots \geq \delta_{j,q}$ holds when all $B_{j,r}$ lie on the ascending branch of the inverted U ($B_{j,r} < B^* = \sqrt{J}/(\sqrt{J}+1)$), as is typical with moderate shrinkage and $J \geq 10$. If some $B_{j,r}$ exceeds $B^*$, the ordering can fail because the DER decreases with $B$ on the descending branch (Corollary~\ref{cor:boundary}(iv)).

\paragraph{Coordinate DER as weighted spectral average.} For the $r$-th coordinate in the original basis,
\begin{equation*}
  \DER_{j,r}^{\mathrm{coord}} = \sum_{s=1}^q w_{rs} \, \delta_{j,s}
\end{equation*}
with $w_{rs} \geq 0$, $\sum_s w_{rs} = 1$ (weights from the eigenvector matrix $\mathbf{P}$). The coordinate DER is bounded between $\delta_{j,\min}$ and $\delta_{j,\max}$ by the Rayleigh quotient.

\paragraph{Conditional-to-marginal gap.}
The Schur complement also gives the gap from conditioning on $\bbeta$ and $\btheta_{-j}$:
\begin{equation*}
  \|\mathbf{DER}_j^{\mathrm{marg}} - \mathbf{DER}_j^{\mathrm{cond}}\| / \|\mathbf{DER}_j^{\mathrm{cond}}\| = O(1/J)
\end{equation*}

The proof follows the scalar case, replacing scalars with $q \times q$ Schur complement blocks. \qed

\subsubsection{Practical Guidance}

\begin{table}[H]
\centering
\small
\begin{tabularx}{\textwidth}{@{}Xl@{}}
\hline
Scenario & Matrix structure needed? \\
\hline
Independent random intercepts ($q=1$) & No---use the scalar formula \\
Uncorrelated random effects ($\bSigma_\theta$ diagonal) & No---scalar DERs \\
Correlated random intercept + slopes & \textbf{Yes}---eigenvectors define principal directions \\
Large condition number $\kappa(\bSigma_\theta) > 5$ & \textbf{Yes}---spectral DERs may differ widely \\
\hline
\end{tabularx}
\end{table}

\subsection{Why the DER Excludes Hyperparameters (Tier III)}
\label{sec:remark1-elaboration}

We explain the exclusion of hyperparameters stated in Section~3.3 of
the main article.

\begin{remark}[DER inapplicability for hyperparameters]
\label{rem:hyperparameter}
The DER does not apply to hyperparameters $\boldsymbol{\psi}$ governing the random-effect distribution. Its sandwich variance estimator does not target inference for $\boldsymbol{\psi}$.
\end{remark}

\subsubsection{(a) Score function mismatch}

The hyperparameter $\sigma_\theta^2$ is identified from variation among $\{\theta_1, \ldots, \theta_J\}$. Its score is
\begin{equation*}
  \frac{\partial}{\partial \sigma_\theta^2} \log \pi(\theta_j \mid \sigma_\theta^2) = -\frac{1}{2\sigma_\theta^2} + \frac{\theta_j^2}{2\sigma_\theta^4}
\end{equation*}
and depends on the group-level quantity $\theta_j$. The cluster-robust meat $\Jclust$ aggregates observation-level pseudo-likelihood scores $\mathbf{s}_i = \tilde{w}_i \nabla_{\boldsymbol{\phi}} \log p(y_i \mid \bbeta, \theta_{j[i]})$, which contain no gradient with respect to $\boldsymbol{\psi}$. Our declared conditional parameter vector is $\boldsymbol{\phi} = (\bbeta^\t, \btheta^\t)^\t$; it excludes $\boldsymbol{\psi}$, so the DER does not assess hyperparameter uncertainty.

\subsubsection{(b) Effective sample size mismatch}

The effective sample size for $\bbeta$ and $\btheta$ is $O(N)$, the total number of observations, adjusted for design effects. For $\boldsymbol{\psi}$, it is $O(J)$, the number of groups: we estimate $\sigma_\theta^2$ from the dispersion of $J$ group effects. The sandwich accounts for within-PSU clustering of $N$ observations, but not the sampling variability among $J$ groups.

\subsubsection{(c) Prior contribution structure}

For inference on $\boldsymbol{\psi}$, the random-effect distribution $\pi(\theta_j \mid \boldsymbol{\psi})$ is both a prior on $\theta_j$ and a sampling model for $\boldsymbol{\psi}$. Our sandwich correction uses the data likelihood $p(y_i \mid \bbeta, \theta_j)$; it does not use the group-level model $\theta_j \mid \boldsymbol{\psi}$ as a likelihood for $\boldsymbol{\psi}$.

\subsubsection{(d) Appropriate alternative}

Hyperparameter uncertainty under informative sampling requires methods
beyond this diagnostic. Replicate-weight methods for the empirical
Bayes estimator of $\sigma_\theta^2$ offer one approach
\citep{SavitskyGershunskaya2026}.

\subsubsection{Hyperparameter gradient}

The log pseudo-posterior gradient for $\boldsymbol{\psi}$ is
\begin{equation*}
  \nabla_{\boldsymbol{\psi}} \ell^* = \underbrace{\nabla_{\boldsymbol{\psi}} \ell_{\mathrm{RE}}(\btheta \mid \boldsymbol{\psi})}_{\text{depends on } \theta_j, \text{ not } y_{ij}} + \underbrace{\nabla_{\boldsymbol{\psi}} \ell_{\mathrm{hyper}}(\boldsymbol{\psi})}_{\text{hyperprior}}
\end{equation*}

Neither term involves $\ell_{\mathrm{lik}}$ or contributes to observation-level scores. These scores therefore provide no direct contribution for $\boldsymbol{\psi}$.

\subsection{Reparameterization and DERs for Derived Quantities}
\label{sec:remark2-elaboration}

We give the delta-method extension summarized in Remark~4 and
Equation~(3.6) of the main article.

\begin{remark}[Reparameterization non-invariance and estimand-based DER]
\label{rem:reparameterization}
\textbf{(a) Non-invariance.} The diagonal DER (Definition~1 of the main article) is defined on the natural parameterization $\boldsymbol{\phi} = (\bbeta^\t, \theta_1, \ldots, \theta_J)^\t$. Under a smooth reparameterization $\boldsymbol{\psi} = h(\boldsymbol{\phi})$, the transformed quantities $\Vsand^{(\psi)} = \nabla h \cdot \Vsand \cdot \nabla h^\t$ and $\mathbf{H}^{-1,(\psi)} = \nabla h \cdot \mathbf{H}^{-1} \cdot \nabla h^\t$ yield a DER:
\begin{equation*}
  \DER_{\psi_k} = \frac{[\nabla h \cdot \Vsand \cdot \nabla h^\t]_{kk}}{[\nabla h \cdot \mathbf{H}^{-1} \cdot \nabla h^\t]_{kk}}
\end{equation*}
that generally differs from $\DER_{\phi_k}$. Extracting diagonal elements does not commute with congruence transformations: $\diag(\mathbf{P} \mathbf{A} \mathbf{P}^\t) \neq \mathbf{P}\, \diag(\mathbf{A})\, \mathbf{P}^\t$ for general $\mathbf{P}$.

\textbf{(b) Delta-method extension.} For a scalar derived quantity $\psi = g(\boldsymbol{\phi})$, we define
\begin{equation*}
  \DER(\psi) = \DER(g) = \frac{\nabla g(\hat{\boldsymbol{\phi}})^\t \, \Vsand \, \nabla g(\hat{\boldsymbol{\phi}})}{\nabla g(\hat{\boldsymbol{\phi}})^\t \, \mathbf{H}^{-1} \, \nabla g(\hat{\boldsymbol{\phi}})}
\end{equation*}
The ratio is well defined whenever $\nabla g^\t \mathbf{H}^{-1} \nabla g > 0$. It compares the design-based asymptotic variance of $g(\hat{\boldsymbol{\phi}})$ with its model-based variance.
This is eq.~(3.6) of the main article.

\textbf{(c) Recovery of diagonal DER.} The diagonal DER is the special case $g(\boldsymbol{\phi}) = \phi_p$, where $\nabla g = \mathbf{e}_p$ is the $p$-th standard basis vector.

\textbf{(d) Practical guidance.} We use the diagonal DER for the model parameters ($\bbeta$, $\btheta$) in the four-tier classification and the compute-classify-correct algorithm (Algorithm~1 of the main article). For derived quantities, such as predicted probabilities, area means, or contrasts, we use the delta-method DER.
\end{remark}

\subsubsection{Proof of Non-Invariance}

Let $\mathbf{P} = \nabla h$ be the Jacobian of the reparameterization. If $\mathbf{P}$ is not a generalized permutation matrix, $h$ mixes parameters, and the diagonal of $\mathbf{P} \mathbf{A} \mathbf{P}^\t$ depends on the off-diagonal elements of $\mathbf{A}$:
\begin{equation*}
  [\mathbf{P} \mathbf{A} \mathbf{P}^\t]_{kk} = \sum_{r,s} P_{kr}\, A_{rs}\, P_{ks}
\end{equation*}

The off-diagonal entries of $\Vsand$ and $\mathbf{H}^{-1}$ generally differ because the sandwich includes cross-terms from the meat. The ratio of the transformed $k$-th diagonal entries therefore differs from the original ratio. \qed

\subsubsection{When Non-Invariance Matters in Practice}

For the four-tier classification and the compute-classify-correct
algorithm (Algorithm~1 of the main article), the diagonal DER on the
natural parameterization is sufficient:
\begin{itemize}
  \item \textbf{Fixed effects} $\bbeta$: The regression coefficients are the natural inferential targets.
  \item \textbf{Random effects} $\btheta$: The group-level deviations are natural model parameters.
  \item \textbf{Hyperparameters} $\boldsymbol{\psi}$: DER is inapplicable (Remark~\ref{rem:hyperparameter}).
\end{itemize}

We use the delta-method extension for the following quantities and comparisons:
\begin{enumerate}
  \item \textbf{Predicted probabilities}: The first-order delta method gives the same ratio for $\DER(\logit p_j)$ and $\DER(\hat{p}_j)$, because the logistic derivative cancels. The DER for a linear predictor can nevertheless differ from those of its intercept and random-effect components.
  \item \textbf{Cross-model comparisons}: Comparing DER values across models that use different link functions or different parameterizations requires the estimand-based DER to ensure a common measurement scale.
  \item \textbf{Area means}: $\DER(\eta_j = \mu + \theta_j)$ is a distinct quantity from $\DER_\mu$ or $\DER_{\theta_j}$.
\end{enumerate}

We choose the diagonal or delta-method DER according to the inferential target.

\subsection{The DER Among Design-Sensitivity Tools}
\label{sec:tool-comparison}

\Cref{tab:comparison} compares the $\DER$ with design effects,
sandwich correction, and shrinkage factors. The $\DER$ combines these
tools in a per-parameter diagnostic.

\begin{table}[htbp]
  \centering
  \caption{Comparison of existing tools and the $\DER$ framework.}
  \label{tab:comparison}
  \small
  \begin{tabularx}{\textwidth}{@{}>{\raggedright\arraybackslash}p{2.8cm}
    >{\centering\arraybackslash}p{2.0cm}
    >{\centering\arraybackslash}p{2.6cm}
    >{\centering\arraybackslash}p{2.0cm}
    >{\centering\arraybackslash}X@{}}
    \toprule
    & \textbf{Kish $\DEFF$}
    & \textbf{Sandwich $+$ joint correction}
    & \textbf{Shrinkage $B_j$}
    & \textbf{$\DER$ (this paper)} \\
    \midrule
    Per-parameter?        & Yes  & No (blanket) & Yes  & Yes \\
    Accounts for priors?  & No   & No           & Yes  & Yes \\
    Accounts for design?  & Yes  & Yes          & No   & Yes \\
    Decomposition?        & ---  & ---          & ---
      & $\DEFF \times (1 - B)$\textsuperscript{a} \\
    Actionable threshold? & No   & ---          & No
      & Yes ($c_0$; four-tier) \\
    Cost                  & $O(d)$ & $O(d^3)$   & $O(J)$ & $O(d)$ \\
    \bottomrule
  \end{tabularx}
  \vspace{2pt}
  {\footnotesize \textsuperscript{a}Representative decomposition for
    between-group fixed effects, in the balanced, common-$\DEFF$ case
    of Theorems~1 and~2 of the main article.  Within-group: $\DER = \DEFF$;
    random effects: $\DER = B \cdot \DEFF \cdot \kappa(J)$.
    Heterogeneous-ensemble forms: Section~S-B.3.\par}
\end{table}

%% file: sm_b.tex

\section{Simulation Details}
\label{app:sim-additional}

We give the sampling, estimation, and diagnostic procedures for the
simulation in Section~4 of the main article.
The sampling mechanism specifies finite populations, inclusion
probabilities, and realized samples. Survey weights are exact inverse
first-order inclusion probabilities, and design effects are measured
on the realized samples.

\subsection{Data-generating mechanism and sampling design}
\label{sec:sim-details-smb}

\paragraph{Finite populations.}
In each replication, we generate a finite population from the hierarchical
logistic model of eq.~(4.1) of the main article, with
$(\beta_0, \beta_1, \beta_2) = (-0.5, 0.5, 0.3)$, a purely
within-group covariate $x^{(\mathrm{W})}$, a group-level covariate
$x^{(\mathrm{B})}_j$, and $\sigma^2_\theta$ set by a latent-scale
intraclass correlation of $0.15$ (fixed across all cells).

\paragraph{Two-stage sampling, two structural arms.}
We then draw a two-stage sample under one of two structural arms,
which differ in whether the design clusters coincide with the model
groups:
\begin{itemize}[nosep]
  \item \emph{Sampled groups} ($c = j$): stage one selects $J$ of
    $M = 200$ population groups by \emph{randomized systematic PPS},
    with size measures proportional to
    $\exp(\gamma_{\mathrm{grp}} \cdot \theta_g + \text{noise})$, so
    that group inclusion probabilities depend on the random effects;
    stage two selects units within sampled groups by
    \emph{Poisson sampling} with inclusion probabilities proportional
    to $\exp(\gamma_{\mathrm{unit}} \cdot y + \text{noise})$, a
    response-dependent selection. The design PSU is the model group.
  \item \emph{PSUs within groups} ($c \neq j$): all $J$ groups are
    observed by construction, as in the NSECE application, where every
    state appears. Within each group, we select $4$ of $12$ population
    PSUs by PPS, with size measures depending on a PSU-level outcome
    shock that the analysis model deliberately omits. We Poisson-sample
    units within selected PSUs.
    The PSU stage is stratified by group. We therefore use the
    stratified, centered, DF-corrected meat of eq.~(2.5) of the main
    article. The model-group and design-PSU targets differ, and the
    true parameter values are known.
\end{itemize}

\paragraph{Exact weights and the declared convention.}
The analysis weights are the exact inverse first-order inclusion
probabilities of the two-stage design (the product of the
stage-specific probabilities), rescaled to the declared \emph{global
unit-mean} convention of Section~2.2 of the main article.
We retain empty second-stage Poisson samples as drawn.
Redrawing until nonempty or forcing inclusion would change the
effective first-order probabilities from the declared $\pi$'s.
The design-PSU cluster universe includes all \emph{selected} PSUs.
A selected PSU with no sampled units has an estimated score total of
exactly zero and enters the stratified meat as a zero-total cluster.
It contributes to within-stratum centering and the
$C_h/(C_h - 1)$ correction.
At $\bar{n}_j = 10$ in the $c \neq j$ arm, $6$--$7\%$ of selected PSUs
are empty and essentially every replication contains at least one,
so this convention matters for the declared design-PSU target.
The replication code records the full cluster universe;
\texttt{svyder}'s \texttt{cluster\_strata} argument implements the
same convention.

We regenerate a replication only when an entire model group has no
sampled units, using a deterministically offset seed. This occurs with
probability of order $10^{-3}$ or below per replicate, only in the
smallest $\bar{n}_j = 10$ cells. The fitted model therefore always
has all $J$ groups observed.
Regeneration counts are recorded in each replication's diagnostics and
provided in the replication archive. Regeneration was needed in
2 of the 4{,}800 replications, both in $\bar{n}_j = 10$ cells.
All summaries condition on every group being observed. The excluded
event has negligible probability, and the within-replicate design
probabilities remain unchanged.

\paragraph{Measured informativeness.}
Informativeness is a design factor with three levels,
$\gamma \in \{\text{none}, \text{moderate}, \text{strong}\}$, acting
through $(\gamma_{\mathrm{grp}}, \gamma_{\mathrm{unit}})$.
These levels specify the sampling mechanism rather than target design
effects. We compute global and within-group Kish design effects on
each realized sample and report them with the results.
The realized global $\DEFF$ averages \slotnum{1.65},
\slotnum{2.0}, and \slotnum{2.5--3.9} across the three levels
(Section~4.1 of the main article).

\subsection{Factorial design}
\label{sec:sim-grid}

The study crosses
\begin{itemize}[nosep]
  \item number of groups $J \in \{20, 50\}$,
  \item expected group sample size $\bar{n}_j \in \{10, 50\}$ (the
    small-$\bar{n}_j$ cells probe the small-cluster regime discussed
    in Section~3.7 of the main article),
  \item informativeness $\gamma \in \{\text{none}, \text{moderate},
    \text{strong}\}$,
  \item structural arm $\in \{\text{sampled groups}, \text{PSUs
    within groups}\}$,
\end{itemize}
for $2 \times 2 \times 3 \times 2 = 24$ design cells, with the
intraclass correlation fixed at $0.15$ and 200 replications per cell:
\slotnum{4{,}800} fitted models in total.

\subsection{Estimation and convergence}
\label{sec:sim-estimation}

\paragraph{Stan model.}
We implement the hierarchical logistic regression in Stan using a
non-centered parameterization for the random effects
\citep{PapaspiliopoulosEtAl2007}:
\begin{align}
  \eta_j &\sim \Norm(0, 1), \notag \\
  \theta_j &= \sigma_\theta \cdot \eta_j, \label{eq:ncp}
\end{align}
where the posterior is specified over the standardized effects
$\eta_j$ rather than the raw effects $\theta_j$.
This parameterization eliminates the funnel geometry that arises when
$\sigma_\theta$ is small relative to the data information.
Survey weights enter through the pseudo-likelihood: each observation's
log-likelihood contribution is multiplied by its normalized weight in
Stan's \texttt{target}.
Priors are weakly informative: independent
$\Norm(0, 10^2)$ on the components of $\bbeta$
\citep{GelmanEtAl2008} and $\HalfCauchy(0, 2.5)$ on $\sigma_\theta$
\citep{Gelman2006}.
We compute all $\DER$s in the centered (natural) parameterization
after back-transforming the posterior draws.

\paragraph{Sampling algorithm.}
We fit each replication with the No-U-Turn Sampler via
\texttt{cmdstanr}, 4 chains in parallel, each with 1{,}000 warmup and
1{,}500 post-warmup iterations, yielding
$S = 4 \times 1{,}500 = 6{,}000$ posterior draws per replication,
with \texttt{adapt\_delta} $= 0.90$ \citep{Betancourt2017}.

\paragraph{Convergence criteria.}
We include a replication in the summaries only if it meets the following
criteria: split-$\widehat{R} \le 1.01$ on all parameters
including every $\theta_j$ \citep{GelmanRubin1992,
vehtari_rank-normalization_2021}, minimum bulk effective sample size
$100$, and zero divergent transitions.
Replications that fail these criteria remain in the archive with their
diagnostics but are excluded from the summaries. We report their count
with every table; \slotnum{99.3\%} of replications meet the criteria
(Section~4.1 of the main article).

\subsection{DER computation and the two targets}
\label{sec:sim-der-computation}

For each replication meeting the criteria, we evaluate the bread $H$
(Equation~(2.4)) and meat $J_{\mathrm{c}}$ (Equation~(2.5)) of the main
article at the posterior mean, our declared plug-in point.
The MCMC covariance $\SigMCMC$ is the sample covariance of the $S$
post-warmup draws.
We compute $\DER_k = [\Vtarget]_{kk} / [\SigMCMC]_{kk}$ for every parameter, as
in Algorithm~1 of the main article.
We compute the $\DER$ under the model-group target in both structural
arms, and additionally under the design-PSU target (stratified by
group, with the centered DF-corrected meat) in the $c \neq j$ arm,
where the two targets differ.
We compare interval strategies (naive, selective at $c_0 = 1.2$,
blanket) at nominal 90\%, as described in Section~4.1 of the main article.
Selective correction transforms the flagged block jointly; blanket
correction rescales each parameter marginally.
We use the \texttt{svyder} package for all computations
(\Cref{app:svyder}) with the declared target supplied explicitly.

\subsection{Full result tables}
\label{sec:sim-full-tables}

The complete per-cell result tables are provided in the replication
repository (one CSV per summary, regenerated by a single aggregation
script).
Table~S-E.2 reports the threshold sweep's per-replication false-positive
rates. Section~4.2 of the main article also reports the stability of
decisions on \emph{protected} parameters across $c_0 \in [1.1, 1.6]$.

%% file: sm_d.tex

\section{NSECE Application: Supplementary Materials}
\label{app:nsece-supplement}

We provide the data and analysis details for the NSECE 2019 application
in Section~5 of the main article, including the full diagnostic table,
MCMC convergence results, and threshold sensitivity.

\subsection{Data description}
\label{app:nsece-data}

The 2019 National Survey of Early Care and Education
\citep[NSECE;][]{NSECE2022} is a
nationally representative study of childcare providers conducted by
the National Opinion Research Center (NORC) under contract with the
Office of Planning, Research, and Evaluation (OPRE), Administration
for Children and Families, U.S.\ Department of Health and Human
Services.
The survey employs a stratified multistage cluster design:
30 strata are formed from geographic and demographic characteristics,
within which 415 primary sampling units (PSUs) are selected.
Providers are sampled within PSUs, yielding a total of $N = 6{,}785$
center-based providers distributed across $J = 51$ states (including
the District of Columbia).
The data extract and model specification follow \citet{Lee2026}.

\paragraph{Outcome variable.}
The binary outcome $y_i \in \{0, 1\}$ indicates whether provider~$i$
participates in the state Infant--Toddler (IT) quality improvement
system, with an observed (weighted) prevalence of 64.7\%.

\paragraph{Covariates.}
Two covariates enter the hierarchical logistic regression model of
eq.~(5.1) of the main article:
\begin{itemize}[nosep]
  \item \emph{Poverty rate (CWC):} A continuous measure of the
    poverty rate in the provider's service area, centered within
    context \citep[CWC; group-mean centering within
    states;][]{enders2007centering}.  The CWC transformation removes
    between-state variation from this covariate, creating a purely
    within-state predictor (Tier~I-a). It is flagged under every
    target--convention combination examined in Section~5 of the main article.
  \item \emph{Tiered reimbursement (state-level):} A binary indicator
    for whether the provider's state operates a tiered reimbursement
    policy.  As a state-level variable, it varies only between
    groups and is fully identified from between-state comparisons
    (Tier~I-b); it is design-protected under every combination
    examined.
\end{itemize}

\paragraph{Survey weights and the declared convention.}
Base survey weights reflect differential selection probabilities and
nonresponse adjustments.
The coefficient of variation of the base weights is
$\mathrm{CV}_w = 1.66$, giving a raw-weight Kish design effect of
$\DEFF = 1 + \mathrm{CV}_w^2 = 3.76$.
For the pseudo-posterior analysis the weights enter under the declared
\emph{global unit-mean} convention of Section~2.2 of the main article,
$\tilde{w}_{ij} = N w_{ij} / \sum_{i'j'} w_{i'j'}$, which preserves
between-state weight variation and fixes the effective sample size of
the pseudo-likelihood at $N$.
The main alternative rescales within each state so that
$\sum_i \tilde{w}_{ij} = n_j$ \citep{PfeffermannEtAl1998b}, removing
between-state weight variation from the fit.
Section~5.3 of the main article examines this sensitivity: rescaling
changes the group-level results and can reverse the ordering of the
two targets for the flagged coefficient.
The weight convention is therefore part of the declared target.

\paragraph{The two declared targets.}
The model groups (states) differ from the design PSUs, and every state
is in the sample by construction. We report both aggregation units
defined in Section~2.3 of the main article because they answer
different questions.
The \emph{design-PSU target} aggregates scores over the 415 PSUs
within 30 strata using centered, DF-corrected meat. It measures the
design's sampling variability and is our primary analysis.
The \emph{model-group target} aggregates over the 51 states and
supports interpretation using the closed-form results in Section~3 of
the main article.
Section~5.2 of the main article interprets the difference between the
targets, particularly for the state random effects.

\paragraph{Group sizes.}
State sample sizes $n_j$ range from 17 (smallest state/territory) to
1{,}110 (largest state), with a median of 64 and mean of 133.
This variation induces corresponding variation in the group-specific
shrinkage factors $B_j$; per-state values under the global convention
are reported in the table of Section~S-E.3, and the within-state fit
is summarized in Table~4 of the main article.

\subsection{Analysis settings}
\label{app:nsece-declared}

We use the settings in Section~5.1 of the main article: global unit-mean
weight scaling (raw-weight Kish $\DEFF = 3.76$), with within-state
rescaling for the sensitivity analysis.
We compute both the design-PSU target (415 PSUs, 30 strata, centered
DF-corrected meat) and the model-group target (51 states).
The plug-in point is the posterior mean,
$\hat{\sigma}_\theta = 0.650$.
We flag parameters at $c_0 = 1.2$ and examine threshold sensitivity
in \Cref{app:tau-sensitivity}.
The fit meets the convergence criteria: $\widehat{R} \le 1.004$ on
all parameters including every $\theta_j$, bulk ESS $\ge 1{,}499$,
and zero divergent transitions.

\subsection{Full dual-target DER table}
\label{app:full-der-table}

The full 54-parameter table reports the $\DER$ for all three fixed
effects and all 51 state random effects under both declared targets
(design-PSU and model-group), using global unit-mean weights.
We interpret the parameters using the tier assignments in the main
article and flag them at $c_0 = 1.2$.
Section~5.2 of the main article reports the main findings: the
within-state poverty coefficient is flagged under both targets
($\DER = 3.55$ design-PSU, $4.59$ model-group, global convention),
the between-state parameters have DERs far below one, and the state
random effects distinguish the targets (\slotnum{27} of 51 exceed $c_0$
under the design-PSU target versus none under the model-group
target).

\input{tables/tabSE_full_der}

\paragraph{Per-state estimator noise under the model-group target.}
Remark~2 of the main article explains the estimator noise in per-state
$\DER$ values under the model-group target.
The $\theta_j$ meat entry uses one centered score-total contrast per
state. It is unbiased for its variance contribution but effectively
uses one observation.
Each estimated per-state $\DER$ therefore has substantial sampling
noise, even when the population formulas in Section~3 of the main
article are exact.
We interpret the model-group column by tier: all states are well
within Tier~II, but these estimates do not support ranking states.
The design-PSU target, which aggregates 415 clusters within 30 strata,
is estimated more precisely per entry.
The wide Tier~II margin makes classification stable despite noisy
individual estimates (Section~4 of the main article).

\subsection{MCMC convergence diagnostics}
\label{app:mcmc-diagnostics}

The fitted model meets the convergence criteria reported in
Section~5.1 of the main article: split-$\widehat{R} \le 1.004$
\citep{GelmanRubin1992, vehtari_rank-normalization_2021} on all 54
parameters plus the hyperparameter $\sigma_\theta$, bulk effective
sample sizes of at least 1{,}499, and zero divergent transitions or
treedepth exceedances, under the non-centered parameterization
\citep{PapaspiliopoulosEtAl2007} with \texttt{adapt\_delta} $= 0.95$
\citep{Betancourt2017}.
Parameter-level diagnostics are provided in the replication repository.

\subsection{Threshold sensitivity}
\label{app:tau-sensitivity}

The threshold sensitivity analysis in Section~5.4 of the main article
sweeps $c_0 \in [0.8, 2.0]$ and tracks the flagged set under both
targets.
The within-state poverty coefficient is flagged at every threshold
value under both targets, and no between-state fixed effect is flagged
at any value.
Under the design-PSU target, the random-effect flag count declines
smoothly with $c_0$.
Total flags are \slotnum{34, 31, 28, 27, 23} at
$c_0 = 0.8, 1.0, 1.2, 1.5, 2.0$, of which
\slotnum{33, 30, 27, 26, 22} are state effects.
Under the model-group target, the totals are \slotnum{2, 1, 1, 1, 1}.
The qualitative contrast between targets therefore persists beyond
the default $c_0 = 1.2$.
\input{tables/tab5_c0_sweep}

\subsection{DER versus group size}
\label{app:der-vs-nj}

A plot of the 51 state random-effect $\DER$ values against group size
$n_j$ under the model-group target shows the joint influence of the
shrinkage factor $B_j$ and state-level design effect $\DEFF_j$.
The theory predicts low values (Tier~II under this target) and a
non-monotone relationship with shrinkage (the inverted U of
Corollary~\ref{cor:boundary}(iv)).
Scatter reflects differences between states in weight distributions
and within-state sampling structure, subject to the per-state
estimator noise described in \Cref{app:full-der-table}.
The plot is provided in the replication repository.

%% file: tables/tabSE_full_der.tex
\begin{longtable}{@{}lrrrcrc@{}}
\caption{Full 54-parameter $\DER$ table for the NSECE application under both declared targets (global unit-mean convention; flags mark $\DER > c_0 = 1.2$).}\label{tab:full_der}\\
\toprule
Parameter & $n_j$ & $B_j$ & \multicolumn{2}{c}{Model-group} & \multicolumn{2}{c}{Design-PSU} \\
 & & & $\DER$ & flag & $\DER$ & flag \\
\midrule\endfirsthead
\toprule
Parameter & $n_j$ & $B_j$ & \multicolumn{2}{c}{Model-group} & \multicolumn{2}{c}{Design-PSU} \\
 & & & $\DER$ & flag & $\DER$ & flag \\
\midrule\endhead
$\beta_0$ (intercept) & --- & --- & 0.038 &  & 0.270 &  \\
$\beta_1$ (poverty, within-state) & --- & --- & 4.590 & $\bullet$ & 3.549 & $\bullet$ \\
$\beta_2$ (tiered reimbursement) & --- & --- & 0.039 &  & 0.280 &  \\
$\theta_{1}$ & 86 & 0.931 & 0.051 &  & 0.455 &  \\
$\theta_{2}$ & 46 & 0.756 & 0.046 &  & 0.257 &  \\
$\theta_{3}$ & 183 & 0.920 & 0.008 &  & 0.787 &  \\
$\theta_{4}$ & 56 & 0.931 & 0.013 &  & 2.902 & $\bullet$ \\
$\theta_{5}$ & 1110 & 0.983 & 0.035 &  & 0.548 &  \\
$\theta_{6}$ & 110 & 0.918 & 0.086 &  & 5.359 & $\bullet$ \\
$\theta_{7}$ & 58 & 0.859 & 0.117 &  & 0.967 &  \\
$\theta_{8}$ & 38 & 0.676 & 0.129 &  & 0.189 &  \\
$\theta_{9}$ & 40 & 0.622 & 0.523 &  & 0.408 &  \\
$\theta_{10}$ & 362 & 0.959 & 0.095 &  & 2.765 & $\bullet$ \\
$\theta_{11}$ & 192 & 0.938 & 0.058 &  & 1.584 & $\bullet$ \\
$\theta_{12}$ & 64 & 0.794 & 0.162 &  & 0.280 &  \\
$\theta_{13}$ & 27 & 0.744 & 0.228 &  & 0.264 &  \\
$\theta_{14}$ & 471 & 0.963 & 0.034 &  & 1.667 & $\bullet$ \\
$\theta_{15}$ & 75 & 0.911 & 0.024 &  & 7.229 & $\bullet$ \\
$\theta_{16}$ & 27 & 0.917 & 0.039 &  & 14.842 & $\bullet$ \\
$\theta_{17}$ & 51 & 0.923 & 0.044 &  & 2.975 & $\bullet$ \\
$\theta_{18}$ & 59 & 0.932 & 0.027 &  & 2.035 & $\bullet$ \\
$\theta_{19}$ & 77 & 0.942 & 0.027 &  & 0.647 &  \\
$\theta_{20}$ & 18 & 0.858 & 0.135 &  & 1.089 &  \\
$\theta_{21}$ & 178 & 0.950 & 0.028 &  & 2.545 & $\bullet$ \\
$\theta_{22}$ & 177 & 0.898 & 0.111 &  & 0.701 &  \\
$\theta_{23}$ & 117 & 0.927 & 0.088 &  & 5.085 & $\bullet$ \\
$\theta_{24}$ & 103 & 0.946 & 0.037 &  & 5.723 & $\bullet$ \\
$\theta_{25}$ & 17 & 0.861 & 0.004 &  & 0.362 &  \\
$\theta_{26}$ & 102 & 0.932 & 0.011 &  & 1.912 & $\bullet$ \\
$\theta_{27}$ & 22 & 0.789 & 0.035 &  & 9.830 & $\bullet$ \\
$\theta_{28}$ & 57 & 0.891 & 0.064 &  & 4.188 & $\bullet$ \\
$\theta_{29}$ & 43 & 0.655 & 0.129 &  & 0.223 &  \\
$\theta_{30}$ & 52 & 0.700 & 0.001 &  & 0.378 &  \\
$\theta_{31}$ & 201 & 0.948 & 0.017 &  & 1.527 & $\bullet$ \\
$\theta_{32}$ & 34 & 0.864 & 0.013 &  & 1.255 & $\bullet$ \\
$\theta_{33}$ & 558 & 0.972 & 0.024 &  & 3.433 & $\bullet$ \\
$\theta_{34}$ & 121 & 0.945 & 0.028 &  & 2.283 & $\bullet$ \\
$\theta_{35}$ & 50 & 0.684 & 0.417 &  & 0.842 &  \\
$\theta_{36}$ & 186 & 0.967 & 0.032 &  & 6.440 & $\bullet$ \\
$\theta_{37}$ & 61 & 0.942 & 0.011 &  & 7.018 & $\bullet$ \\
$\theta_{38}$ & 74 & 0.893 & 0.079 &  & 0.629 &  \\
$\theta_{39}$ & 141 & 0.952 & 0.014 &  & 2.456 & $\bullet$ \\
$\theta_{40}$ & 57 & 0.732 & 0.022 &  & 0.229 &  \\
$\theta_{41}$ & 50 & 0.906 & 0.040 &  & 3.404 & $\bullet$ \\
$\theta_{42}$ & 40 & 0.629 & 0.507 &  & 0.540 &  \\
$\theta_{43}$ & 98 & 0.915 & 0.012 &  & 2.085 & $\bullet$ \\
$\theta_{44}$ & 578 & 0.978 & 0.031 &  & 2.158 & $\bullet$ \\
$\theta_{45}$ & 24 & 0.783 & 0.866 &  & 1.150 &  \\
$\theta_{46}$ & 26 & 0.772 & 0.373 &  & 1.129 &  \\
$\theta_{47}$ & 120 & 0.940 & 0.038 &  & 3.044 & $\bullet$ \\
$\theta_{48}$ & 131 & 0.900 & 0.067 &  & 0.995 &  \\
$\theta_{49}$ & 36 & 0.795 & 0.003 &  & 0.748 &  \\
$\theta_{50}$ & 138 & 0.947 & 0.063 &  & 5.661 & $\bullet$ \\
$\theta_{51}$ & 43 & 0.691 & 0.073 &  & 0.153 &  \\
\bottomrule
\end{longtable}

%% file: tables/tab5_c0_sweep.tex
\begin{table}[H]
\centering
\caption{Threshold sensitivity. NSECE: number of parameters (of 54) flagged at each threshold $c_0$ under both declared targets. Simulation: per-replication false-positive rate at the same thresholds under the model-group target, pooled over all 24 cells. This is the share of replications meeting the convergence criteria in which any non-target parameter (between-group fixed effect or random effect) exceeds $c_0$. The recommended default $c_0 = 1.2$ is shown in bold.}
\label{tab:c0_sensitivity}
\begin{tabular}{@{}cccc@{}}
\toprule
 & \multicolumn{2}{c}{NSECE flagged (of 54)} & Simulation \\
\cmidrule(lr){2-3}
$c_0$ & Model-group & Design-PSU & False-positive rate (\%) \\
\midrule
0.8 & 2 & 34 & 16.0 \\
1.0 & 1 & 31 & 5.9 \\
\textbf{1.2} & \textbf{1} & \textbf{28} & \textbf{2.1} \\
1.5 & 1 & 27 & 0.3 \\
2.0 & 1 & 23 & 0.0 \\
\bottomrule
\end{tabular}

\vspace{4pt}
\begin{minipage}{\textwidth}
\footnotesize
\textit{Notes.} NSECE counts are computed from the full $\DER$ vectors of the fitted model. The poverty coefficient $\beta_1$ is flagged at every threshold shown under both targets. Simulation false positives are based on 4{,}767 replications meeting the convergence criteria. A false positive costs a wider interval for a protected parameter, not invalid inference.
\end{minipage}
\end{table}

%% file: sm_e.tex

\section{Using the \texttt{svyder} R Package}
\label{app:svyder}

The compute-classify-correct workflow specified by Algorithm~1 in the
main article is implemented in the open-source R package \texttt{svyder}
(``\textbf{s}ur\textbf{v}e\textbf{y} \textbf{d}esign \textbf{e}ffect
\textbf{r}atios''), version~0.2.0.
The package is available on GitHub at
\url{https://github.com/joonho112/svyder}, with full documentation,
vignettes, and worked examples hosted at
\url{https://joonho112.github.io/svyder/}.
It is released under the MIT license.
Users specify the variance target through arguments for the
aggregation unit, design strata, and weight convention.

\texttt{svyder} accepts posterior draws from any Bayesian sampler that
produces an $S \times d$ draws matrix.
S3 methods support \texttt{brms}, \texttt{cmdstanr}, and
\texttt{rstanarm} model objects.
The required model information includes the response vector $\by$,
design matrix $\mathbf{X}$, group indicators, and family.
Some methods extract this information automatically; others require
explicit inputs, as described in \Cref{app:software-design}.
For other samplers, users supply a numeric draws matrix and this model
information together with the survey design variables.
The design can be supplied as weights and cluster identifiers or as a
\texttt{survey.design2} object, from which first-stage clusters and
strata are extracted.
Computing the diagnostic, classifying parameters, and correcting the
selected block take a fraction of a second for the 54-parameter NSECE
model. Runtime scales linearly with the number of posterior draws.

\subsection{Specifying the variance target}
\label{app:software-target}

\begin{samepage}
Use \texttt{der\_compute()} to specify the variance target defined in
Section~2.3 of the main article:

\begin{verbatim}
der_compute(draws, y, X, group, weights,
            cluster,       # PSU or model-group identifiers (required)
            strata = NULL, # design strata
            cluster_strata = NULL,
            # strata for all selected clusters, including
            # those with no sampled observations
            normalize = c("unit_mean",    # global unit-mean scaling
                          "group_size",  # within-group rescaling
                          "none"),
            family, sigma_theta, ...)
\end{verbatim}
\end{samepage}

\noindent
The required \texttt{cluster} argument identifies the aggregation
unit. Set it to the design PSU identifiers for the design-PSU target,
or to the group identifiers for the model-group target.
The \texttt{strata} argument specifies the strata used to center
cluster score totals and apply the finite-cluster correction in
Equation~(2.5) of the main article.
If selected clusters can contain no sampled observations, provide
\texttt{cluster\_strata} for all selected clusters, indexed by their
identifiers. Empty clusters then contribute zero score totals
(Section~S-D); their count is returned as \texttt{cluster\_n\_empty}.
The \texttt{normalize} argument sets the weight convention:
\texttt{"unit\_mean"} uses the global convention in Section~2.2 of
the main article, and \texttt{"group\_size"} uses the within-group
alternative examined in Section~5.3 of the main article.
The argument \texttt{psu} of earlier releases is retained in 0.2.0 as
a deprecated alias for \texttt{cluster} and warns on use.

\subsection{Minimal workflow}
\label{app:software-workflow}

The following example uses the bundled demonstration data.
The function \texttt{der\_diagnose()} runs all three steps in one call.

\begin{verbatim}
## Install from GitHub
# install.packages("pak")
pak::pak("joonho112/svyder")

library(svyder)
data(nsece_demo)

## Step 1: Compute DER values under the declared target
result <- der_compute(
  nsece_demo$draws,
  y = nsece_demo$y, X = nsece_demo$X,
  group = nsece_demo$group,
  weights = nsece_demo$weights,
  cluster = nsece_demo$psu, # declared aggregation unit (design PSUs)
  normalize = "unit_mean",      # declared weight convention
  family = "binomial",
  sigma_theta = nsece_demo$sigma_theta,
  param_types = nsece_demo$param_types
)

## Step 2: Classify parameters (threshold argument `tau` is the
##         threshold c_0 in the main article)
result <- der_classify(result, tau = 1.2)

\end{verbatim}
\begin{samepage}
\begin{verbatim}

## Step 3: Correct the flagged block (joint block-Cholesky)
result <- der_correct(result)

## Summarize and extract corrected draws
summary(result)
corrected_draws <- as.matrix(result)
\end{verbatim}
\end{samepage}

\noindent
In the package interface the flagging threshold is the argument
\texttt{tau}; it corresponds to the threshold denoted $c_0$ in the
main article (the symbol $\tau$ is reserved there for the prior
precision $1/\sigma^2_\theta$).

\subsection{Dual targets: \texttt{der\_compare()}}
\label{app:software-compare}

When design PSUs and model groups differ, we recommend reporting both
targets (Sections~2.3 and~5.2 of the main article).
The function \texttt{der\_compare()} computes the $\DER$ under a named
list of aggregation units and reports the results side by side.
For a fixed fit and weight convention, only the meat changes; the
bread and MCMC covariance remain the same:

\begin{verbatim}
comp <- der_compare(result, clusters = list(
  design_psu  = nsece_demo$psu,
  model_group = nsece_demo$group
))
head(comp)   # columns: param, cluster_name, der
\end{verbatim}

\noindent
We use this comparison for the dual-target tables in Section~5 of the
main article and \Cref{app:full-der-table}.

\subsection{Selective correction: block-Cholesky \texttt{der\_correct()}}
\label{app:software-correct}

\texttt{der\_correct()} implements Step~3 of Algorithm~1 of the main article.
Its default method, \texttt{"block\_cholesky"}, extracts the flagged
block $\mathcal{F}$, forms the Cholesky factors of
$\SigMCMC[\mathcal{F},\mathcal{F}]$ and
$\Vtarget[\mathcal{F},\mathcal{F}]$, and applies the joint
transformation to the flagged coordinates of every draw.
The corrected block matches the target variances and covariances on
$\mathcal{F}$. Unflagged parameter draws remain unchanged.
The cost is $O(|\mathcal{F}|^3)$ rather than $O(d^3)$, and the
transformation reduces to a scalar rescaling when
$|\mathcal{F}| = 1$ (the NSECE model-group case).
A \texttt{"marginal"} method (per-parameter rescaling only) is
provided for comparison.

For the threshold sweep reported in the simulation and application
supplements, the call
\begin{center}
\texttt{der\_sensitivity(result, tau\_range = seq(0.8, 2.0, by = 0.1))}
\end{center}
tabulates the flagged set over a grid of $c_0$ values.

\subsection{Tier III exclusion}
\label{app:software-tier3}

Hyperparameters are excluded from the diagnostic by construction for
the reasons summarized in Section~3.3 of the main article and
elaborated in \Cref{sec:remark1-elaboration}.
The package assigns no $\DER$ to hyperparameters. It lists them in
the object's \texttt{excluded} component, and
\texttt{der\_classify()} prints ``Tier III excluded (DER undefined).''
\Cref{tab:tier3-exclusion} summarizes the treatment of each parameter
class.

\begin{table}[H]
  \centering
  \caption{Parameter classes and their treatment in \texttt{svyder}
    0.2.0.  Tiers follow Table~1 of the main article.}
  \label{tab:tier3-exclusion}
  \small
  \begin{tabularx}{\textwidth}{@{}llX@{}}
    \toprule
    \textbf{Tier} & \textbf{Parameters} & \textbf{Package treatment} \\
    \midrule
    I-a & Within-group fixed effects
      & $\DER$ reported; flagged when $\DER > c_0$; corrected by the
        block-Cholesky step. \\
    I-b & Between-group fixed effects
      & $\DER$ reported; typically far below $c_0$ (shielded);
        monitored. \\
    II & Random effects $\theta_j$
      & Per-group $\DER$ reported under the declared target, with the
        estimator-noise caveat in Remark~2 of the main article
        under the model-group target. \\
    III & Hyperparameters ($\sigma_\theta$, \ldots)
      & Excluded by construction: no $\DER$ assigned; listed in the
        \texttt{excluded} component with the reason
        (``DER undefined''); pointer to profile and multi-level
        methods in Section~3.3 of the main article. \\
    \bottomrule
  \end{tabularx}
\end{table}

\subsection{Inputs, results, and additional functions}
\label{app:software-design}

The three core functions, \texttt{der\_compute()},
\texttt{der\_classify()}, and \texttt{der\_correct()}, mirror the
three phases of Algorithm~1 of the main article.
Each returns an S3 object of class \texttt{svyder}.
After computation, it contains $\DER$ values and sandwich matrices.
Classification adds tier assignments and flags; correction adds the
rescaled posterior draws.

\paragraph{Model inputs.}
The inputs required by \texttt{der\_compute()} depend on the model
object supplied as its first argument.
For \texttt{brmsfit} and \texttt{stanreg} objects, the method
automatically extracts $\by$, $\mathbf{X}$, group indicators, the
model family, and $\sigma_\theta$, so that users need only supply the
survey design variables.
For \texttt{CmdStanMCMC} and raw matrix inputs, all quantities are
passed explicitly.
A \texttt{survey.design2} object may be supplied through the
\texttt{design} argument, from which first-stage clusters and strata
are extracted.

\paragraph{Decomposition and closed-form comparisons.}
The function \texttt{der\_decompose()} breaks each parameter's $\DER$
into its constituent factors (the Kish design effect $\DEFF$,
shrinkage factor $B_j$, protection ratio $R_k$, and finite-$J$
correction $\kappa$), as described in Section~3.2 of the main article.
The companion function \texttt{der\_theorem\_check()} compares the
closed-form predictions of Theorems~1 and~2 of the main article against the
numerically computed $\DER$ values, as in the numerical comparison
in \Cref{app:unbalanced}.

\paragraph{Visualization.}
The \texttt{plot()} method for \texttt{svyder} objects produces
diagnostic figures with tier-based coloring: a \emph{profile} plot
showing $\DER$ values across all parameters with the threshold
marked, a \emph{decomposition} plot displaying the factor breakdown,
a \emph{comparison} plot for \texttt{der\_compare()} output, and a
\emph{correction} plot comparing original and corrected posterior
intervals for flagged parameters.